\RequirePackage{fix-cm}
\documentclass[smallextended]{svjour3}       % onecolumn (second format)
\smartqed  % flush right qed marks, e.g. at end of proof

\usepackage[titletoc,toc,title]{appendix}
\usepackage[usenames,dvipsnames]{color}
\usepackage[usenames,dvipsnames]{xcolor}
\usepackage{alltt}
\usepackage{multicol}
\usepackage{geometry}
\usepackage{adjustbox}
\usepackage{float}
\usepackage{multirow}
\makeatletter
\newcommand{\xRightarrow}[2][]{\ext@arrow 0359\Rightarrowfill@{#1}{#2}}
\makeatother
\usepackage{amsfonts}
\usepackage{amssymb}
\usepackage{pgfplots}
\usepackage{tikz}
\usetikzlibrary{arrows,automata}
\usetikzlibrary{arrows}
\usepackage{pgfplotstable}
\usepackage{enumerate}
\usepackage{stmaryrd}
\usepackage{epstopdf}
\usepackage{graphicx}
\usepackage{subfig}
\usepackage{url}
\usepackage {tikz}
\usetikzlibrary{patterns}
\usetikzlibrary {positioning}
\definecolor {processblue}{cmyk}{0.96,0,0,0}
\usepackage{calc}
\usepackage{ifthen}
\usepackage{listings}
\usepackage{mathtools}
\usepackage{syntax}
\usepackage[normalem]{ulem}
\usepackage[normalem]{ulem}
\usepackage{algorithm}% http://ctan.org/pkg/algorithm
\usepackage[noend]{algpseudocode}% http://ctan.org/pkg/algorithmicx
\usepackage{caption}% http://ctan.org/pkg/caption

\lstdefinelanguage{palang}{
    morekeywords={actor, int, if, else, self, main},
    otherkeywords={=>,<-,<\%,<:,>:,\#,@},
    sensitive=true,
    morecomment=[l]{//},
    morecomment=[n]{/*}{*/},
    morestring=[b]",
    morestring=[b]',
    morestring=[b]"""
}
\newsavebox{\mylistingbox}

\def\ID{\mathit{ID}}
\def\Var{\mathit{Var}}
\def\Val{\mathit{Val}}
\def\Env{\mathit{Env}}
\def\MName{\mathit{MName}}

\def\Mtd{\mathit{Mtd}}
\def\Actor{\mathit{Actor}}
\def\Statement{\mathit{Stat}}

\def\Expr{\mathit{Expr}}

\def\self{\mathit{self}}
\def\body{\mathit{body}}

\def\eval{\mathit{eval}}
\def\expr{\mathit{expr}}

\def\svars{\mathit{svars}}

\def\Snd{\mathit{Snd}}

\def\Rcv{\mathit{Rcv}}
\def\IntM{\mathit{InfM}}
\def\Act{\mathit{Act}}
\def\Info{\mathit{Info}}

\def\Observer{\mathit{observer}}
\def\Transmitter{\mathit{transmitter}}
\def\Quadricopter{\mathit{quadricopter}}
\def\Feedback{\mathit{feedback}}
\def\Controller{\mathit{controller}}

\def\lmeta{\langle\!\langle}
\def\rmeta{\rangle\!\rangle}
\newcommand{\RN}[1]{%
    \textup{\uppercase\expandafter{\romannumeral#1}}%
}

\newcommand{\overto}[1]{\stackrel{#1}{%
        \overrightarrow{\smash{\,\scriptsize{\phantom{#1}}\,}}}}
\newcommand{\Rule}[2]{                                  % operational rule
    \frac{\raisebox{.7ex}{\normalsize{$#1$}}}
    {\raisebox{-1.0ex}{\normalsize{$#2$}}}}

\newcommand{\FAxiom}[1]{                                  % operational rule
    {\normalsize {#1}}
}

\journalname{Acta Informatica}
\begin{document}

\title{Verification of Asynchronous Systems with an Unspecified Component}
\author{Rosa Abbasi \and Fatemeh Ghassemi \and Ramtin Khosravi}
\authorrunning{R. Abbasi, F. Ghassemi, R. Khosravi}
\institute{R. Abbasi \at
    University of Tehran, Iran \\
\email{ro.abbasi@ut.ac.ir}
\and
    F. Ghassemi \at
              University of Tehran, Iran \\
              %Tel.: +98-21-82084995\\
%              Fax: +123-45-678910\\
              \email{fghassemi@ut.ac.ir}           %  \\
%             \emph{Present address:} of F. Author  %  if needed
           \and
           R. Khosravi \at
              University of Tehran, Iran \\
\email{r.khosravi@ut.ac.ir}}

%\institute{School of Electrical and Computer Engineering, College of Engineering,\\
%   University of Tehran, Tehran, Iran}

\maketitle
\begin{abstract}
Component-based systems evolve as a new component is added or an existing one is replaced by a newer version. Hence, it is appealing to assure the new system still preserves its safety properties. However, instead of inspecting the new system as a whole, which may result in a large state space, it is  beneficial to reuse the verification results by inspecting the newly added component in isolation.
%In other words, should we generate the monolithic transition system for the whole system when a new component is added or an existing one is replaced by a newer version, or we can find a way to verify the newly added or  modified component?
To this aim, we study the problem of model checking component-based asynchronously communicating systems in the presence of an unspecified component against safety properties. Our solution is based on assume-guarantee reasoning, adopted for asynchronous environments, which generates the weakest assumption. If the newly added component
conforms to the assumption, then the whole system still satisfies the property. To make the approach efficient and convergent, we produce an overapproximated interface of the missing component and by its composition with the rest of the system components, we achieve an overapproximated specification of the system, from which we remove those traces of the system that violate the property and generate an assumption for the missing component.

We have implemented our approach on two  case studies. Furthermore, we compared our results with the state of the art direct approach. Our resulting assumptions are smaller in size and achieved faster.
\end{abstract}

\section{Introduction}\label{Sec::Intro}
%\subsection{Motivation}
In the basic use of model checking to verify a system, it is usually assumed that all of the components of the system are specified. The system specification is transformed into a monolithic transition system against which temporal logic properties are verified \cite{baier2008principles}.
% and by applying model checking all the possible executions of the system are explored.
In the domain of Component-Based Software Engineering (CBSE), the system is composed of several reusable components, usually developed by different vendors. The notion of component is generally assumed to be substitutable \cite{Szyperski:2002:CSB:515228}. When applying model checking to a component-based system, the question is if we can reuse the verification results, the same way we reuse software components. In other words, should we generate the monolithic transition system for the whole system when a new component is added or an existing one is replaced by a newer version, or can we find a way to verify the newly added or  modified component?

This question is not limited to the classical use of CBSE, but is also valid
 in the context of Service Oriented Architecture (SOA) \cite{wan2006service,hashimi2003service}, in which services can be viewed as asynchronously communicating components. Changes to the components in such a system  may be caused by migrating legacy systems into new services. Since this kind of migration usually takes place when the system is operational, the problem of verification of the new system can be very important
  \cite{almonaies2010legacy,canfora2008wrapping,parveen2008research,lewis2005service}.
  In such cases, the specification of the legacy system may not be available.
 This is the case for the black-box modernization approaches  towards migration such as \cite{canfora2008wrapping} . Although model checking of SOA-based systems has been studied in several works
 \cite{foster2007ws,zheng2007model,dai2007framework,gao2006model},
 the problem of analyzing systems with unspecified components is not considered.

In a component-based system, where the components are developed by various vendors, it is usually common for the components to be integrated using asynchronous message passing. Since the coupling caused by such an integration mechanism is more loose than traditional remote procedure call approach, it allows more heterogeneity among the components, resulting in a higher degree of reuse, higher availability and scalability in the system \cite{linthicum2000enterprise}. Moreover, asynchronous communication is the common interaction mechanism used in distributed systems gaining more attention as the trend towards Cyber-Physical Systems is rapidly growing. Therefore we have chosen asynchronously communicating environments as the context of our work.

%Actor-based modeling with objects that communicate asynchronously, has been established as a suitable way of representing concurrent distributed systems and has become very popular in practice \cite{hewitt2008orgs,hewitt2009actorscript,karmani2009actor,Scala}.
%Therefore we introduce NPalang a toy actor-based language for modeling system components.
%NPalang is explained in details in Section \ref{sec:background}, but first the problem motivating our research
%is explicitly defined.

\begin{figure}[t]
    \centering
    \includegraphics[width=0.3\textwidth]{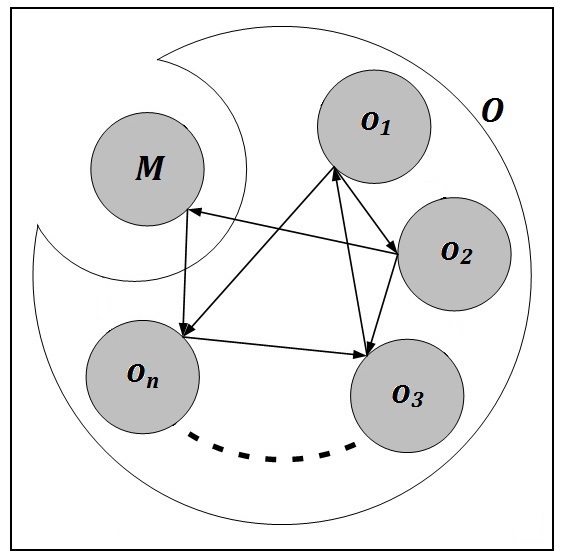}
    \caption{An asynchronous system with an unspecified component $M$}
    \label{fig:echoSys}
\end{figure}

%\subsection{Problem Definition}\label{Sec::Intro}
In this paper, we study the problem of model checking component-based systems  communicating asynchronously on bounded channels in the presence of an unspecified component against safety properties.
Consider the system  of Figure \ref{fig:echoSys} where the component $M$ operates as  a service and  its specification is  unavailable or subject to modifications. The rest of the components, which form an open system $O$, are  specified. Some basic information about the open system interactions with its environment in terms of the exchanged messages is presumed (described more formally in Section \ref{sec:OpenSystem}). This information is provided by the system modeler or is extracted from the documents of the components in $O$. The system is required to satisfy a certain safety property $P$. Our problem can be reduced to an assumption generation problem, applied in the context of assume-guarantee reasoning \cite{jones1983specification,pnueli1985transition}: We would like to find an assumption $A$ such that if $M$ satisfies $A$, then the system will satisfy $P$. Hence instead of verifying the system composed of $M$ and components of $O$ against $P$, we verify $M$ against $A$.
%Producing an assumption such that if it is satisfied by the modified, added or unspecified component, then the system meets its specified properties or requirements .

Components encapsulate their data and functions and are loosely implemented without any pre-assumption about their environment. Hence,  the formal behavior of an open component is too general to cope with any environment. Conversely, its closed behavior with a specific environment is reduced substantially due to the specific exertion of the component behavior by the environment. This generality, intrinsic in the behavior of the component, makes existing approaches of assume-guarantee reasoning inapplicable to the asynchronous environments. % due to the nature of asynchrony. The behavior of an open component is too general to cope with any environment,while its behavior in a specific environment is reduced substantially. 
Existing approaches follow two main streams: $1$) An initial interface is generated and refined based on counterexamples iteratively, $2$) The parallel composition of a characteristic model of $P$ and $O$, usually both specified by LTSs, is generated %(by synchronizing on their common behaviors), 
and then those traces that violate $P$ are eliminated. The first set of approaches are not time-efficient~\cite{cobleigh2008breaking}, may generate a very large assumption and may fail to verify $M$ against $A$ in some cases
\cite{nam2006learning,alur2005symbolic}. Besides, these approaches iteratively use the specification of all components (including $M$) for assumption generation, concluding that for each modified component, the assumption should be generated from scratch and the verification results cannot be reused.
The second set of approaches, so called \emph{direct}, are based on the behavior of the open system $O$, so they may result in a huge state space or may not converge in bounded time and memory in the asynchronous settings. To tackle the problem, we take advantage of both courses: we adopt the direct approach of \cite{giann}, and readjust it for the asynchronous settings in two respects: $1$) We consider an auxiliary component, which acts as an interface of $M$, to close the open system; $2$) Furthermore, as the interface is still too general, we consider the effect of $P$ at the level of component.

%\subsection{Outline of the Proposed Method}
To provide an asynchronous model for the system, we use the actor model \cite{agha1985actors,agha1990structure,agha1997foundation} to specify components in the system. Such modeling has become very popular in practice \cite{hewitt2008orgs,hewitt2009actorscript,karmani2009actor}\footnote{Scala programming language supports actor-models \url{http://www.scala-lang.org}}. To simplify the introduction of our method, we use a simple actor-based language, called Actor Modeling Language (AML), for describing component models. AML is explained in details in Section \ref{sec:background}. Note that our method can be easily extended to more sophisticated actor-based modeling languages, such as Rebeca \cite{sirjani2004modeling,sirjani2011ten} or Creol \cite{johnsen2006creol,johnsen2007asynchronous}. Furthermore, we use Labeled Transition Systems (LTSs) as the setup for generated assumptions.

%We use Labeled Transition System (LTS) as the setup for the generated assumption and also to describe the transition system semantics of an NPalang model. The semantics for  LTS are explained in details in Section \ref{sec:background}.

\begin{figure}[t]
    \centering
    \includegraphics[width=0.8\textwidth]{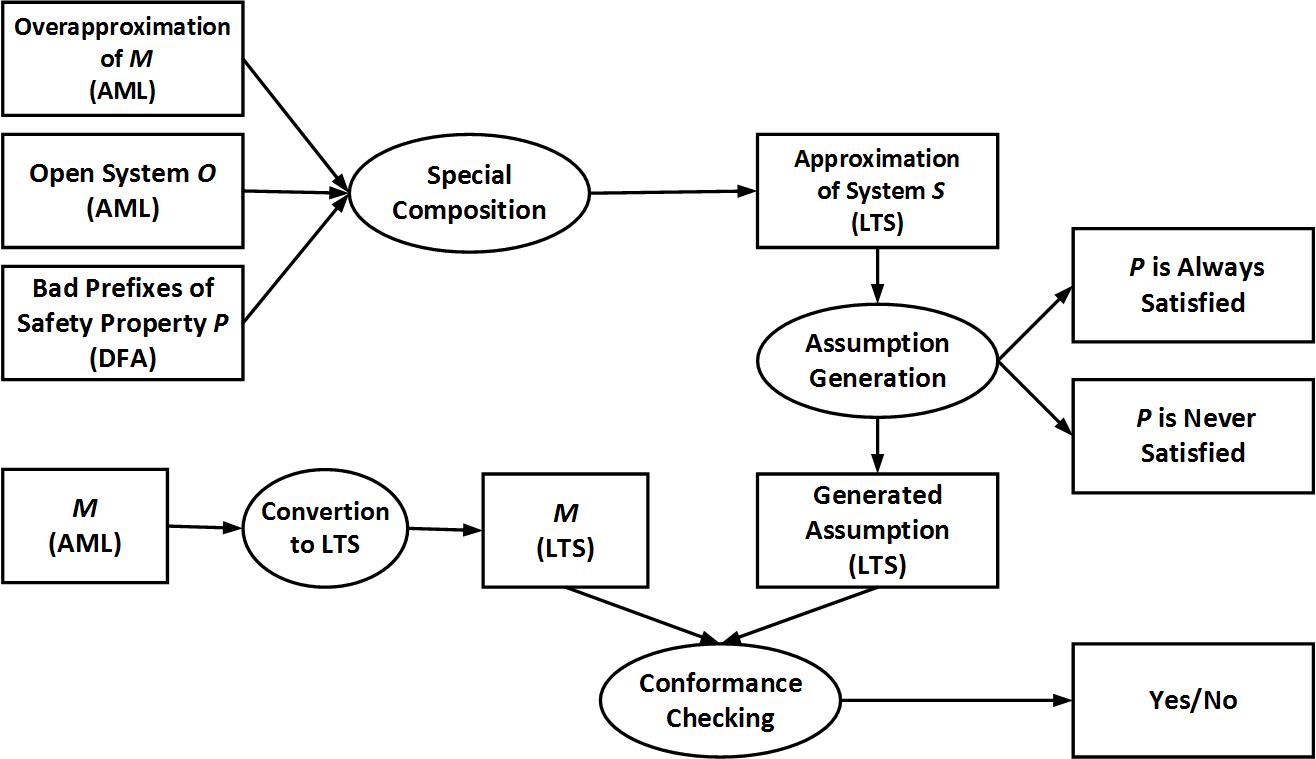}
    \caption{The proposed steps towards assumption generation in the asynchronous setting  }
    \label{fig:steps}
\end{figure}

The steps of the proposed method are illustrated in Figure~\ref{fig:steps}: By using the information extracted from the components of the open system we generate an actor which acts as an overapproximated interface of the missing component $M$. By
special composition of this actor with the actors of the open system using the property $P$, an LTS is produced
that describes an
overapproximation
%overapproximation
of the system. At this point we prune the resulted LTS to generate the assumption, adapting the approach of \cite{giann} for the asynchronous setting. %, while to remove the error traces from the overapproximated system. Unlike \cite{giann}, we assume asynchronous communication among components, so some adaptations have to be made.
The assumption generation step may result in the LTS of the generated assumption, or a deceleration that the property is always or never satisfied irrespective to $M$.

To check a given actor $M$ against the generated assumption, %first we %use the algorithm (presented in Section {\ref{sec:generatingM}}) to
% convert the \textit{open} actor $M$ into a \textit{closed} component (\textcolor{red}{ using a  \emph{wild environment}   which nondeterministically and continuously sends messages to $M$})
we propose an approach to derive its LTS characterization in terms of its interactions with an environment in the same line of \cite{MayActor}. Finally we investigate if the LTS conforms to the generated assumption according to the trace inclusion relation \cite{brookes1983behavioural}.

%\subsection{Paper Structure}
The paper is organized as follows: Section \ref{sec:background} explains the necessary background to our work and is followed by
our presented approach  in Section \ref{sec:approach}.
Section \ref{sec:CaseStudy}  discusses our experience with real life case studies.  Section \ref{sec:releatedWork} presents the related work. Section \ref{sec:futureWork&conclude} concludes the paper and briefly goes over the future extensions to our work.

\section{Background}\label{sec:background}
    %The necessary background  to our work consists of the semantics of NPalang language and  and formalization of safety properties.
We introduce AML, an actor-based language  for modeling communicating components in  concurrent systems, to study various aspects of actor systems. In this section, we briefly explain the syntax and semantics of AML. LTSs are used as the formal framework to define the semantics of AML models and our generated assumptions. We also briefly explain the adapted state-of-the-art direct approach  of \cite{giann} towards assumption generation.%  from which we have borrowed some steps. % to generate the assumption.

\subsection{Labeled Transition System}  \label{sec:LTS}
%Our notation is borrowed from \cite{giann} but has been reformatted to the notation used in \cite{baier2008principles}.
An LTS $T_L$ is formally introduced by a quadruple
$(S_L,\alpha T_L,R_L,$ $s^L_0)$, where $S_L$ is the set of states, $s^L_0 \in S_L$ is the initial state, $\alpha T_L \subseteq \Act$ denotes the set of actions of $T_L$, where $\Act$ is the set of all action labels, and $R_L \subseteq S_L \times \alpha T_L \cup \{\tau\} \times S_L $ is the transition relation. We use $s \xrightarrow{a} s'$ to denote $(s,a,s') \in R_L$, where $a\in\alpha T_L \cup \{\tau\}$. The action $\tau$ denotes an internal action of the component which is unobservable to its environment.

A $\emph{trace}$ of the LTS $T_L$ is a %maximal and initial trace of the system \cite{baier2008principles}, meaning it is a
sequence of actions $a_0~a_1~a_2~\ldots$ with a corresponding execution path $s_0\overto{a_0}s_1\overto{a_1} s_2\overto{a_2} \ldots$ which starts at the initial state ($s_0=s_0^L$), and is either infinite or ends in a deadlock state \cite{baier2008principles}.
Let $\mathit{Tr}(T_L)$ denote the set of all traces of $T_L$ and ${\it Tr}_{\it fin}(T_L)$ denote the set of   all prefixes of traces of $T_L$, i.e., finite traces of $T_L$. 

We use the renaming operator to rename the actions of an LTS using the renaming function $c: \Act \rightarrow \Act$. For a given LTS $T_{L}=(S_L,\alpha T_{L},R_L,s^L_0)$,
$\rho_c(T_L)=( S_L,\alpha T_{\rho_c(L)},R_{\rho_c(L)},s^L_0)$ where $R_{\rho_c(L)}=\{(s,c(a),s')|(s,a,s') \in R_L\}$ and $\alpha T_{\rho_c(L)}=\{c(a)\,\mid\, a\in \alpha T_L\}$.

\subsection{The Actor Modeling Language}\label{sec:palang}
The computation model of actors, introduced by Agha \cite{agha1985actors} for modeling distributed systems, consists of a set of encapsulated concurrent computation units, called actors, communicating through message passing. Each actor has a mailbox with a unique address which stores the received messages from other actors. It is assumed that the speed of computations is so fast in comparison to the message transmissions and hence, messages are processed in an atomic manner. Upon processing a message, an actor can update its local state, create new actors or send messages to other actors.

We introduce AML as the formal foundation of our work to model asynchronous systems. AML is a toy actor modeling language, inspired by Rebeca \cite{sirjani2004modeling}. Rebeca is an operational interpretation of the actor model with formal semantics and model checking tools developed to fill the gap between formal methods and software engineering \cite{sirjani2011ten,sirjani2004modeling}. AML has omitted some features of Rebeca for simplicity and is used for studying various aspects of actor systems.

\begin{figure}[b]
 \begin{center}
 \fbox{\parbox{\columnwidth-5mm}{\vspace{-2mm}
 % \begin{lstlisting}[captionpos=b,frame = single,  mathescape,lineskip=.05cm,tabsize=2, language=rebeca]
  \begin{align*}
  \mathrm{Model} &\Coloneqq~\mathrm{Actor}^+~\mathsf{main}~\{\mathrm{Send}^*~\}\\
  \mathrm{Actor} &\Coloneqq~\mathsf{actor}~\ID~(n)~\{~
  \langle\mathsf{int}~\Var;\rangle^*~|~\mathrm{Methods}^*~\}\\
  \mathrm{Methods} &\Coloneqq~~\MName~\{~\mathrm{Stat}^*~\} \\
  \mathrm{Stat}
  &\Coloneqq~
  \Var:=Expr;~|~\mathrm{Conditional}~|~\mathrm{Send}~|~
  \mathrm{nonDet};\\
  \mathrm{Conditional}
  &\Coloneqq~\mathsf{if}~(Expr)~\{\mathrm{Stat}^*\}~\mathsf{else}~\{\mathrm{Stat}^*\}~\\
  \mathrm{Send}
  &\Coloneqq~\lmeta \mathsf{self}~|~\ID \rmeta\,! \MName;\\
  \mathrm{nonDet}
  &\Coloneqq~\Var:=?(\lmeta Expr,\rmeta^*Expr );
  \end{align*}
 }}
 \end{center}
 \vspace{-2mm} \caption{AML syntax:  Angle brackets
 ($\lmeta~\rmeta$) are used as meta-parentheses.
 Superscript + is
 used for more than one repetition, and
 * indicates zero or more
 repetitions.
  The symbols  $\ID$, $\MName$, and $\Var$ denote the set of actor, method and variable names, respectively. The symbol $Expr$ denotes the set of arithmetic expressions and $n \in \mathbb{N}$ signifies the capacity of the actor's mailbox.
 \label{Fig::NPalangGrammar}}
 \end{figure}

\subsubsection{Notation}
Given a set $\Xi$, the set $\Xi^*$ is the set of all finite sequences over elements of $\Xi$. For a sequence $\xi\in \Xi^*$ of length $n$, the symbol $\xi_i$ denotes the $i^\mathrm{th}$ element of the sequence, where $1\leq i\leq n$. In this case, we may also write $\xi$ as $\langle \xi_1,\xi_2,\ldots,\xi_n\rangle$. The empty sequence is represented by $\epsilon$, and $\langle \iota|\xi\rangle$ denotes a sequence whose head and tail elements are $\iota\in\Xi$ and $\xi$, respectively. We use ${\it len}(\xi)$ to denote the length of the sequence $\xi$. For two sequences $\xi$ and $\xi'$ over $\Xi^*$, $\xi\oplus\xi'$ is the sequence obtained by appending $\xi'$ to the end of $\xi$. %, where $\oplus$ is a distributive operator.
This notation is lifted to the set of sequences in the usual way: $\xi\oplus\{\xi^1,\ldots,\xi^n\}=\bigcup_{1\le i\le n}\xi\oplus\xi^i $ where $\xi^i\in\Xi^
*$. For a function $f:X\rightarrow Y$, we use the notation $f[\mathfrak{x}\mapsto \mathfrak{y}]$ to denote the function $\{(\mathfrak{a},\mathfrak{b})\in f|\mathfrak{a}\neq \mathfrak{x}\}\cup\{(\mathfrak{x},\mathfrak{y})\}$, where % We also use the notation
$\mathfrak{x}\mapsto \mathfrak{y}$ is an alternative notation for $(\mathfrak{x},\mathfrak{y})$.

\subsubsection{Syntax}
%AML is configured to have nondeterministic statements. Furthermore,  and

Each AML model consists of several actor definitions and a $\mathsf{main}$ block which is used to specify the initial messages of the actors. Each actor is specified by using the $\mathsf{actor}$ reserved word. The capacity of the actor's mailbox is defined as part of its declaration. Each actor has a set of local variables and methods, where the method $m$ defines the behavior of the actor upon processing a message of type $m$. For  simplicity method arguments are abstracted away.

The formal syntax of AML
is described in Figure \ref{Fig::NPalangGrammar}. An example of an AML model is illustrated in Figure \ref{fig:examplePalang} which consists of two actors, namely $\it client$ and $\it server$ with the mailbox capacities of $10$. Initially, the message $\it reply$ is sent to the actor $\it client$ in the main block. Upon processing the messages of type $\it reply$, the actor $\it client$ nondeterministically sends a message $\it delay$ or $\it request$ to the
actor $\it server$. By taking a message $\it delay$ from its mailbox, $\it server$
 sends a message $\it request$ to itself, using the reserved word $\mathsf{self}$, and by taking a message $\it request$, responds to $\it client$  with a message $\it reply$. In both cases, the
 actor $\it server$ updates its state variable $z$.
 % upon taking a message.

\begin{figure}
\centering
\begin{lstlisting}[language=palang, multicols=2]
actor client(10) {
    int l;
    reply {
        l=?(0,1);
        if(l==0) {server!request;}
        else {server!delay;}
    }
}

actor server(10) {

    int z;
    request {
        client!reply;
        z=0;
    }

    delay{
        z=1;
        self!request;
    }
}

main {
    client!reply;
}
\end{lstlisting}
\caption{An example of AML model}
\label{fig:examplePalang}
\end{figure}

Each actor is abstractly an instance of the type $\Actor=\ID\times 2^\Var\times 2^\Mtd \times \mathbb{N}$, where:\begin{itemize}
    \item $\ID$ is the set of all actor identifiers in the model,
    \item $\Var$ is the set of all variable names,
    \item $\Mtd= \MName\times\Statement^*$ is the set of all method declarations where {\it MName} is the set of method names,
    \item $\mathbb{N}$ is the set of natural numbers. % \textcolor{red}{denoting the capacity of the actor mailbox}.
\end{itemize}

An actor $(\mathit{id},\mathit{vars},\mathit{mtds},n)$ has the identifier $\mathit{id}$, the set of state variables $\mathit{vars}$, the set of methods $\mathit{mtds}$ and the mailbox with the capacity of $n$. 
We assume that by default $\self\in \mathit{vars}$. Each declared method is defined by the pair $(m,b)\in \mathit{mtds}$, where $m$ defines the message type communicated between actors, and the name of the method used to serve messages of type $m$. Furthermore, $b$ contains the sequence of statements comprising the body of the method. For example the actor $\it client$ of Figure \ref{fig:examplePalang} is
specified by
$({\it client},\{l\},\{({\it reply},\langle l=?(0,1);, \mathsf{if}~(l==0)~\{ {\it server}!{\it request};\}$ $ ~\mathsf{else}~\{ {\it server}!{\it delay};\}  \rangle)\},10)$.
The set of AML models is specified by $2^\Actor\times\Statement^*$, where the second component corresponds to the main block consisting of a sequence of message send statements.

A given AML model is called \textit{well-formed}, if no two
state variables and methods of an actor or two actors have identical
names, identifiers of variables, methods and actors do not
clash,
the main block is restricted to  send statements with non-$\mathsf{self}$ receivers
and the receiver of a message has a method with the same name as the
message.
%identifier is not used in the main block
Also all actor accesses, message communications and
variable accesses occur over declared ones.  In the following we restrict to well-formed AML models.
\subsubsection{Semantics}\label{sec:semantics}
The formal semantics of AML models is given in terms of LTSs. To this aim, we make a few definitions and assumptions.

Let $\Val=\mathbb{N}\cup\ID$ contain all possible values that can be assigned to the state variables or to be used within the expressions. We treat the value $0$ as $\it false$ and others as $\it true$ in conditional statements.
%$\Val=\mathbb{N}\cup\{\mathrm{true},\mathrm{false}\}$.
%$\Val=\mathbb{N}$.
Furthermore, we define the set of all variable assignments as $\Env=\Var\rightarrow\Val$. %The semantics of expressions is given by the function $\eval_v:\Expr\rightarrow\Val$: \[\begin{array}{l}
%\hspace*{1.2cm}\eval_v(z) = v(z),~~ z\in\Var,\\
%\hspace*{1.2cm}\eval_v(c) = c , ~~c\in\Val,\\
%\hspace*{1 cm}\eval_v((e)) = \eval_v(e)\\
%\eval_v(e_1 ~{\it op}~e_2) = \eval_v(e_1)~{\it op}~\eval_v(e_2)
%\end{array}\]where $v \in \Env$ and $\it op$ contains the custom arithmetic operations.
As the main focus is on the message passing and the interleaving of actor executions, we abstract away from the semantics of expressions by assuming the function $\eval_v:\Expr\rightarrow\Val$
% $\eval_v:\Expr \cup \BExpr\rightarrow\Val$
which evaluates an expression for a specific variable assignment $v \in \Env$. %, where $\Expr$ is the set of arithmetic expressions. %, where $\Env$ is the set of all variable assignments.
%We assume that $v(x)=x$ if $x \notin \Var$.
We also consider $\eval_v$ is  lifted to the sequence of expressions: $\eval_v(\langle e_1,e_2,\ldots,e_n\rangle)=\langle \eval_v(e_1),\eval_v(e_2),\ldots,\eval_v(e_n)\rangle$.

We exploit the auxiliary functions $\svars:\ID\rightarrow2^\Var$, $\mathit{methods}:\ID \rightarrow 2^{\MName}$, $qSize: \ID \rightarrow \mathbb{N}$, and $cap:\ID \rightarrow \mathbb{N}$ to denote  the state variables, the set of method names, the current size and the mailbox capacity of a given actor identifier, respectively. Furthermore, we use $\body:\ID\times\MName\rightarrow\Statement^*$ to denote the body of the given method name declared by the given actor identifier.

We provide the \emph{coarse-grained} semantics of AML in which the execution of methods are non-preemptive, i.e., when an actor takes a message, it executes the entire body of the method before starting execution of another method. Therefore, the executions of statements in a method are not interleaved, in contrast to the \emph{fine-grained} semantics.

\begin{definition}
The coarse-grained semantics of an AML model $\mathcal{M}=({\it actors},\sigma_0)$, where ${\it actors}\subseteq \Actor$ and $\sigma_0\in \Statement^*$ denotes the sequence of send statements of the main block, is given by a labeled transition system  ${\it TS}(\mathcal{M})=(S, \Act_c ,\Rightarrow, s_0)$ such that:
\begin{itemize}
\item $S = \ID\rightarrow \Env\times\MName^*\times\Statement^*$ is the set of global states. The global state $s$ maps an actor identifier to its local state. The local state of an actor $(x,\mathit{vars},\mathit{mtds},n)\in {\it actors}$ is defined by the triple  $(v,\mathfrak{q},\sigma)$, where $v$ gives the values of its variables, $\mathfrak{q}:\MName^*$ is the mailbox of the actor with the capacity of $n$, and $\sigma:\Statement^*$ contains the sequence of statements the actor is going to execute to finish the service to the message currently being processed.

\item  $s_0$ is the initial state, where all state variables
are initialized to zero and the auxiliary variable $\self$ is initialized to the actor identifier.
% have their default values
Also the messages specified in the main block are put in the corresponding actors' mailboxes.
% : zero for integer and {\it false} for boolean
Formally, $s_0(x) = (\{\mathit{z}\mapsto 0\,\mid\,z\in\svars(x)\}\cup \{\self\mapsto x\}, cq(x,\sigma_0), \epsilon)$, where $cq(x,\sigma_0)$ constructs the initial mailbox of actor $x$ from $\sigma_0$:%the sequence of send statements $\sigma_0$ (from the main block) as defined below:

\begin{eqnarray*}
    cq(x, \epsilon) &=& \epsilon\\
    cq(x, \langle x!m|\sigma\rangle) &=&  \langle m |cq(x, \sigma)\rangle\\
    cq(x, \langle y!m|\sigma\rangle) &=& cq(x, \sigma),~y\neq x.
\end{eqnarray*}

\item $\Act_c  = \Act_T \times \Act_s^*$, where $\Act_T= \{t_x|x \in \ID\}$ and $\Act_s = \{\Snd(m)::x| x \in \ID, m \in \MName\}$ are the sets of take and send actions, respectively. Each action expresses that taking a message from the mailbox is followed by a sequence of send actions.
\item $\Rightarrow\,\subseteq S \times \Act_c \times S$ is the smallest relation derived from our two-tiered operational rules given in Table \ref{Tab::SOS}. A state transition only occurs as a consequence of applying the rule $\it Take$, in which the relation $\Rightarrow$ is defined in terms of the reflexive and transitive closure of the auxiliary relation $\rightarrow \,\subseteq\, S \times \Act_s \cup \{\epsilon\}\times S $. The relation $\rightarrow$ specifies the semantics of each statement while $\Rightarrow$ captures the effect of processing a message. 
\begin{table}[htbp]
       \centering
    \caption{AML natural semantic rules: $\sigma$, $\zeta$, and $T$ denote a sequence of statements, send actions, and method names, respectively. Furthermore, $\expr$ denotes an expression, $z\in \Var$ and $x,y\in\ID$. }\label{Tab::SOS}
    \begin{tabular}{lc}
        %\hline\\
${\it Assign}$ & $\Rule
{s(x)=(v,\mathfrak{q},\langle z:=\expr|\sigma\rangle)}
{s\overto{\epsilon} s[x\mapsto(v[z\mapsto\eval_v(\expr)],\mathfrak{q},\sigma)]}$\vspace*{4mm}\\
${\it NonDet}$ & $\Rule
{s(x)=(v,\mathfrak{q},\langle\mathit{var}:=?(\expr_1,\expr_2,\ldots,\expr_n)|\sigma\rangle)}
{\forall 1\leq i\leq n\,(s\overto{\epsilon} s[x\mapsto(v[\mathit{var}\mapsto\eval_v(\expr_i)],\mathfrak{q},\sigma)])}$\vspace*{4mm}\\
${\it Cond}_1$ & $\Rule
{s(x)=(v,\mathfrak{q},\langle\mathsf{if}~(\expr)~\{\sigma\}~\mathsf{else}~\{\sigma'\}|\sigma''\rangle)\,,
%%\wedge
\, \eval_v(\expr)}
{s\overto{\epsilon} s[x\mapsto(v,\mathfrak{q},\sigma\oplus\sigma'')]}$\vspace*{4mm}\\
${\it Cond}_2$ & $\Rule
{s(x)=(v,\mathfrak{q},\langle\mathsf{if}~(\expr)~\{\sigma\}~\mathsf{else}~\{\sigma'\}|\sigma''\rangle)\,,
%%\wedge
\, \neg\eval_v(\expr)}
{s\overto{\epsilon} s[x\mapsto(v,\mathfrak{q},\sigma'\oplus\sigma'')]}
$\vspace*{4mm}\\
${\it Snd}$ & $\Rule{
\begin{array}{c}
%   s(x)=(v,q,\langle y!m|\sigma\rangle)\,\wedge\,
    s(x)=(v,\mathfrak{q},\langle {\it expr}!m|\sigma\rangle)\,,
    %%\wedge
    \,
     y = {\eval_v({\it expr})}\,,
     %%\wedge
     \,
    s(y)=(v',\mathfrak{q'},\sigma')\,,
   %% \wedge
     \\\,qSize(y) < cap(y)
    \end{array}}
{s\overto{\Snd(m)::y} s[x\mapsto(v,\mathfrak{q},\sigma)][y\mapsto(v',\mathfrak{q'}\oplus\langle m\rangle,\sigma')]}$
\vspace*{4mm}\\
%${\it Snd}_u$ & $\Rule{
%    s(x)=(v,\mathfrak{q},\langle e!m|\sigma\rangle)\,,
%   %% \wedge
%    \, y=\eval_v(e)\,,
%    %%\wedge
%    \,
%         qSize(y) \ge cap(y)}
%{s\overto{\Snd(m)::y} s[x\mapsto(v,\mathfrak{q},\sigma)]}$
% \vspace*{4mm}\\
 ${\it Take}$ & $\Rule
 {\begin{array}{c}s(x)=(v,\langle m|T\rangle,\epsilon)\,,
 %%\wedge
 \\s'(x)=(v',\mathfrak{q'},\epsilon) \,,
 %%\wedge
 \,
  \forall y\in\ID\setminus\{x\}\,(s(y)=(v'',\mathfrak{q''},\epsilon)) %\wedge s'(i)=(v'',q^*,\epsilon))
 \,,
%% \wedge
 \,\\
 s[x\mapsto(v,T,body(x,m))]\overto{\zeta}\!\!^* s'
 \end{array}}
 {s\xRightarrow{(t_x,\zeta)} s'}$
\end{tabular}
\end{table}

\end{itemize}
\end{definition}
%To define $\Rightarrow$ in Table \ref{Tab::SOS}, we exploit an auxiliary relation $\rightarrow \,\subseteq\, S \times \Act_s \cup \{\epsilon\}\times S $
%to specify the semantics of each statement. 
%:\begin{itemize}
%    \item  $s\xrightarrow{\epsilon}\!\!^* s$;
%    \item $ s\xrightarrow{\zeta}\!\!^* t\, \wedge\, t\xrightarrow{ a} r \implies s\xrightarrow{\zeta \oplus \langle a\rangle}\!\!^*r$.
%\end{itemize}

An actor can take a message $m$ from its mailbox if no other actor is in the middle of processing a message, as explained by the premise of $\it Take$. The effect of processing the message $m$ is expressed by the semantics of the statements in $\body(x,m)$, captured by the reflexive and transitive closure of $\overto{a}$, denoted by
$\xrightarrow{\zeta}\!\!^*$ where $\zeta \in Act_s ^*$, and is defined as $s\xrightarrow{\epsilon}\!\!^* s$ and $ s\xrightarrow{\zeta}\!\!^* s'\, \wedge\, s'\xrightarrow{ a} s'' \implies s\xrightarrow{\zeta \oplus \langle a\rangle}\!\!^*s''$. % in either global states of $s$ and $s'$ and there is a path($\rightarrow ^*$) between $s$ and $s'$ such that through out the entire path no other actor is in the middle of processing a message.

Rule $\it Assign$ explains that the state variable of an actor is updated by an assignment statement. The non-deterministic assignment statement, non-deterministically chooses one of the values of its input expressions and assigns it to the specified state variable, as expressed by the rule $\it NonDet$. The behavior of the conditional statement is given by the rules ${\it Cond}_{1,2}$. A message is successfully sent to another actor (or itself) if the  mailbox of the recipient actor is not full, and as a consequence the message is inserted into the actor's mailbox, as explained by the rule ${\it Snd}$. Otherwise, a deadlock will occur in the model. %Otherwise by the rule ${\it Snd}_u$, the message is dropped. 
We remark that the sender identifier, denoted by $\it expr$ in the rules ${\it Snd}$, can be either an actor identifier $y\in\ID$ or the reserved word $\mathsf{self}$ which is mapped to its  corresponding identifier by $\eval_v$. %Furthermore, if the queue of the 

We use the notation $A_{o_1}  \interleave A_{o_2} \interleave \ldots \interleave A_{o_n}$ to denote the LTS induced from the actors $A_{o_1},A_{o_2}, \ldots,$ $A_{o_n}$,  using the coarse-grained semantic description of AML.

\subsection{The Direct Approach Towards Assumption Generation}\label{subsec::direct}
The direct approach of \cite{giann}
aims at finding the weakest  assumption that describes those environments of a component that their composition with the component satisfies a safety property. In this approach the component and the property are specified by LTSs. From the LTS of the property, an error LTS is generated that traps all the traces violating the property with a $\pi$ state. As we have been inspired by this approach, we briefly explain its steps.

%\sout{The component and the property are specified by LTSs. From the LTS of the property, an error LTS is generated  which traps all the traces violating the property with a $\pi$ state. To this aim, states are completed by adding missing transitions which terminate at the $\pi$ state. The following demonstrates the steps of \cite{giann}. We have applied each step in Figure \ref{fig:AssumpArticle} to the open system with two components ``Mutex'' and ``Writer'' (as specified by Figure \ref{Fig::writer}) to find an assumption which guarantees the mutual exclusion property (its error LTS is specified by Figure \ref{Fig::muProp}).}
%%todo this was completely wrong
%{\color{red}In this approach the component and the property are specified by LTSs. From the LTS of the property, an error LTS is generated  which traps all the traces violating the property with a $\pi$ state. To compute and eliminate the violating traces, the product of the component and the error LTS is computed and the assumption is generated by  completion of the  states by adding missing transitions which terminate at the $\theta$ state and removing the $\pi$ state. The following demonstrates the steps of \cite{giann}. In Figure \ref{fig:AssumpArticle}, each step has been applied to the components ``Mutex'' and ``Writer'' (as specified by Figure \ref{Fig::writer}) to find an assumption which guarantees the mutual exclusion property (its error LTS is specified by Figure \ref{Fig::muProp}).}

 \begin{enumerate}
 \item \textbf{Composition and minimization}
% LTS the property is satisfied in any environment.
In order to compute the violating traces of the system, the CSP-like parallel composition of the error LTS and the component is computed by which either they evolve  simultaneously by performing the same actions or only one proceeds by performing actions not common with its counterpart. 
In addition, the internal actions of the system are turned into $\tau$. These are the actions not controlled by the environment.  The resulting LTS is minimized using an arbitrary congruence which preserves the error traces. At this stage, if the $\pi$ state is not reachable, it is concluded that the property is satisfied in any environment. %Applying this step on the mentioned example results {\color{red}in} the LTS of Figure \ref{Fig::step1}.
\item \textbf{Backward error propagation} First those states from which  the $\pi$ state is reachable only through $\tau$ transitions are turned into $\pi$ by a backward analysis. Then all $\pi$ states are merged into a single $\pi$ state by integrating their transitions while self-loops are removed. This prunes those states  which the environment cannot prevent being entered into the $\pi$ state via internal actions.  If the initial state is pruned, then it is concluded that no environment can prevent the system from entering the error state and therefore the property is violated in any environment. Since only the error traces are  of interest, those states that are not backward reachable from the $\pi$ state are also pruned. %Applying this step on the LTS of Figure \ref{Fig::step1} results {\color{red}in the LTS of} Figure \ref{Fig::step2}.
%  \item \textbf{Backward Error Propagation} First states that the $\pi$ state is reachable from through only $\tau$ transitions are found by a backward analysis. Then,  such state together with those that are not backward reachable from the $\pi$ state are pruned. If the initial state should be pruned, then no environment can prevent the system from entering the error state and therefore the property is violated in any environment. Applying this step on the LTS of Figure \ref{Fig::step1} results Figure \ref{Fig::step2}.
\item\textbf{Property extraction}
Treating $\tau$ transitions as $\epsilon$-moves, the LTS of the previous step can be made deterministic by using the well-known subset construction method \cite{Noord00}. Consequently $\tau$ transitions are eliminated and the LTS is also made complete by adding missing transitions of the states. These transitions, leading to a newly added sink state $\theta$, are those behaviors that are not exercised by the component. Adding such transitions generalizes the assumption, and hence environments with such additional behaviors become 
acceptable. % which means there is no need for assumptions about these behaviors. 
We remark that each state of the resulting deterministic LTS is a subset of the states of the original nondeterministic LTS.
Therefore, states of the resulting deterministic LTS containing the $\pi$ state are treated as $\pi$. 
%Considering a trace in the original LTS which may lead to the error state (due to non-determinism) as an error trace of the resulting deterministic LTS, states of the resulting deterministic LTS containing the $\pi$ state are treated as $\pi$. 
Finally, the assumption is achieved by removing the $\pi$ state and its transitions. %Figure \ref{Fig::assumption} shows the result assumption.
  \end{enumerate}

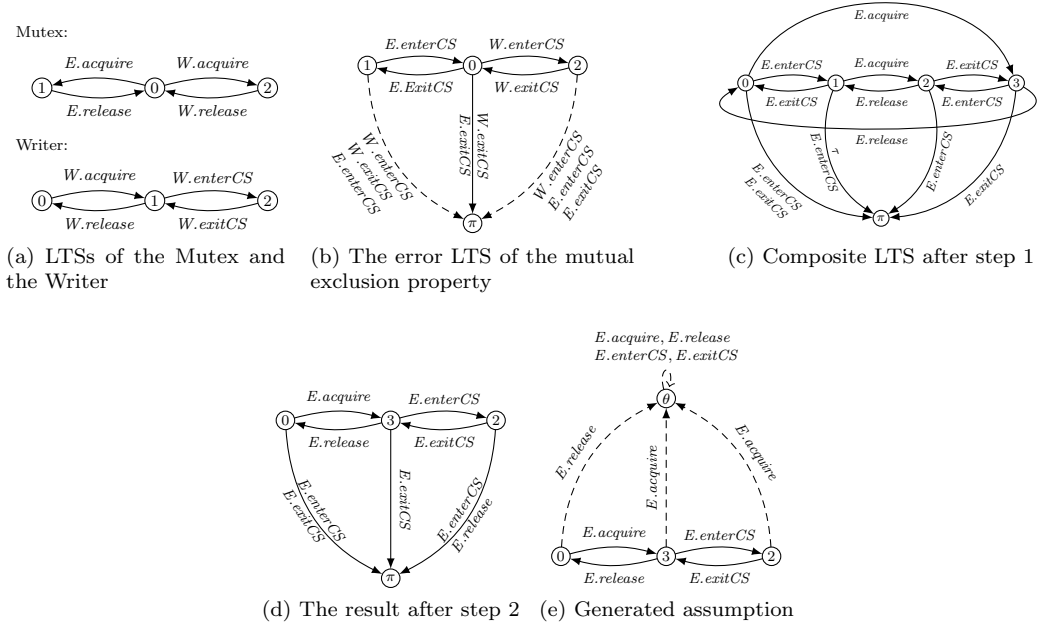
\begin{figure}[thp]
    \centering
  ~~~~~~~~
\subfloat[LTSs of the Mutex and the Writer\label{Fig::writer}]{
    \begin{tikzpicture}[scale=.75, transform shape]
\node [style=circle,draw,outer sep=0,inner sep=1,minimum size=10] (v3)at (-1,1) {$2$};
\node [style=circle,draw,outer sep=0,inner sep=1,minimum size=10] (v2) at (-3,1) {$0$};
\node [style=circle,draw,outer sep=0,inner sep=1,minimum size=10] (v1) at (-5,1) {$1$};
\node [style=circle,draw,outer sep=0,inner sep=1,minimum size=10] (v4)at (-5,-1) {$0$};
\node [style=circle,draw,outer sep=0,inner sep=1,minimum size=10] (v5)at (-3,-1) {$1$};
\node [style=circle,draw,outer sep=0,inner sep=1,minimum size=10] (v6)at (-1,-1) {$2$};
\node [outer sep=0,inner sep=1,minimum size=10] at (-5,2) {Mutex:};
\node [outer sep=0,inner sep=1,minimum size=10] at (-5,0) {Writer:};

\path[<-,-latex]  (v2) edge  [bend left = -15] node[above =0.03 cm] {\it E.acquire} (v1);
\path[ ->,-latex]  (v1) edge  [bend left = -15] node[below =0.03 cm] {\it E.release} (v2);
\path[->,-latex]  (v3) edge  [bend right = -15] node[below =0.03 cm] {\it W.release} (v2);
\path[ <-,-latex]  (v2) edge  [bend right = -15] node[above =0.03 cm] {\it W.acquire} (v3);

\path[->,-latex]  (v4) edge  [bend right = -15] node[above =0.03 cm] {\it W.acquire} (v5);
\path[->,-latex]  (v5) edge  [bend left = 15] node[below =0.03 cm] {\it W.release} (v4);
\path[->,-latex]  (v5) edge  [bend right = -15] node[above =0.03 cm] {\it W.enterCS} (v6);
\path[ <-,-latex]  (v6) edge  [bend right = -15] node[below =0.03 cm] {\it W.exitCS} (v5);
        \end{tikzpicture}
}   ~~~~
\subfloat[The error LTS of the mutual exclusion property]{\label{Fig::muProp}
\begin{tikzpicture}[scale=.7, transform shape]

        \node [style=circle,draw,outer sep=0,inner sep=1,minimum size=10] (v3)at (-1,1) {$2$};
        \node [style=circle,draw,outer sep=0,inner sep=1,minimum size=10] (v2) at (-3,1) {$0$};
        \node [style=circle,draw,outer sep=0,inner sep=1,minimum size=10] (v1) at (-5,1) {$1$};
        \node [style=circle,draw,outer sep=0,inner sep=1,minimum size=10] (v5)at (-3,-2) {$\pi$};

        \path[<-,-latex]  (v1) edge  [bend right = -15] node[above =0.03 cm] {\it E.enterCS} (v2);
        \path[ ->,-latex]  (v2) edge  [bend right = -15] node[below =0.03 cm] {\it E.ExitCS} (v1);
        \path[->,-latex]  (v3) edge  [bend right = -15] node[below =0.03 cm] {\it W.exitCS} (v2);
        \path[ <-,-latex]  (v2) edge  [bend right = -15] node[above =0.03 cm] {\it W.enterCS} (v3);
        \path[ ->,-latex]  (v2) edge   node[sloped] {$\begin{array}{c}{\it W.exitCS}\\{\it E.exitCS}\end{array}$} (v5);
        \path[ ->,-latex,densely dashed]  (v3) edge  [bend right = -30] node[sloped,below=-0.07 cm]  {$\begin{array}{l}{\it W.enterCS}\\{\it E.enterCS}\\{\it E.exitCS}\end{array}$} (v5);
        \path[ ->,-latex,densely dashed]  (v1) edge   [bend left = -30] node[sloped,below=-0.07 cm]  {$\begin{array}{l}{\it W.enterCS}\\{\it W.exitCS}\\{\it E.enterCS}\end{array}$} (v5);
        \end{tikzpicture}
}
\subfloat[Composite LTS after step $1$
%: ``$W$'' labeled transitions are internal
]{\label{Fig::step1}
        \begin{tikzpicture}[scale=.6, transform shape]
      
      \node [style=circle,draw,outer sep=0,inner sep=1,minimum size=10] (v3)at (-1,1) {$2$};
      \node [style=circle,draw,outer sep=0,inner sep=1,minimum size=10] (v2) at (-3,1) {$1$};
      \node [style=circle,draw,outer sep=0,inner sep=1,minimum size=10] (v1) at (-5,1) {$0$};
      \node  [style=circle,draw,outer sep=0,inner sep=1,minimum size=10] (v4)at (1,1) {$3$};
      \node [style=circle,draw,outer sep=0,inner sep=1,minimum size=10] (v5)at (-2,-2) {$\pi$};
      
      \path[->,-latex]  (v1) edge  [bend right = -15] node[above =0.03 cm] {\it E.enterCS} (v2);
      \path[ <-,-latex]  (v2) edge  [bend right = -15] node[below =0.03 cm] {\it E.exitCS} (v1);
      \path[->,-latex]  (v3) edge  [bend right = -15] node[below =0.03 cm] {\it E.release} (v2);
      \path[ <-,-latex]  (v2) edge  [bend right = -15] node[above =0.03 cm] {\it E.acquire} (v3);
      \path[->,-latex]  (v3) edge  [bend right = -15] node[above =0.03 cm] {\it E.exitCS} (v4);
      \path[ <-,-latex]  (v4) edge  [bend right = -15] node[below =0.03 cm] {\it E.enterCS} (v3);
      \path[->,-latex]  (v4) edge  [bend right = -150] node[below =0.03 cm] {\it E.release} (v1);
      \path[ <-,-latex]  (v1) edge  [bend right = -68] node[below =0.03 cm] {\it E.acquire} (v4);
      
      \path[->,-latex]  (v1) edge  [sloped,bend left = -40] node[below =0.03 cm] {$\begin{array}{c} {\it E.enterCS}\\ {\it E.exitCS} \end{array}$} (v5);
      \path[->,-latex]  (v2) edge  [bend left = -40] node[sloped] {$\begin{array}{c}  \tau~~~\\ {\it E.enterCS}\end{array}$} (v5);
      \path[->,-latex]  (v3) edge  [bend right = -40] node[sloped,below =-0.03 cm] {\it E.enterCS} (v5);
      \path[->,-latex]  (v4) edge  [sloped,bend right = -40] node[below =0.03 cm] {\it E.exitCS} (v5);
      
              \end{tikzpicture}
}
\hfill
\subfloat[The result after step $2$ ]{\label{Fig::step2}
        \begin{tikzpicture}[scale=.7, transform shape]

\node [style=circle,draw,outer sep=0,inner sep=1,minimum size=10] (v3)at (-1,1) {$2$};
\node [style=circle,draw,outer sep=0,inner sep=1,minimum size=10] (v2) at (-3,1) {$3$};
\node [style=circle,draw,outer sep=0,inner sep=1,minimum size=10] (v1) at (-5,1) {$0$};
\node [style=circle,draw,outer sep=0,inner sep=1,minimum size=10] (v5)at (-3,-2) {$\pi$};

\path[->,-latex]  (v1) edge  [bend right = -15] node[above =0.03 cm] {\it E.acquire} (v2);
\path[ <- ,-latex]  (v2) edge  [bend right = -15] node[below =0.03 cm] {\it E.release} (v1);
\path[->,-latex]  (v3) edge  [bend right = -15] node[below =0.03 cm] {\it E.exitCS} (v2);
\path[ <-,-latex]  (v2) edge  [bend right = -15] node[above =0.03 cm] {\it E.enterCS} (v3);
\path[ ->,-latex]  (v2) edge   node[sloped,above=-0.03 cm] {$\begin{array}{c}{\it E.exitCS}\end{array}$} (v5);
\path[ ->,-latex]  (v3) edge  [bend right = -30] node[ sloped] {$\begin{array}{l}{\it E.enterCS}\\{\it E.release}\end{array}$} (v5);
\path[ ->,-latex]  (v1) edge   [bend left = -30] node[sloped] {$\begin{array}{l}{\it E.enterCS}\\{\it E.exitCS}\end{array}$} (v5);
        \end{tikzpicture}
}~
\subfloat[Generated assumption]{\label{Fig::assumption}
        \begin{tikzpicture}[scale=.7, transform shape]

\node [style=circle,draw,outer sep=0,inner sep=1,minimum size=10] (v3)at (-1,-2) {$2$};
\node [style=circle,draw,outer sep=0,inner sep=1,minimum size=10] (v2) at (-3,-2) {$3$};
\node [style=circle,draw,outer sep=0,inner sep=1,minimum size=10] (v1) at (-5,-2) {$0$};
\node [style=circle,draw,outer sep=0,inner sep=1,minimum size=10] (v5)at (-3,1) {$\theta$};

\path[->,-latex]  (v1) edge  [bend right = -15] node[above =0.03 cm] {\it E.acquire} (v2);
\path[ <-,-latex]  (v2) edge  [bend right = -15] node[below =0.03 cm] {\it E.release} (v1);
\path[->,-latex]  (v3) edge  [bend right = -15] node[below =0.03 cm] {\it E.exitCS} (v2);
\path[ <-,-latex]  (v2) edge  [bend right = -15] node[above =0.03 cm] {\it E.enterCS} (v3);
\path[ ->,-latex,densely dashed]  (v2) edge   node[sloped,above=-0.03 cm] {$\begin{array}{c}{\it E.acquire}\end{array}$} (v5);
\path[ ->,-latex,densely dashed]  (v3) edge  [bend left = -30] node[sloped, above] {$\begin{array}{l}{\it E.acquire}\end{array}$} (v5);
\path[ ->,-latex,densely dashed]  (v1) edge   [bend right = -30] node[sloped,above] {$\begin{array}{l}{\it E.release}\end{array}$} (v5);
\path[ ->,-latex,densely dashed]  (v5) edge   [loop above] node[align=left] {$\begin{array}{l}{\it E.acquire},{\it E.release}\\{\it E.enterCS},{\it E.exitCS}\end{array}$} (v5);
        \end{tikzpicture}
}
\caption{The steps of assumption generation \cite{giann}: the states with the label of $0$ are initial}
    \label{fig:AssumpArticle}
\end{figure}

%In Figure \ref{fig:AssumpArticle},
These steps have been applied to an open system consisted of ``Mutex'' and ``Writer'', as specified by Figure \ref{Fig::writer}. ``Mutex'' component controls accesses to a critical section. It may receive requests to access from either Writer or the environment. We aim to find an assumption for the environment which guarantees mutual access to the critical section, called \emph{the mutual exclusion property} (its error LTS is specified by Figure \ref{Fig::muProp}). Figure \ref{Fig::step1} illustrates the resulted LTS after applying step $1$ by which the parallel composition of the property and the components is computed while the actions of Writer component, i.e., $\it W.acquire$, $\it W.release$, $\it W.enterCS$, and $\it W.exitCS$ are turned into $\tau$.  As state $1$ is backward reachable from $\pi$ through a $\tau$ transition, it is merged to the $\pi$ state by the backward error propagation step as shown in Figure \ref{Fig::step2}. After adding the missing transitions, identified by the dashed lines in Figure~\ref{Fig::assumption}, and removing the $\pi$ state, the assumption is achieved. As imposed by the Mutex, the actions 
${\it E.acquire}$ and ${\it E.release}$ can only alternate and therefore, the access of any acceptable environment to the critical
section is definitely protected by the Mutex. Behaviors leading to the sink state, such as  ${\it E.release}$ in the state $0$, are inconsequential as Mutex is not able to do such an action. 

\subsection{Formalization of Correctness Properties}\label{sec:property}
A safety property $P$ can be specified by a Linear Time Logic (LTL) formula (\cite{huth2004logic}) over $\alpha P$ as the set of atomic propositions. A given LTL property can be validated to be safety by using the approach of \cite{safteyCheck}.
%,  such that $\alpha P = \Act_T \times \Act_s^*$.
The semantics of LTL formula $P$ is defined
as a language of all infinite words over the alphabet $\alpha P$ which satisfy $P$.

Any infinite word that violates a safety property $P$, starts with a finite \emph{bad prefix}~\cite{baier2008principles}. %We are interested in minimal bad prefixes, i.e., those bad prefixes with no bad prefix.
A Deterministic Finite Automata (DFA) accepting the bad prefixes of a safety LTL property $P$, denoted by $P_{\mathit{err}}$, can be formed using the approach of \cite{LTLDFA}, to reduce the verification to reasoning about finite words: a system satisfies $P$ if none of its computations has a
bad prefix.  %Deterministic Finite Automata (DFA) can be formed for each $P$, denoted by $P_{\mathit{err}}$, whose accepting language is the set of minimal bad prefixes of $P$ . 
Formally,  $P_{\mathit{err}}$ is a quintuple  $( Q,\Sigma,\partial,q_0,\{\pi\})$  where $Q$ is the set of States,
$q_0$ is the initial state and $\partial$ is the transition relation $ \partial :Q \times \Sigma \rightarrow Q$ where for all states $q \in Q$ and all symbols $\jmath \in \Sigma$, $|\partial(q,\jmath)|=1$. The unique \textit{accept} (or final) state of $Q$ is called $\pi$ which traps all possible violations of $P$. Therefore, a verification problem can be easily decided by inspecting the reachability of $\pi$.  This idea was used by \cite{giann} in generating an assumption. %Figure \ref{fig:DFAProperty} demonstrates an LTL property and its corresponding $P_{\mathit{err}}$. 

In following we consider   $\Sigma \subseteq \Act_s$ as our atomic propositions. Hence,  properties are defined on send actions.

\section{Assumption Generation}\label{sec:approach}
According to the problem definition, a system is composed of several components which  communicate through message passing. This setting makes the AML model a suitable candidate for specification of the system components. In our context, the specification of a component of the system which operates as a service is assumed unavailable. This component is addressed as $M$.
The set containing the rest of the system components, which are all specified in AML, is addressed as the open system $O$. %the components of $O$ are described as actors of an NPalang model.
%As part of the presented solution toward the defined problem,  we  generate an actor which is an overapproximation of the unspecified component $M$.  This actors along side the actors of $O$ will form a whole system.
%The main idea is to
% build an spacial  {\it property LTS} by using an special composition (defined in Section \ref{sec:propertyLTS}) of these actors  and then reconfigure this LTS by partially adapting  the approach  mentioned in \cite{giann} to generate an assumption.
‌‌‌‌‌‌

Our generation algorithm exploits an overapproximated actor for the undefined component $M$. Furthermore, it computes the effect of the given property on the behavior of the system efficiently by generating a so-called \textit{property LTS }. Finally by refining this LTS, an assumption is generated.

%The following demonstrates the steps of our algorithm in four main phases: First we generate an overapproximated actor which is supposed to replace $M$,
%then we compute the effect of the given property on the behavior of the system efficiently, resulting an LTS. Afterwards, we refine the LTS to produce an assumption.

 In the following we first describe our running example and afterwards we go through the generation of the overapproximated actor and the  property LTS. We then elaborate on the steps of the actual assumption generation and in the end
we explain how a given actor can be checked against the generated assumption.

\begin{figure}

    \begin{lstlisting}[language=palang, multicols=2]
actor left(2) {

    initialL{
        mutex!reqL;
    }
    permitL {
        //access critical section
        mutex!release;
        mutex!reqL;
    }
}

actor right(2) {

    initialR{
        mutex!reqR;
    }
    permitR {
        //access critical section
        mutex!release;
        mutex!reqR;
    }
}

main {
    left!initialL;
    right!initialR;
}
    \end{lstlisting}
    \caption{Specification of  open system}
    \label{fig:exampleOpenSystem}
\end{figure}

\subsection*{Running Example}
We present our approach through an example. The system  contains three actors: $\mathit{left}$, $\mathit{right}$, and $\mathit{mutex}$. The actors $\mathit{left}$ and $\mathit{right}$ that constitute the open system (Figure~\ref{fig:exampleOpenSystem}), communicate through $\mathit{reqL}$ and $\mathit{reqR}$ messages respectively with  $\mathit{mutex}$ to request for permission to enter the critical section. The actor $\mathit{mutex}$ plays the role of the unspecified component and responds with $\mathit{permitL}$ and $\mathit{permitR}$ to grant access to $\mathit{left}$, $\mathit{right}$ respectively if the critical section is not currently occupied. Furthermore, $\mathit{left}$ and $\mathit{right}$ send a $\mathit{release}$ message to $\mathit{mutex}$ upon exiting the critical section.

The required property ensures mutual exclusion over the critical section,  specified by $P$ and its corresponding $P_{\mathit{err}}$ in Figure~\ref{fig:DFAProperty}. Note that an edge from a state $q$ to a state $q'$ labeled with formula $\varphi$ means that there are transitions $q \xrightarrow{a} q'$ for all $a\in \alpha P\,(a \models \varphi)$, where $a\models {\it true}$, $a\models a$, $a\models \varphi_1 \odot \varphi_2$ if $ a\models \varphi_1 \odot a\models \varphi_2$ with $\odot\in\{\wedge, \vee\}$, and $a\models \neg \varphi$ if $a\not\models \varphi$.
%. {\color{red}We have borrowed the notation of \emph{action formulas} of the regular alternation-free $\mu$-calculus {\footnote{http://cadp.inria.fr/man/mcl3.html}}, to represent a set of labels}.
\begin{figure}
\centering
    \begin {tikzpicture}

      \tikzstyle{bordered} = [draw,outer sep=0,inner sep=1,minimum size=10]

\node [style=circle,draw,outer sep=0,inner sep=1,minimum size=10] (v1)at (-3,2) {$1$};
\node [style=circle,draw,outer sep=0,inner sep=1,minimum size=10] (v2)at (-6,2) {$0$};
\node [style=circle,draw,outer sep=0,inner sep=1,minimum size=10] (v3) at (-1,2) {$\pi$};
\node (v0) at(-6.5,2) {};

\path[<-,-latex]  (v2) edge  [bend right = -15] node[above =0.03 cm] {$\alpha \lor \beta $} (v1);
\path[ ->,-latex]  (v1) edge  [bend right = -15] node[below =0.03 cm] {$\gamma$} (v2);
\path[->,-latex]  (v1) edge node[above =0.03 cm] {$\alpha \lor \beta $} (v3);
\path[<-,-latex]  (v2) edge   [loop above] node {$\lnot(\alpha \lor \beta)$ } (v2);
\path[<-,-latex]  (v1) edge   [loop below] node {$\lnot(\alpha \lor \beta \lor \gamma)$} (v1);
\path[<-,-latex]  (v0) edge (v2);

\node at (2,2)[bordered] { $\begin{array}{l}{\alpha =\it Snd(permitL)::left}\\
\beta = {\it Snd(permitR)::right}\\ \gamma = {\it Snd(release)::mutex} \end{array}$};
\end{tikzpicture}
\caption{$P_{\mathit{err}}$ for property $P =  \square\hspace{1mm} ((\alpha \lor \beta) \rightarrow \lnot ( \alpha \lor \beta) \bigcup \gamma)$}
\label{fig:DFAProperty}
    \end{figure}

\subsection{Overapproximating the Unspecified Component}\label{sec:OpenSystem}
To generate an actor which is an over-approximated interface of the missing actor, we assume that some basic information on the interactions of the open system with its environment is available. Software components are usually developed based on some sort of documentation and therefore this information can be either provided by the system modeler or extracted from the software documents. For example by using  {UML} activity or sequence diagrams, one can easily extract the interface communications of components, based on the intended level of abstraction. A more general abstraction is more appealing due to reduction in the preparation effort of such information. Furthermore, the essence of asynchrony necessitates that the order of sending messages to different actors is not of importance. Therefore, we are only interested in information about the messages sent to $M$ by the components of $O$ and an over-approximation of the messages that $O$’s components can receive in response.

Formally we describe this information by a set of pairs of type $\MName \times 2^{(\MName^* \times \ID)}$, denoted by $\Info$. Each pair $(m,\mathcal{I}) \in \Info$, where $\mathcal{I}=\{(r^\mathcal{I}_i,o^\mathcal{I}_i)|~0\leq i \leq k\}$, expresses that the actor $A_{o^\mathcal{I}_i} \in O$ expects to receive the message sequence $r^\mathcal{I}_i$ from $M$ after  the message $m$ is sent to $M$. By $\Info$, we aim to characterize the most general behaviors expected from an unspecified component. In other words, any actor playing the role of $M$ should not exhibit any behavior beyond it. Such actors are said to be in \textit{compliance} with $\Info$ (the notion of compliance will be explained  in Section \ref{sec:checkM}). If the information of $\Info$ is not accurate, in the sense that the response of $\Info$ to the message $m$ is more or less than what really the components of $O$ expect, then our approach still works but the set of possible $M$s will change accordingly due to the notion of our compliance. In other words, only $M$s triggering responses specified by $\Info$, may be accepted.

To cover cases such that the open system expects to receive multiple possible sequences of responses by different actors of $O$, multiples of $(m,\mathcal{I}_1),\ldots,(m,\mathcal{I}_j)$ for some $j>0$ can be used.
%, or a case that the response  is partly to a different actor of $O$; By use of this formalism we are able to describe these different circumstances.
For example, $\Info= \{( m_1,\{(\langle m_{11},m_{12}\rangle,o_1), (\langle m_{21}\rangle, o_2)\}) ,(m_1,\{(\langle m_{31}\rangle,o_3)\}),(m_2,\{(\langle m_{31} \rangle,o_3)\})\}$ explains that the components of $O$ expect one of two possible scenarios to happen in response to sending $m_1$ to $M$: Either $o_1$ receives $m_{11}$ and then $m_{12}$ and also $o_2$ receives $m_{21}$, or only $o_3$ receives $m_{31}$. In response to sending $m_2$, $o_3$ expects to receive $m_{31}$. The usage of sets belonging to $2^{(\MName^* \times \ID)}$ instead of sequences is justified by the fact that in most asynchronous systems there is no guarantee that messages are received by different actors in the order they have been sent. A description of $\Info$ is called well-formed if for all $(m,\mathcal{I})\in\Info$ and for all $(r,o)\in \mathcal{I}$, ${\it len}(r)$ is less than or equal to the capacity of the actor $A_o$. In following, we only consider well-formed $\Info$s.

\begin{example}\label{EX::Info}
For the case of our example, since the goal of the system is mutual access to the critical section, it is reasonable to assume that the open system expects to eventually receive $\mathit{permitL}$ in response to $\mathit{reqL}$, $\mathit{permitR}$ in response to $\mathit{reqR}$ and upon sending $\mathit{release}$, it does not expect a response. Hence, $\Info = \{(\mathit{reqL},\{(\langle \mathit{permitL}\rangle,\mathit{left})\}, (\mathit{reqR},\{(\langle\mathit{permitR}\rangle,\mathit{right})\}),
(\mathit{release},\{\})\}$.
\end{example}

We use $\Info$ to form an actor as a stub for the unspecified component $M$.
This actor which is an overapproximated interface for  $M$, is called $A_{\IntM}$.
The idea is to compose this replacement of $M$ with the components of $O$ by using $P_{\mathit{err}}$ to produce an LTS in such a way that the error traces of the system are computed.  We then reconfigure this LTS to produce the assumption. The actor $A_{\IntM}$ is defined such that for each pair  $(m,\mathcal{I})\in \Info$, a method with the same name as $m$  will be defined. In the body of $m$, the messages in the message sequence $r^\mathcal{I}_i$ of each $(r^\mathcal{I}_i,o^\mathcal{I}_i) \in \mathcal{I}$ will be sent to the actor $o^\mathcal{I}_i$. If there exists multiple pairs for a message $m$ in $\Info$, then in the body of the method serving $m$, the sequence of messages will be sent to the corresponding actors nondeterministically. Due to the nature of asynchrony, the order of messages sent to different actors should not be of importance. As the order of send actions in the labels of an {AML} semantic model are the same as the order of their corresponding send statements, we implement this property during the construction of $\IntM$.  Thus, any order between sending the messages of the sequences $r^\mathcal{I}_i$ and  $r^\mathcal{I}_j$ to the actors $o^\mathcal{I}_i$ and $o^\mathcal{I}_j$ should be considered while the order of  messages in each sequence should be preserved. We remark that the order of messages from a sender to a receiver is preserved in actor models which conforms to the reality of most asynchronous communicating systems~\cite{zakeriyan2015jacco}. To this aim, we exploit an auxiliary function ${\it shuffle}(\mathcal{I})$ which interleaves the message sequences to different actors, as defined by $\mathcal{I}$,  accordingly:{{\allowdisplaybreaks
\begin{flalign*}
&{\it shuffle}(\emptyset)=\emptyset\\
&{\it shuffle}(\{(r,o)\}\cup \mathcal{I})= {\it xshuffle}(o\boldsymbol{!}r,{\it shuffle}(\mathcal{I}))\\ %\\
&{\it xshuffle}(\sigma,\emptyset)=\{\sigma\}\\
& {\it xshuffle}(\sigma_1,\{\sigma_2\}) =  {\it hshuffle}(\sigma_1,\sigma_2)\\
& {\it xshuffle}(\sigma_1,\{\sigma_2\}\cup \Upsilon) =  {\it hshuffle}(\sigma_1,\sigma_2)\cup  {\it xshuffle}(\sigma_1,\Upsilon) \mbox{, where $\Upsilon\neq\emptyset$}\\ %\\
&  {\it hshuffle}(\epsilon,\sigma)=\sigma\\
&  {\it hshuffle}(\sigma,\epsilon)=\sigma\\
& {\it hshuffle}(\langle x!m_1;\rangle\oplus \sigma_1,\langle y!m_2;\rangle\oplus \sigma_2)=\\
 &\hspace*{1cm} \langle x!m_1;\rangle\oplus{\it hshuffle}( \sigma_1,\langle y!m_2;\rangle\oplus \sigma_2)\cup \langle y!m_2;\rangle\oplus {\it hshuffle}(\langle x!m_1;\rangle\oplus \sigma_1,\sigma_2)\\ %\\
%& c!\epsilon = \epsilon\\
%& c! \langle m|T \rangle = \langle c!m\rangle\oplus c!T
\end{flalign*}}%
}where $o\boldsymbol{!}r$  abbreviates  sending the message sequence $r$ to the component $o$ and is defined by the equations $o\boldsymbol{!}\epsilon = \epsilon$ and $o\boldsymbol{!} \langle m|T \rangle = \langle o!m;\rangle\oplus o\boldsymbol{!}T$. The auxiliary function ${\it xshuffle}$ inserts the send statements of $o\boldsymbol{!}r$ randomly into the send statement sequences of its second operand with the help of ${\it hshuffle}$. This function interleaves the statements of two sequences arbitrarily while preserving the order of send messages in each sequence. %\[
%c^{S_j}_i!r^{S_j}_i=
%\begin{cases}
%\epsilon & r^{S_j}_i=\epsilon,\\
%\langle c^{S_j}_i!msg;\rangle\oplus c^{S_j}_i!T  &r^{S_j}_i =\langle msg|T \rangle.\\
%\end{cases}
%\]

We would have got a considerable practical gain if such reorderings could had been handled semantically through an appropriate notion of equivalence such as may-testing \cite{MayActor,MayAsynch} instead of explicit shuffling of the messages in each message body of $\IntM$. Due to the non-existence of a practical tool to compare the traces regardless of their action orders in our experiments, we were confined to follow this approach.

\begin{definition}\label{def:IntM}
Actor $A_{\IntM}$ is denoted by $(M,\{l\},\mathit{mtds}_{\IntM},\infty)$ such that for all distinct pairs $(m,\mathcal{I}_1)$, $(m,\mathcal{I}_2)$, $\ldots,( m,\mathcal{I}_n)$ $\in$ $\Info$ for some $n\ge 1$, where $\forall\, 1\leq j \leq n\,(       {\it shuffle}(\mathcal{I}_j) = \{\sigma_i^{\mathcal{I}_j}|~ 1 \leq i \leq k_j\})$, then $\{(m,b)\}\cup\{(m_{\mathcal{I}_j,\sigma_i}^g,b_{\mathcal{I}_j,\sigma_i}^g)\,\mid\, %1\le j\le n,\,1\le i\le k_j,\, 
1\le g \le {\it len}(\sigma_{i}^{\mathcal{I}_j}) \}\subseteq \mathit{mtds}_{\IntM}$.  The bodies of methods $m$ and $m_{\mathcal{I}_j,\sigma_i}^g$ are defined by: % the pairs $(m,b)$ and $(m_{S_j,i}^v$, respectively:
\begin{align*}
\begin{split}
%b &=\langle\\
 b=\quad\quad& \langle l = ? (1,\ldots,n);,\\
& \mathsf{if}(l==1)~\\
& \hspace*{1cm}\{ l = ? (1,\ldots,k_1);\\ 
& \hspace*{1cm}\mathsf{if}(l==1)~\{\mathsf{self}!m_{\mathcal{I}_1,\sigma_1}^1;\}~ \mathsf{else}~ \{\} \\
& \hspace*{1cm}\cdots \\
& \hspace*{1cm}\mathsf{if}(l==k_1)~\{\mathsf{self}!m_{\mathcal{I}_1,\sigma_{k_1}}^1;\}~ \mathsf{else}~ \{\}\} \\
&\mathsf{else}~ \{\},\\
& \cdots \\
& \mathsf{if}(l==n)~\\
& \hspace*{1cm}\{ l = ? (1,\ldots,k_n);\\ 
& \hspace*{1cm}\mathsf{if}(l==1)~\{\mathsf{self}!m_{\mathcal{I}_n,\sigma_1}^1;\}~ \mathsf{else}~ \{\} \\
& \hspace*{1cm}\cdots \\
& \hspace*{1cm}\mathsf{if}(l==k_n)~\{\mathsf{self}!m_{\mathcal{I}_n,\sigma_{k_n}}^1;\}~ \mathsf{else}~ \{\}\} \\
&\mathsf{else}~ \{\}\rangle\\
b_{\mathcal{I}_j,\sigma_i}^g = & \langle \sigma_{i_g}^{\mathcal{I}_j},\mathsf{self}!m_{\mathcal{I}_j,\sigma_i}^{g+1};\rangle~~\mbox{where $1\le g<{\it len}(\sigma_i^{\mathcal{I}_j})$}\\
b_{\mathcal{I}_j,\sigma_i}^g = & \langle \sigma_{i_g}^{\mathcal{I}_j}\rangle~~\mbox{where $ g=={\it len}(\sigma_i^{\mathcal{I}_j})$}
\end{split}
\end{align*}where $\sigma_{i_g}^{\mathcal{I}_j}$ denotes the $g^{\it th}$ send statement of the sequence $\sigma_{i}^{\mathcal{I}_j}$, the
variable $l$ is the state variable and the capacity of $\IntM$'s mailbox is set to $\infty$.
\end{definition}

%Such details could have been used when $A_{\IntM}$ is generated by substituting non-determinism of Definition \ref{def:IntM} by conditions on sending a message. 

We remark that the modeler can set the mailbox capacity as it fits to its resources and model such that no deadlock occurs due to a mailbox overflow. Furthermore, the outermost non-deterministic assignment to $l$ leads to an arbitrarily selection of the response $\mathcal{I}_j$ to $m$ while the innermost non-deterministic assignment results in a random selection of $\sigma_i^{\mathcal{I}_j}$, an interleaving of the messages to the actors defined by $\mathcal{I}_j$. The messages of the sequence $\sigma_{i}^{\mathcal{I}_j}$ are sent one-by-one by using the $\mathsf{ self}$ keyword to make our $A_\IntM$ as general as possible (A given $M$ may send multiple of such messages together). The outermost non-deterministic assignment to $l$ can be removed if there exists only one pair $(m,\mathcal{I})$ for the message $m$ in $\Info$ while the innermost non-deterministic assignment to $l$ can be eliminated if only one actor of $O$ expects to receive from $M$ in response to processing $m$. For brevity, the statement  $\mathsf{self}!m_{\mathcal{I}_j,\sigma_i}^{1};$ in the body of $m$ can be replaced by the body of the method $m_{\mathcal{I}_j,\sigma_i}^{1}$. 
%$ \langle \sigma_{i_g}^{\mathcal{I}_j}\rangle$ if $len(\sigma_{i}^{\mathcal{I}_j})=1$. In other words: if $O$ expects to only receive one probable  message in response to $m$ this message can be sent directly in the body of $m$ without any usage of ${\mathsf {self}}$ statements.}
%the indexes of pairs in the set $S_j = \{(r_i^{S_j},c^{S_j}_i)|~ 1 \leq i \leq k_j\})$ are not of importance as two consecutive send actions to different actors are received in any order. 

\begin{example}For the open system of Figure \ref{fig:exampleOpenSystem}, the generated actor $A_{\IntM}$ for $\Info$ of Example \ref{EX::Info} is shown in Figure \ref{fig:IntMOfexampleOpenSystem}. Upon taking a $\mathit{reqL/R}$ message from the mailbox and executing the corresponding method, a $\mathit{permitL/R}$ message is sent to $O$ in response and after executing a $\mathit{release}$ message no response message is sent.
\end{example}

\begin{figure}[h]
    \centering
    \begin{lstlisting}[language=palang, multicols=2]
actor mutex(4){

    reqL{
        left!permitL;
    }

    reqR{
        right!permitR;
    }

    release{}
}

    \end{lstlisting}
    \caption{The overapproximated $\IntM$ for 
    $\Info$ of Example \ref{EX::Info}}
    \label{fig:IntMOfexampleOpenSystem}
\end{figure}

The declarative specification of $\Info$ is so general that it can be derived independent of the details in software documents with the least effort. However, it can be enriched %in two respects: either an ordering is defined between the messages sent to different actors of $O$ in response to $m$, or it can be extended to be 
by being stateful to support cases where the expected messages vary as the component interactions evolve. %The first is in contrast with the asynchrony setting while the second 
This makes it possible to construct a  more precise $A_\IntM$ by replacing its non-deterministic behavior with conditions on sending a message sequence. Such an  enrichment does not affect the generality and even the performance of our assumption generation approach. We inspect the effect of a more accurate $\Info$ on our approach in Section \ref{sec:preciseIntM}. Furthermore, it complicates the process of compliance checking due to the involvement of variables. The more accurate $\Info$, requires more information to be provided or derived from software documents. It is only beneficial in the process of checking a given $M$ against an assumption as the set of $M$s complying with $\Info$ becomes more limited.

%However, it can be extended to be stateful to support cases where the expected messages vary as the component interactions evolve. Such extensions make it possible to construct a  more precise $A_\IntM$ by replacing its non-deterministic behavior with the conditions on sending a message sequence.  
%Such modifications as to make $\Info$ more precise do not affect the generality of our approach but may improve its performance.  We inspect the effect of  a more accurate $\Info$ on our approach in Section \ref{sec:preciseIntM}. 

\subsection{Property LTS Semantics}\label{sec:propertyLTS}
To generate a valid assumption,  those traces of the whole system which violate the property are eliminated. Traditionally this is done by computing the composition of the semantic model of the specified parts of the system and the property which is inefficient and sometimes even unfeasible in  asynchronous environments due to inadequacy  to converge (see Section \ref{sec:genM} which explains how
the LTS semantic of an open actor-based component is generated). 
%The generated state space is called as it reflects traces that violate the property.
Instead, we introduce a novel composition operator which computes such error traces efficiently during the derivation of the semantic model. Intuitively, this operator generates the states of the system only as long as the property is not violated.
To this end, the state space of the system, i.e., $O$ and the interface of $M$, is generated in the same way of the AML semantics while we follow $P_{\mathit{err}}$ with each action as long as the error state has not been visited. The generated LTS is called the property LTS.

\begin{definition}\label{Def::propertyLTS}
Given $P_{\mathit{err}}=( Q,\Sigma,\partial,q_0,\{\pi\})$  and  the actor model $\mathcal{M}$, where $\mathit{TS}(\mathcal{M})=(S,\Act_c,\\\Rightarrow,s_0)$,
we produce the property LTS $\mathit{TS}(\mathcal{M}\otimes P_{\it err})=(S \times Q, \Act_c ,\Rightarrow_p, s_{0} \times q_{0})$  such that the transition relation $\Rightarrow_p$ is the least relation
derived by the rules of Table \ref{Tab::SOS} except ${\it Take}$ together with the rules of Table \ref{Tab::SOSP}.

\begin{table}[htbp]
    \centering
    \caption{The natural semantic rules of $\mathcal{M}\otimes P_{\it err}$. }\label{Tab::SOSP}
    \begin{tabular}{lc}
        %\hline\\
        ${\it Exe}$ & $\Rule
        {s\xrightarrow{a}s'\,\wedge\,q\neq \pi\,\wedge\, ((a=\epsilon\wedge q'=q) \vee (a=\Snd(m)::y\,\wedge \,q'=\partial(q,m) ) )}
        {(s,q) \xrightarrow{a}_p (s',q')}$\vspace*{4mm}\\
        ${\it Take}$ & $\Rule
        {\begin{array}{c}
            s(x)=(v,\langle m|T\rangle,\epsilon)\,\wedge\,
            (s'(x)=(v',\mathfrak{q'},\epsilon)\vee q'=\pi)\,\wedge\,\\
            \forall y\in\ID\setminus\{x\}\,(s(y)=(v'',\mathfrak{q''},\epsilon))\,\wedge \\
            (s[x\mapsto(v,T,body(x,m))],q)\xrightarrow{\zeta}_p\!\!^* (s',q')\,\wedge\\((q'=\pi\,\wedge\, \mathfrak{s} = \pi)\vee (q'\neq\pi\,\wedge\, \mathfrak{s}=(s',q')))\end{array}}
        {(s,q)\xRightarrow{(t_x,\zeta)}_p \mathfrak{s}}$\vspace*{4mm}\\
    \end{tabular}
\end{table}

\end{definition}

Each state of the semantics constitutes two elements: the first denotes the state of the actor model while the second presents the state of the property. As explained by the rule $\it Exe$, the part of the actor model is updated by the semantic rules given in Table \ref{Tab::SOS} while the property part of the state is updated using its transition relation $\partial$ whenever a send action is performed. The rule $\it Take$ expresses that the global state changes upon processing a message and afterwards either the error state of the property is visited or the body of the message handler is fully executed. In the former case, the global state is turned into $\pi$ by this rule.  We remark that $\overto{\zeta}\!\!^*_p$ is reflexive and transitive closure of $\overto{a}_p$, defined in the same way of  $\overto{\zeta}\!\!^*$.

A property LTS is called well-formed if the error state $\pi$ does not have any outgoing transition. By construction, only well-formed LTSs are considered here. The set of error traces of a property LTS $T_L$, denoted by $\mathit{errTr}(T_L)$, is the set all traces which lead to the error state $\pi$. The following definition characterizes actor models that satisfy a given property $P$.

\begin{definition}\label{def::satisfaction}

Given the actor model $\mathcal{M}=({\it actors},\sigma_0)$, where ${\it actors} = O\cup \{A_x\}$, we briefly denote  $\mathit{TS}(\mathcal{M}\otimes P_{\it err})$ by $A_x \interleave_P O$. We say that $A_x \interleave O$  satisfies a property $P$  denoted by
$(A_x \interleave O) \models P$ iff
$\mathit{errTr}(
%\mathit{TS}(\mathcal{M}\otimes P_{\it err})
A_x \interleave_P O) 
= \emptyset$.
\end{definition}

\subsection{Generating the Assumption}\label{sec:gencond}
The following demonstrates the steps of the assumption generation for a given safety property $P$, $\Info$, $O$, where  $O=\{A_{o_1},\ldots,A_{o_\imath}\}$ for some $\imath\ge 0$,  and the unspecified component $M$. Thus, we avoid passing them as arguments to  the proceeding auxiliary functions. We start by generating the property LTS and then through several steps of refinement we generate the assumption.

%\begin{enumerate}
    %\item
%\subsection*{step1}
\vspace*{1mm}\noindent\textbf{Step~1:}\label{item:step1}~First  we  use  $\Info$  to  build $A_{\IntM}$ according to Definition \ref{def:IntM}. Then, %we compute the property LTS, i.e., $\mathit{TS}(\mathcal{M}\otimes P_{\it err})$, for the actor model $\mathcal{M}=({\it actors},\sigma_0)$, denoted by the notation $A_{\IntM} \interleave_P O$. % 
we compute the property LTS $A_{\IntM} \interleave_P O$ and %, where   $O = \{A_{o_1},\ldots,A_{o_n}\}$ . 
%Afterwards, we 
apply the renaming operator $\rho_{c_1}$ on it %$A_{\IntM} \interleave_P O$ 
to hide the internal communications of the open system and $\IntM$. Communications of $\IntM$ with itself, introduced as the result of $A_\IntM$ construction, and the communications between the components of $O$ are considered internal. %The internal actions of $M$ are those transitions labeled as $(t_M,\langle Snd(m)::M \rangle)$ which were introduced as the result of $A_\IntM$ construction. 
The renaming function $c_1:\Act_c \rightarrow \Act_c$ maps $(t_x,\langle a_1, \ldots ,a_k \rangle)$ into $(t_x,h({\it set}_\ID(x) ,\langle a_1, \ldots ,a_k \rangle))$ 
%\[
%\begin{aligned}
%c_1(t_x,\langle a_1, a_2,\ldots , a_k \rangle )=
%&\begin{cases}
%\tau & h({\it set}_\ID(x) ,\langle a_1, \ldots ,a_k \rangle)=\epsilon\\
%(t_x,h({\it set}_\ID(x) ,\langle a_1, \ldots ,a_k \rangle)) & \mbox{otherwise}
%\end{cases}
%\end{aligned}
%\]
where ${\it set}_\ID$ symbolically denotes the set of actor identifiers that any communication between its members should be hidden: ${\it set}_\ID(M)=\{M\}$ and $\forall A_y\in O,({\it set}_\ID(y)=\{o_1,\ldots,o_\imath\})$ with the help of  
the auxiliary function $h:2^\ID\times\Act_s^*\rightarrow \Act_s^*$. This function filters out all communications between the components of its first operand occurred in the sequence of actions given by its second operand, and is defined recursively as:
\[
\begin{cases}
h(\wp,\epsilon)=\epsilon\\
h(\wp,\langle \Snd(m)::y|\zeta\rangle)=\langle \Snd(m)::y \rangle \oplus h(\wp, \zeta) & y \notin \wp\\
h(\wp,\langle \Snd(m)::y|\zeta\rangle)=h(\wp,\zeta) & y \in\wp
\end{cases}
\]It is worth noting that the action sequence as a consequence of a take action by $M$ is turned into either $\epsilon$ or a sequence containing a single send action to a component of $O$.

%Due to our asynchronous setting.
Since we are going to eventually check a given actor $A_M$ against the generated assumption,  we reconfigure the generated assumption to be consistent with the LTS generated from a single actor $A_M$ (see Section \ref{sec:generatingM}). Through the actor $\IntM$'s perspective, which has an unknown environment, it may receive any of its messages at any time and one by one. However, due to the coarse execution of a message handler, it may receive the sequence of messages entirely through a single action from $O$. Therefore, the reconfiguration process is managed to make our approach general and independent of the granularity of the semantic implementation. We remark that the coarse semantics is beneficial to avoid unnecessary interleaving of statement executions of concurrent actors.  To this aim, transitions labeled as $(t_y,\langle \Snd(m_1)::M,\ldots,\Snd(m_n)::M\rangle)$, where $A_y\in O$, are transformed into a set of consecutive transitions labeled by $\Rcv(m_i)$ for each item of $\Snd(m_i)::M$ where $i\le n$. % i.e., a send action triggered by $O$. Next such actions are relabeled to the receive actions by $\IntM$. 
%we relabel the
%send actions as a consequence of a take action by $O$'s actors into a set of consecutive actions received by $\IntM$. 
Another reconfiguration is to
turn the actions like $(t_M,\langle \Snd(m)::y \rangle )$ into the
actions like $\Snd(m)::y$ as part of  our goal to generate the assumption as $M$'s perspective in terms of its interactions with an environment.

\begin{example} The actor $\mathit{left}$ takes the message $\it permitL$ and then sends $\it release$ and $\it reqL$ messages to $\it mutex$ consecutively. Hence, $\it mutex$ will receive the messages $\it release$ and $\it reqL$ through the single action $(t_\mathit{left},\langle \Snd(\mathit{release})::M,\Snd(\mathit{reqL})::M \rangle)$.
However, considering $\it mutex$ as a single actor with an unknown environment, it may receive any of these messages at any time but consecutively as the order of messages to the same receiver is preserved in actor models as mentioned earlier. Thus, the single transition with the label $ (t_\mathit{left},\langle \Snd(\mathit{release})::M,\Snd(\mathit{reqL})::M \rangle)$ is turned into two consecutive transitions with the labels
$ \Rcv(\mathit{release})$ and $\Rcv(\mathit{reqL})$ (while a new state is defined to connect these transitions together). Furthermore, the action $(t_\mathit{mutex},\langle\Snd(\mathit{permitL})::\mathit{left}\rangle)$ is renamed into
$\Snd(\mathit{permitL})::\mathit{left}$.
\end{example}

The operator $\psi$ is introduced to handle such reconfigurations. For a given LTS $T_{L_1}=$ $ (S_{L_1},\alpha T_{L_1},R_{L_1},s_{0}^{L_1})$, $\psi(T_{L_1}) =  ( S_{L_2},\alpha T_{L_2},R_{L_2},s_{0}^{L_2} )$ where:\begin{itemize}
\item $s_{0}^{L_2} = s_{0}^{L_1}$
\item $\alpha T_{L_2} \subseteq \Act_r \cup \Act_s$  Where $\Act_r = \{\Rcv(m)|m \in \MName\}$
\item $R_{L_2}$ is defined such that $\forall s_1^{L_1} \xrightarrow{a} s^{L_1}_2 \in R_{L_1}$:
\begin{itemize}
%\item if $a = (t_x,\epsilon)$, then
\item  if $a = (t_y,\langle \Snd(m_1)::M,\ldots,\Snd(m_n)::M \rangle)$ for some $n\ge 1$ and $A_y \in O$, then %\begin{itemize}
%\item 
$s_1^{L_2} \xrightarrow{\Rcv(m_1)} s_{a_{11}}^{L_2}\xrightarrow{\Rcv(m_2)} s_{a_{12}}^{L_2},\ldots, s_{a_{1(n-1)}}^{L_2}\xrightarrow{\Rcv(m_n)} s_{2}^{L_2}$;
%\end{itemize}
\item  if $a= (t_M,\langle \Snd(m_1)::y\rangle)$, then $s_1^{L_2} \xrightarrow{\Snd(m_1)::y} s_{2}^{L_2}$;
\item if $a= %\tau \,\vee\, a=
(t_x,\epsilon)$, then $ s_1^{L_2} \xrightarrow{\tau} s^{L_2}_2$;
\end{itemize}
\item $S_{L_2}=\{s_i^{L_2}\mid\, s_i^{L_1}\in S_{L_1}\}
\cup \{s_{a_{ij}}^{L_2}\mid (s_i^{L_1},a, s_k^{L_1}) \in R_{L_1}\,\wedge\,a=(t_y,\langle \Snd(m_1)::M, \Snd(m_2)::M,\ldots,\Snd(m_n)::M \rangle)\,\wedge\,n\ge 1\,\wedge\, 0<j< n\,\wedge\, A_y \in O\}.$
\end{itemize}
Applying the operator $\psi$ does not change the semantics at all as the order of actions is preserved. Actions are renamed so that they express the interactions of $M$ with an environment from its own perspective.

%********************
Adopting the first step of \cite{giann}, we minimize this LTS with respect to any equivalence relation that preserves the error traces while it may remove $\tau$ actions. We exploit the weak-trace equivalence relation \cite{GrooteMousavi14}.                    If the error state is not reachable in this LTS, then the property is satisfied for the systems containing actors of $O$ and  any actor $A_M$ that complies with $\Info$. This is reported to the user and the algorithm is terminated.

\begin{figure}
\centering
\begin{tikzpicture}[scale=.7, transform shape]
\node [style=circle,draw,outer sep=0,inner sep=1,minimum size=15] (v1)at (-5,4) {$6$};
\node  [style=circle,draw,outer sep=0,inner sep=1,minimum size=15] (v2)at (-1,4) {$10$};
\node  [style=circle,draw,outer sep=0,inner sep=1,minimum size=15] (v3)at (0.5,2.5) {$12$};
\node [style=circle,draw,outer sep=0,inner sep=1,minimum size=15] (v4) at (-2.1,2) {$8$};
\node  [style=circle,draw,outer sep=0,inner sep=1,minimum size=15] (v5)at (-5,0) {$2$};
\node  [style=circle,draw,outer sep=0,inner sep=1,minimum size=15] (v6)at (-3,-1.5) {$4$};
\node  [style=circle,draw,outer sep=0,inner sep=1,minimum size=15] (v7)at (3,4) {$\pi$};
\node  [style=circle,draw,outer sep=0,inner sep=1,minimum size=15] (v8)at (1.5,1.5) {$11$};
\node  [style=circle,draw,outer sep=0,inner sep=1,minimum size=15] (v9)at (3,0) {$9$};
\node  [style=circle,draw,outer sep=0,inner sep=1,minimum size=15] (v10)at (-2.5,-2) {$3$};
\node  [style=circle,draw,outer sep=0,inner sep=1,minimum size=15] (v11)at (1,-1.1) {$7$};
\node  [style=circle,draw,outer sep=0,inner sep=1,minimum size=15] (v12)at (3,-4) {$5$};
\node  [style=circle,draw,outer sep=0,inner sep=1,minimum size=15] (v13)at (-5,-4) {$0$};
\node  [style=circle,draw,outer sep=0,inner sep=1,minimum size=15] (v14)at (-1,-4) {$1$};
\node[draw,text width=6.2cm](v15) at (8,0){{\large $s_0:~((\epsilon,\langle{\it initialL}\rangle,\langle{\it initialR}\rangle),0)$\\
 $s_{1}:~((\langle{\it reqR}\rangle,\langle{\it initialL}\rangle,\epsilon),0)$\\
 $s_{2}:~((\langle{\it reqL}\rangle,\epsilon,\langle{\it initialR}\rangle),0)$\\
 $s_{3}:~((\langle{\it reqR,reqL}\rangle,\epsilon,\epsilon),0)$\\
 ${s_4}:~((\langle{\it reqL,reqR}\rangle,\epsilon,\epsilon),0)$\\
 $s_5:~((\epsilon,\langle{\it initialL}\rangle,\langle {\it permitR}\rangle),1)$\\
 $s_6:~((\epsilon,\langle{\it permitL}\rangle,\langle {\it initialR}\rangle),1)$\\
 $s_7:~((\langle{ \it release,reqR}\rangle,\langle{\it initialL}\rangle,\epsilon),0)$\\
 $s_8:~((\langle{ \it release,reqL}\rangle,\epsilon,\langle {\it initialR}\rangle),0)$\\
 $s_9:~((\langle { \it reqL}\rangle,\epsilon,\langle {\it permitR}\rangle),1)$\\
 $s_{10}:~((\langle { \it reqR}\rangle,\langle{\it permitL}\rangle,\epsilon),1)$\\
 $s_{11}:~((\langle { \it reqL,release,reqR}\rangle,\epsilon,\epsilon),0)$\\
 $s_{12}:~((\langle { \it reqR,release,reqL}\rangle,\epsilon,\epsilon),0)$\\
}};

\path[ ->,-latex]  (v1) edge   node[above =0.03 cm] {\it Rcv(reqR)} (v2);
\path[ ->,-latex]  (v2) edge[sloped]   node[below=0.03 cm] {\it Rcv(release)} (v3);
\path[ ->,-latex]  (v2) edge   node[above =0.03 cm] {\it Snd(permitR)::right} (v7);
\path[ ->,-latex]  (v3) edge[sloped]  node[above =0.03 cm] {\it Rcv(reqL)} (v10);
\path[ ->,-latex]  (v14) edge [sloped]   node[below =0.03 cm] {\it Rcv(reqL)} (v10);
\path[ ->,-latex]  (v9) edge [sloped]   node[below =0.03 cm] {\it Snd(permitL)::left} (v7);
\path[ ->,-latex]  (v9) edge[sloped]  node [below =0.03 cm]{\it Rcv(release)} (v8);
\path[ ->,-latex]  (v12) edge [sloped]   node[below =0.03 cm] {\it Rcv(reqL)} (v9);
\path[ ->,-latex]  (v14) edge [sloped]   node[below =0.03 cm] {\it Snd(permitR)::right} (v12);
\path[ ->,-latex]  (v12) edge [sloped]   node[below =0.03 cm] {\it Rcv(release)} (v11);
\path[ ->,-latex]  (v11) edge [sloped]   node[below =0.03 cm] {\it Rcv(reqR)} (v14);
\path[ ->,-latex]  (v13) edge [sloped]   node[below =0.03 cm] {\it Rcv(reqR)} (v14);
\path[ ->,-latex]  (v13) edge [sloped]   node[above =0.03 cm] {\it Rcv(reqL)} (v5);
\path[ ->,-latex]  (v8) edge[sloped]   node[below =0.03 cm]{\it Rcv(reqR)} (v6);
\path[ ->,-latex]  (v5) edge [sloped]   node[below =0.03 cm] {\it Rcv(reqR)} (v6);
\path[ ->,-latex]  (v1) edge [sloped]   node[below =0.03 cm] {\it Rcv(release)} (v4);
\path[ ->,-latex]  (v4) edge [sloped]   node[above =0.03 cm] {\it Rcv(reqL)} (v5);
\path[ ->,-latex]  (v5) edge [sloped]   node[above =0.03 cm] {\it Snd(permitL)::left} (v1);
\path[ ->,-latex]  (v6) edge [sloped]   node[below =0.03 cm] {\it Snd(permitL)::left} (v2);
\path[ ->,-latex]  (v10) edge [sloped]   node[above =0.05 cm] {\it Snd(permitR)::right} (v9);
\node (v0)at (-6,-4) {};
\path[ ->,-latex]  (v0) edge (v13);

\node at (3,0) {};
\end{tikzpicture}

\caption{The resulting LTS of step $1$: $s_i :((\mathfrak{q}_{\IntM},\mathfrak{q}_{\it left},\mathfrak{q}_{\it right}),q)$ expresses the mailbox contents for the actors $\IntM$, left,  right, and the state of the $P_{err}$ in the state $i$.}
\label{fig:stepOne}
    \end{figure}
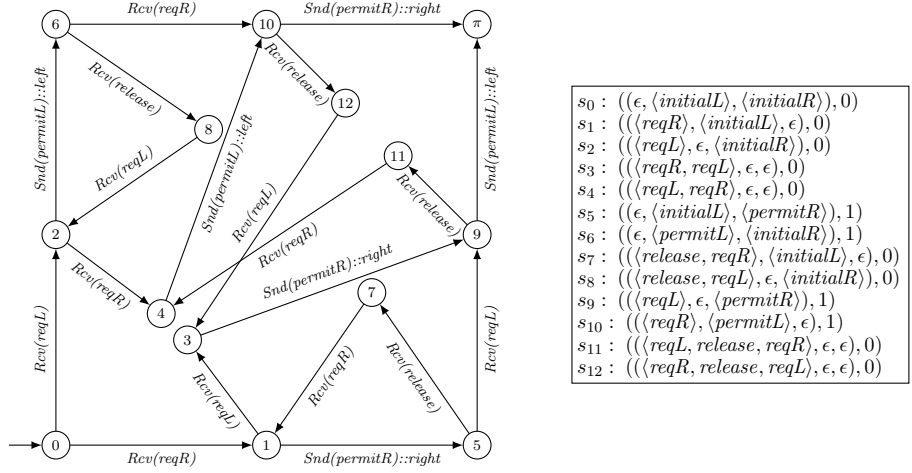

%\item
\vspace*{1mm}\noindent\textbf{Step~2:}~ \label{item:termination} The resulting property LTS of the previous step contains all the traces of the system that violate the property. Referring to the fact that the internal communications of $\IntM$, turned into $\tau$, do not involve in the safety property $P$. Thus, $\tau$ transitions immediately preceding the $\pi$ state are definitely those generated from the communications among $O$'s components (achieved by the computation of the property LTS). Therefore, the second step of \cite{giann} can be safely adopted to propagate the $\pi$ state over $\tau$ transitions in a backward fashion. Consequently, this step prunes those states that $\IntM$ cannot prevent to be entered into the $\pi$ state via $\tau$ steps. If the initial state has become $\pi$ after the backward propagation, it is deduced that by only internal actions of $O$ and possibly of $\IntM$, one can reach the error state from the initial state without any communication between $O$ and $\IntM$. Hence, for any given $M$, the property is never satisfied as the $\pi$ state is reachable only by the internal actions of $O$. Concluding that no $A_M$ can prevent the system from possibly entering the error state. This is reported to the user and the algorithm is terminated. Since we are only interested in error traces, we also remove states that are not backward reachable from the error state.

%\item
\vspace*{1mm}\noindent\textbf{Step~3:}~
We follow \cite{giann}'s third step and make the resulting LTS deterministic and complete with respect to the actions of $\Act_s\cup \Act_r$. We then omit the error state $\pi$ and the transitions leading to it and the assumption, denoted by the LTS $T_A$, is generated. Adding missing receive actions to a state is sound as $O$ will not send such messages and any given $M$ will be executed within the environment $O$.  Furthermore, adding missing send actions is also sound as we consider an additional restriction in accepting an environment (see Section \ref{sec:checkM}): a given $M$ must conform to the information interactions of $\Info$. Thus, adopting this step of \cite{giann} is possible for the asynchronous setting by considering a deterrent step in the process of checking a given $M$ against the generated assumption. %by adding a new \emph{sink} state to
%the LTS, and also adding a transition from each state for each missing transition of the “incomplete” LTS of the previous step to the sink state.

%\end{enumerate}
\begin{example}After applying the first step to our running example, the LTS shown in Figure \ref{fig:stepOne} is achieved. We proceed to the next step as its $\pi$ state is reachable. However, this LTS is not changed by the second step as the actors of the open system and $\IntM$ do not perform any internal actions. Figure \ref{fig:step3} demonstrates the generated assumption
by the third step which indicate that after sending a $\mathit{permitL}$
message, no $\it permitL/R$ message will be sent as long as no
$\it release$ message has been received. %Adding the missing send actions, e.g., $\Snd({\it permitL})::{\it left}$, to the initial state leading to the sink state is inconsequential as $M$ serves as a service actor and cannot send such a message autonomously. Furthermore, 
Adding a missing send action, e.g., $\Snd({\it permitL})::{\it left}$, to the gray state of Figure \ref{fig:step3} is safe, as any given $M$ in compliance with $\Info$, does not produce a response staring with $\Snd({\it permitL})::{\it left}$ upon receiving $\it reqR$. % that is rejected by our compliance checking. Adding missing the receive actions like $\Rcv(relase)$ to the initial state is safe as $O$ will not send such a message to its environment.   
\end{example}
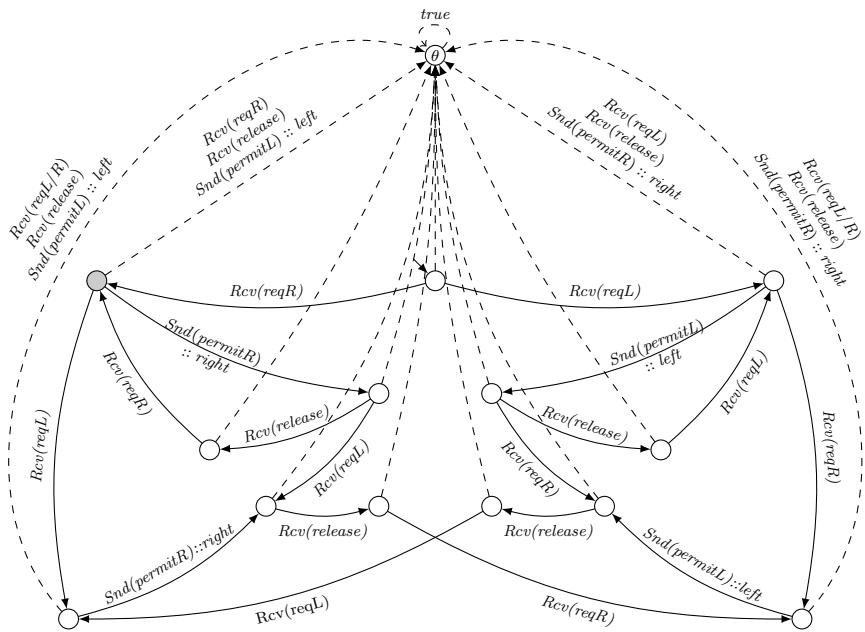
\begin{figure}
    \centering
\begin{tikzpicture}[scale=0.75, transform shape]

\node [style=circle,draw,outer sep=0,inner sep=1,minimum size=10, fill=black!20] (v1)at (-5,4) {};
\node [style=circle,draw,outer sep=0,inner sep=1,minimum size=10] (v2)at (1,4) {};
\node [style=circle,draw,outer sep=0,inner sep=1,minimum size=10] (v3)at (7,4) {};
\node [style=circle,draw,outer sep=0,inner sep=1,minimum size=10] (v4)at (-3,1) {};
\node [style=circle,draw,outer sep=0,inner sep=1,minimum size=10] (v5)at (0,2) {};
\node [style=circle,draw,outer sep=0,inner sep=1,minimum size=10] (v6)at (0,0) {};
\node [style=circle,draw,outer sep=0,inner sep=1,minimum size=10] (v7)at (-2,0) {};
\node [style=circle,draw,outer sep=0,inner sep=1,minimum size=10] (v8) at (-5.5,-2) {};
\node [style=circle,draw,outer sep=0,inner sep=1,minimum size=10] (v9)at (2,0) {};
\node [style=circle,draw,outer sep=0,inner sep=1,minimum size=10] (v10)at (4,0) {};
\node [style=circle,draw,outer sep=0,inner sep=1,minimum size=10] (v11)at (2,2) {};
\node [style=circle,draw,outer sep=0,inner sep=1,minimum size=10] (v12)at (7.5,-2) {};
\node [style=circle,draw,outer sep=0,inner sep=1,minimum size=10] (v13)at (5,1) {};
\node (v0)at (.5,4.5) {};
\node [style=circle,draw,outer sep=0,inner sep=1,minimum size=10] (v14) at (1,8) {$\theta$};

\path[<-,-latex]  (v0) edge (v2);
\path[ ->,-latex]  (v2) edge[bend left=15]   node[left=15,above =0.03 cm] {\it Rcv(reqR)} (v1);
\path[ ->,-latex]  (v2) edge[bend right=15]  node[right=15,above =0.03 cm] {\it Rcv(reqL)} (v3);
\path[ ->,-latex]  (v4) edge [bend left=15,sloped]   node[below =0.03 cm] {\it Rcv(reqR)} (v1);
\path[->,-latex]  (v1) edge [bend right=15,sloped]   node[left=-.7] {$\begin{array}{c}{\it Snd(permitR)}\\{\it ::right}\end{array}$} (v5);
\path[ ->,-latex]  (v5) edge [bend left=15,sloped]   node[left=.15,above] {\it Rcv(release)} (v4);
\path[ ->,-latex]  (v7) edge [bend right=15,sloped]   node[below =0.03 cm] {\it Rcv(release)} (v6);
\path[ ->,-latex]  (v5) edge  [bend left=15]  node[sloped,below]{\it Rcv(reqL)} (v7);
\path[ ->,-latex]  (v6) edge [sloped,bend right=15]   node[below =0.03 cm] {\it Rcv(reqR)} (v12);
\path[ ->,-latex]  (v8) edge [bend right=15,sloped]   node[above =0.03 cm] {\it Snd(permitR)::right} (v7);
\path[ ->,-latex]  (v12) edge   [sloped,bend left=15] node [above =0.03 cm] {\it Snd(permitL)::left} (v10);
\path[ ->,-latex]  (v3) edge [bend left=15,sloped]   node[right=-.7] {$\begin{array}{c}{\it Snd(permitL)}\\{\it ::left}\end{array}$} (v11);
\path[ ->,-latex]  (v1) edge [bend right=15,sloped]   node[above =0.03 cm] {\it Rcv(reqL)} (v8);
\path[ ->,-latex]  (v9) edge  [sloped,bend left=15]node[below =0.03 cm] {Rcv(reqL)} (v8);
\path[ ->,-latex]  (v10) edge [bend left=15] node[below =0.03 cm] {\it Rcv(release)} (v9);
\path[ ->,-latex]  (v11) edge [bend right=15]   node[sloped, below] {\it Rcv(reqR)} (v10);
\path[ ->,-latex]  (v11) edge [sloped,bend right=15]   node[right=.15,above] {\it Rcv(release)} (v13);
\path[ ->,-latex]  (v13) edge [bend right=15,sloped]   node[below =0.03 cm] {\it Rcv(reqL)} (v3);
\path[ ->,-latex]  (v3) edge[bend left=15,sloped]   node[above =0.03 cm] {\it Rcv(reqR)} (v12);

\path[ ->,dashed,-latex]  (v1) edge[sloped]   node[above =0.03 cm] {$\begin{array}{c}{\it Rcv(reqR)}\\{\it Rcv(release)}\\{\it Snd(permitL)::left}\end{array}$} (v14);
\path[ ->,dashed,-latex]  (v2) edge  node[above =0.03 cm] {} (v14);
\path[ ->,dashed,-latex]  (v3) edge [sloped]   node[above =0.03 cm] {$\begin{array}{c}{\it Rcv(reqL)}\\{\it Rcv(release)}\\{\it Snd(permitR)::right}\end{array}$} (v14);
\path[ ->,dashed,-latex]  (v4) edge [sloped,bend right = 7]   node[above =0.03 cm] {} (v14);
\path[ ->,dashed,-latex]  (v5) edge [sloped,bend right = 7]   node[below =0.03 cm] {} (v14);
\path[ ->,dashed,-latex]  (v6) edge [sloped, bend right=7]   node[above =0.03 cm] {} (v14);
\path[ ->,dashed,-latex]  (v7) edge[bend right =18]   node{} (v14);
\path[ ->,dashed,-latex]  (v8) edge [sloped,bend left = 75]   node[above =0.03 cm] {$\begin{array}{c}{\it Rcv(reqL/R)}\\{\it Rcv(release)}\\{\it Snd(permitL)::left}\end{array}$} (v14);
\path[ ->,dashed,-latex]  (v9) edge[sloped,bend left = 7 ]    node[above =0.03 cm] {} (v14);
\path[ ->,dashed,-latex]  (v10) edge[bend left =18]  node {} (v14);
\path[ ->,dashed,-latex]  (v11) edge [sloped,bend left = 7]   node[below =0.03 cm] {} (v14);
\path[ ->,dashed,-latex]  (v12) edge [sloped,bend right =75]   node[above =0.03 cm] {$\begin{array}{c}{\it Rcv(reqL/R)}\\{\it Rcv(release)}\\{\it Snd(permitR)::right}\end{array}$} (v14);
\path[ ->,dashed,-latex]  (v13) edge[bend left =7 ]  node [below =0.03 cm] {} (v14);
\path[ ->,dashed,-latex]  (v14) edge[loop,above ]  node [] {${\it true}$} (v14);
\end{tikzpicture}

    \caption{The generated assumption%: the self-loop with the label ${\it true}$ is an abbreviation for a self-loop of $\theta$ for each action belonging to $\alpha T_A$ 
    	}\label{fig:step3}
 
\end{figure}

\subsection{Checking the Unspecified Component Against the Generated Assumption}\label{sec:genM}

The following demonstrates the steps for checking a given $M$ against an assumption $T_A$, based on the given safety property $P$, $\Info$, and the open system $O$ where  $O=\{A_{o_1},\ldots,A_{o_\imath}\}$. To this aim, we will introduce an algorithm to derive an LTS characterizing  actor $A_M$ in terms of its interactions with an environment, denoted by $T_M$, and then check this LTS against $T_A$.

\subsubsection{Generating the LTS of the Actor $A_M$}\label{sec:generatingM}
%   \begin{enumerate}

Actor $A_M$ interacts with its environment in two ways: either it receives a message or it fetches one from its mailbox and sends its responses. Thus to specify $T_M$, we consider a \emph{wild environment}
for the actor $A_M$ which nondeterministically and continuously sends messages to $A_M$ and upon receiving a message from $A_M$ discards it. We construct the LTS of $A_M$ by composing $A_M$ with this wild environment and reconfiguring the produced LTS  in the same way we reconfigured our assumption. % (see the first step of our algorithm in Section \ref{sec:gencond}).

For the given actor $(M,\mathit{vars_M},\mathit{mtds_M},n)$, we derive the LTS of
 $T_M = \rho_{c_2}(A_M \interleave A_{M^c}  \interleave A_{o_1^d} \interleave \ldots \interleave A_{o_\imath^d})$ where:
\begin{itemize}
\item For each $A_{o_i} \in O$, where $1 \leq i \leq \imath$, a dummy actor $A_{o_i^d}$ is specified by $(o_i,\emptyset,\emptyset,0)$. These actors discard every messages that they receive due to their mailbox capacity of zero.
\item $(M^c,\{l\},\{(m_c,b_c)\},1)$ is the complement actor for $M$ with only one method $m_c$, which non-deterministically sends all the possible messages to $A_M$, with the body $b_c$ defined as:
\begin{align*}
\begin{split}
b_c =\langle~ &
l = ? (1,2,...,k);,\\
& \mathsf{if}(l==1)~  \{ M!m_1;\} ~ \mathsf{else}~ \{\}, \\
%& \mathsf{if}(j==2)~ \mathsf{then}~  \langle M!m_2;\rangle ~ \mathsf{else}~ \langle \Skip ;\rangle, \\
& \ldots\\
& \mathsf{if}(l==k)~ \{ M!m_k;\}~ \mathsf{else}~ \{\}, \\
& \mathsf{self} !m_c;,l=0;\rangle\\
\end{split}
\end{align*}
%where $\{m_1,\ldots,m_k\}$ are the method names of $M$.
where $\mathit{methods(M)}=\{m_1,\ldots,m_k\}$.

\item The main block for this system contains the sequence $\langle M^c!m_c; \rangle$.
\end{itemize}

To minimize the resulting state space
of $T_M$, the state variable $l$ of $M^c$ is reset at the end of the method $m_c$ to restrict state evolutions of $T_M$ to the changes of the state variables and mailbox of $M$. Next, we rename the actions of $A_M \interleave A_{M^c}  \interleave A_{o_1^d} \interleave \ldots \interleave A_{o_\imath^d}$ to depict $M$ as an open entity. The labels of the resulting LTS denote either $M$ or $M^c$ has taken a message and consequently a sequence of send actions has been performed. The take information is removed from the labels. The send actions as the consequence of a message taken by $M^c$ are renamed to depict $M$ receiving such %as the recipient of the 
messages with the application of the renaming operator $\rho_{c_2}$ where the renaming function
  $c_2: \Act_T \times \Act_s^* \rightarrow \Act_r \cup \Act_s^*$  is defined as:
 ‌‌‌‌ \[
  \begin{cases}
  c_2((t_M,w)) = w &w \in \Act_s^* \\
  c_2((t_{M^c},\langle \Snd(m)::M,\Snd(m_c)::M^c\rangle)) =\Rcv(m) \\
  \end{cases}
  \]

\begin{figure}
    \centering
\begin{lrbox}{\mylistingbox}%
\begin{minipage}{.45\linewidth}%
        \begin{lstlisting}[language=palang]
    actor mutex(3) {

        int taken;
        reqL{
            if(taken==0){
                left!permitL;
                taken=1;
            }else{
	            self!reqL;
            }
        }

        reqR {
            if(taken==0){
                right!permitR;
                taken=1;
            }else{
	            self!reqR;
            }
        }

        release{
	        taken=0;
	    }
        
    }
\end{lstlisting} %
\end{minipage}%
\end{lrbox}%
\subfloat[Previously Unavailable Actor $A_M$\label{fig:actorM}]{\usebox{\mylistingbox}}%
\hfill
\begin{lrbox}{\mylistingbox}%
\begin{minipage}{.45\linewidth}%
        \begin{lstlisting}[escapeinside={(*}{*)}][language=palang]
        actor (*$\mathit{mutex}^c$*)(1) {

            int l;
            
            
            (*$m_c$*) {
                l = ?(1,2,3);
                if(l==1)
	                {mutex!reqL;} 
                 else{}
                if(l==2) 
	                {mutex!reqR;} 
                 else{}
                if(l==3) 
	                {mutex!release;} 
                 else{}
                self!(*$m_c$*);
                l=0;
            }
        }
        \end{lstlisting} %
\end{minipage}%
\end{lrbox}%
\subfloat[Actor $A_{\mathit{mutex}^c}$\label{fig:actorMc}]{\usebox{\mylistingbox}}%
\caption{Actor $A_M$ and its corresponding $A_{\mathit{mutex}^c}$}\label{Fig::Actors}
\end{figure}

\begin{example} 
Figure \ref{fig:actorM} depicts  the previously unavailable actor $A_M$ for our ongoing example. To generate $T_M$,  $\mathit{mutex} \interleave \mathit{left}^d \interleave \mathit{right}^d \interleave \mathit{mutex}^c$  is constructed where $\mathit{left}^d$ and $\mathit{right}^d$ are specified by  $(\mathit{left},\emptyset,\emptyset,0)$ and $(\mathit{right},\emptyset,\emptyset,0)$, respectively, and the actor $\mathit{mutex}^c$ is depicted in Figure \ref{fig:actorMc}. By apply  the renaming operator $\rho_{c_2}$, the action $(t_{\mathit{mutex}^c} , \langle \Snd(\mathit{release})::mutex,\Snd(m_c)::\mathit{mutex}^c \rangle)$ is renamed to $\Rcv(\mathit{release})$ while
$(t_{\mathit{mutex}},\langle\Snd(\mathit{permitL})::\mathit{left}\rangle)$ is renamed to $\langle\Snd(\mathit{permitL})::\mathit{left}\rangle$.
\end{example}

Our approach for establishing the LTS characterization of an actor in terms of its interactions with an environment is consistent with the one proposed for $\pi$-calculus specifications embedding the actor model \cite{MayActor}, where channel names indicate actor names. This approach is based on the concept of \emph{Interfaces} that maintain actor names to/by which a message can be sent/received, called external actors and receptionists, respectively. Intuitively, the receptionists are those actors of the model that are not hidden from the environment while the external actors are part of the environment. The environment is dynamic due to the communication of names over the channels.  Since the communicated names cannot be derived from the specification of the model, Interfaces are exploited. To generate LTSs, a set of semantic rules are provided:
\begin{itemize}
\item IN: any message targeted to one of receptionists can be received asynchronously while being added to the pool of messages. 
\item OUT: only messages targeted to an external actor can be sent. 
\item TAU: communications between actors in the model are considered as internal.
\end{itemize} 
In our case, the environment is static, $A_{o^i_d}$s play the role of the external actors and $A_M$ is the sole receptionist. The non-deterministic behavior of $A_{M^c}$ imitates the semantic rule IN. By using the appropriate renaming functions $\rho_{c_2}$ and $\rho_{c_3}$ (introduced in the following), internal communications are turned into $\tau$ (implementing the rule TAU), and only messages targeted to the external actors are preserved as expressed by OUT.

\subsubsection{Checking the LTS of a Component Against an Assumption}\label{sec:checkM}
Due to the openness of actors, various notions of equivalence exist based on testing equivalence \cite{AghaMST97,MayActor} to compare their behaviors. Both trace and may-testing equivalences \cite{MayTesting} are appropriate for reasoning about safety properties \cite{MayAsynch}. However, a trace-based characterization of may-testing such as \cite{MayAsynch,MayActor} alleviates reasoning, as it does not involve quantification over observing contexts required by may-testing. Such a characterization captures the essence of asynchrony by permitting receive and send actions to be delayed through the following rules, adapted to our setting:\begin{enumerate}
\item Receive actions can be commuted.
\item Two send actions to two different receptionists can be commuted.   
\item A send and a receive action can be commuted.
\end{enumerate} 
Therefore, traces (of the LTSs characterizing actor models) are partially compared, e.g., by matching $\langle \Rcv(m_1),\Snd(m_2)::o_1,\Snd(m_3)::o_2\rangle$ to $\langle \Snd(m_2)::o_1,\Rcv(m_1),\Snd(m_3)::o_2\rangle$ where $o_1\neq o_2$. Such rules are not enforced through the operational semantics of our modeling framework, but are naturally achieved by our approach of generating $T_M$ and the assumption. The first rule is a consequence of the non-deterministic behavior of $A_{M^c}$ and the completion step of our assumption generation algorithm. The second rule is the result of shuffling messages during the construction of $\IntM$, and finally the third is the result of non-deterministic execution of actors in our semantics (i.e., $M$ and $O$ in generating the property LTS and $M$ and $M^c$ in generating the LTS characterization). Therefore, our LTS characterization considers message reordering and hence, trace-equivalence can be used to compare actors. As we explained earlier, we would have got a considerable practical gain if the effect of the asynchronous setting on reordering of messages had been handled semantically through an equivalence relation based on may-testing instead of explicit shuffling of the messages in the specification (in a message body of $\IntM$). Due to the unavailability of a practical tool to partially compare the traces in our experiments to validate our claims, we were confined to exploit trace equivalence.

To check a given actor $A_M$ against the generated assumption, two steps should be followed: first its compliance with $\Info$ is checked and them its LTS characterization is compared to the assumption.

\vspace*{2mm}\noindent\textbf{Step 1}.\ As a precondition to composition of any given actor $A_M$  with the actors of $O$, $A_M$ should comply with $\Info$ to  respect the prerequisites of $O$. If $A_M$ does not comply with $\Info$, we abort and notify the modeler.

The reason for this check lies behind  our primary assumption which led to a significant reduction in the size of the generated assumption. In more detail by using $\Info$ to generate $A_{\IntM}$, we restricted  $A_M$'s behavior to what is mentioned in $\Info$ and therefore our assumption will not generate a sound result for any $A_M$ with a behavior beyond what has been assumed. For example, if a given  $A_M$  sends the messages $\it permitL$ and $\it permitR$ in response to a message $\it reqL$, it has violated the safety property $P$. However, in the generated assumption of Figure \ref{fig:step3}, the corresponding  trace is a valid one leading to the sink state and therefore, this violation will not be captured if the compliance of $A_M$ is not checked first.

To prevent false positive results, we check that $A_M$ respects the prerequisite of $O$ as declared by $\Info$ by computing its responses to each message $m$. However, $M$ can not be considered in isolation as  %the behavior of $M$ should be inspected in composition with $O$ as 
the responses of $M$ in processing $m$ may be determined in terms of the behavior of $O$. To clear the case, consider Example \ref{Ex::comply}. Additionally to compute a response, messages sent by $M$ to itself  as a consequence of processing $m$ should be also considered as part of the computation.  In other words, messages sent to $O$ as a consequence of such internal communications %(and also their consequent internal communications) 
should be included in the responses of $M$ to $m$. As we consider components with bounded channels, responses are finite.
\begin{example}\label{Ex::comply}
Consider the actor $\it mutex$ in Figure \ref{fig:actorM} that upon processing the message $\it reqL$, it may send $\it permitL$ to $O$ or $\it reqL$ to itself depending on the value of $\it taken$. Assume this actor is composed with a wrongly implemented $O$ such that its $\it right$ actor does not send $\it release$ after receiving $\it permitR$. Then, it does not send $\it permitL$ in response to   $\it reqL$ preceding a $\it reqR$ by recursively resending $\it reqL$ to itself. Therefore, this actor generates either no response or ${\it permitL}::{\it left}$ upon processing $\it reqL$. In contract, if it is composed with $O$ specified by Figure \ref{fig:exampleOpenSystem}, it only generates the response ${\it permitL}::{\it left}$.
\end{example}
Noting to the fact that the responses of $M$ upon processing $m$ differ depending on $O$ with which it has been composed, we define the notion of compliance. To this aim, we use the notation $\zeta\downharpoonleft y$ to extract the messages sent to the actor $A_y$ from the send action sequence $\zeta$, defined by the equations:\[
\begin{array}{l} \epsilon\downharpoonleft y=\epsilon\\ \langle \Snd(m)::y|\zeta\rangle \downharpoonleft y = \langle m\rangle \oplus \zeta\downharpoonleft y\\ \langle \Snd(m)::x|\zeta\rangle \downharpoonleft y = \zeta\downharpoonleft y,~ \mbox{$x\neq y$}\end{array}\] 

\begin{definition}\label{def:complianceNew}
For the given actor $A_M$ and $\Info$, % = (M,\mathit{vars},\mathit{mtds},n)$ and the set of $\Info$, %assume that actions $(t_M,\zeta_1)$, $\ldots$, and $(t_M,\zeta_n)$ are caused by first processing $m$ and its consequent internal communications in a trace of the system $A_M\interleave O$.
we say $A_M$ complies with $\Info$ (denoted by $A_M \sqsubseteq_c \Info$) if for any  set of actions $(t_M,\zeta_1)$, $\ldots$, and $(t_M,\zeta_n)$ caused by first processing $m$ and its consequent internal communications in any trace of the system $A_M\interleave O$, 
%generated due to processing the message $m$ by $M$ in the system $A_M\interleave O$,  
$\exists (m,\mathcal{I}) \in \Info$ such that $\forall\, A_y\in O\,((\zeta\downharpoonleft y,y)\in \mathcal{I})$, where $\zeta=\langle\zeta_1\rangle\oplus\ldots\oplus\langle \zeta_n\rangle $.
\end{definition}

To avoid computation of $A_M\interleave O$ and decide the conformance of $M$ in isolation, we propose a sound approach to approximate the responses of a given $A_M$. To this aim, we execute $M$ as a standalone actor to first generate its \emph{message flow graph}, for each method $m$ of $A_M$. This graph, denoted by ${\it MF}(m)$, is achieved from the modulo weak trace equivalence of a special LTS  %$(\Statement^*\times\*,\Act_s,\rightharpoonup,\Statement)$, 
%that its transition relation is 
derived from the rules of Table \ref{Tab::SOS4compliance}. Intuitively by processing the message $m$, we compute its consequent interactions with $O$. Each state of the mentioned LTS is a pair like $(\sigma,\mathfrak{q})$ where $\sigma$ and $\mathfrak{q}$ denote the sequence of statements to be executed and the queue of messages to be handled as a consequence of processing $m$. The initial state is defined by $(\epsilon,\langle m\rangle)$. When there is no statement to be executed, a message is consumed from the queue by the rule $\it Inv$ to execute its respective body. Rules ${\it NoIntr}_{1-3}$ express that the assignment, non-deterministic assignment, and conditional statements do not make any interactions with $O$. Furthermore, no interaction with $O$ happens upon sending a message to itself, but the message is added to the queue as explained by the rule ${\it Intr}_{\it intern}$. However, an interaction with an actor of $O$ occurs by executing a send statement as explained by the rule $\it Intr$. Therefore, the transitions of the resulting LTS carry the labels of $\Act_s\cup\{\tau\}$. We remark that the effect of variables has been abstracted away as a consequence of rules ${\it NoIntr}_{2,3}$.

\begin{table}[htbp]
\large
       \centering
    \caption{Semantic rules to derive the special LTS for generating ${\it FM}(m)$ }\label{Tab::SOS4compliance}
     \begin{adjustbox}{max width=\textwidth}
    \begin{tabular}{c}
        %\hline\\
${\it NoIntr}_1:~\FAxiom{(\langle {\it stat}|\sigma \rangle,\mathfrak{q})\xrightharpoonup{\tau}(\sigma,\mathfrak{q})}$, where ${\it stat}\in\{{\it var}:={\it expr},{\it var}:=?(e_1,e_2,..,e_n)\}$\vspace*{2mm}\\
${\it Intr}_{\it intern}:~\FAxiom
{(\langle \mathsf{self}!m'|\sigma \rangle,\mathfrak{q}\oplus\langle m'\rangle)\xrightharpoonup{ \tau}(\sigma,\mathfrak{q})}$~~~~
${\it Intr}:~\FAxiom
{(\langle y!m'|\sigma \rangle,\mathfrak{q})\xrightharpoonup{\Snd(m')::y}(\sigma,\mathfrak{q})}$\vspace*{2mm}\\
${\it NoIntr}_2:~\FAxiom{(\langle \mathsf{if}~(expr)~\{\sigma_1\}~ \mathsf{else}~\{\sigma_2\}
|\sigma\rangle,\mathfrak{q})\xrightharpoonup{ \tau}(\sigma_1\oplus\sigma,\mathfrak{q})}$\vspace*{2mm}\\
${\it NoIntr}_3:~\FAxiom{(\langle \mathsf{if}~(expr)~ \{\sigma_1\}~ \mathsf{else}~\{\sigma_2\}|\sigma\rangle,\mathfrak{q})\xrightharpoonup{ \tau}(\sigma_2\oplus\sigma,\mathfrak{q})}$~~~~~~~${\it Inv}:~\FAxiom
{(\epsilon,\langle m|\mathfrak{q} \rangle)\xrightharpoonup{\tau}(\body(M,m),\mathfrak{q})}$\vspace*{4mm}\\
\end{tabular}
 \end{adjustbox}
\end{table}

\begin{example}Consider the example $A_M$ in Figure \ref{fig:actorM}. The resulting LTS for the method $\it reqL$ of this actor is depicted in Figure \ref{fig:LTSofFG} while the states in a weak-trace equivalence class have been represented with the same color. The resulting message flow graph contains two states: the initial $(\epsilon,\langle {\it reqL}\rangle)$ and the final state $(\epsilon,\epsilon)$ with one transition, labeled by $\mathit{permitL}::\mathit{left}$, connecting them. In the initial state by the rule $\it Inv$, the message $\it reqL$ is taken from the queue while its method body is loaded, and the state $(\mathsf{if}(taken==0)~\{ \mathit{left}!\mathit{permitL};taken=1;\}~ \mathsf{else}~ \{\mathsf{self}!\mathit{reqL;}\}\rangle,\epsilon))$ is generated. Due to abstracting variables, two next states are generated by application of rules ${\it NoIntr}_{2,3}$. In the state $(\langle \mathsf{self}!\mathit{reqL;}\rangle,\epsilon)$, the message $\it reqL$ is inserted to the queue by application of rule ${\it Intr}_{\it intern}$ while returning back to the initial state.% %
\end{example}

Afterwards, we extract the finite traces of ${\it MF}(m)$ which  are the  sequences of send actions that are possibly sent to $O$'s actors as the result of  processing the message $m$. Such traces are those from the initial state $(\epsilon,\langle m\rangle)$ to the end state $(\epsilon,\epsilon)$ %from which a well-formed $\Info$ can be extracted and are 
computed through a backward analysis by calling the procedure $\it Backward$ given in Figure \ref{Fig::trace}. If the message interactions derivable from such traces are included in $\Info$, then we say $A_M$ compiles with $\Info$. 

%In order to extract these traces we use the method {\it backwards}.
%${\it FG}(\epsilon,\langle m \rangle)$ by calling the method  where 

The backward analysis starts from the end state $(\epsilon,\epsilon)$ %, which indicates that no statement exists to be executed and no message is in the queue of the actor $M$, 
in a depth-first fashion over the transitions of ${\it MF}(m)$ while the labels over the transition are concatenated. The formal parameter ${\it TS}$ of $\it Backward$ denotes the graph over which the backward analysis is carried on,  $s$ the current state that has been reached, $\zeta$ the most recent computed trace, % by concatenating the labels, 
and ${\it result}$ the set of found traces. By traversing $(s', a ,s)\in \rightharpoonup$ such that $a\neq \tau$, we check that the number of sent messages to each actor of $O$ does not exceed its mailbox capacity, and no overflow occurs. This is done by first updating $\zeta$ to $\langle a \rangle\oplus \zeta$ and then calling ${\it isWellformed}(\zeta)$.  %We treat $\tau$ labels as $\epsilon$ during concatenations. are 
%This function checks that the number of messages has been send to each actor of $O$ does not exceed its mailbox capacity, and no overflow occurs.  
This validation is in line with the well-formedness property of $\Info$ as defined in Section \ref{sec:OpenSystem}. When an overflow is detected, the algorithm is immediately terminated by returning $({\it true},{\it result})$.   A trace is added to ${\it result}$ when the initial state is reached.  By recursive invocation of $\it Backward$ %while the label of the most recently traversed transition is added to $\zeta$, 
we proceed towards the initial state. The return value of this recursive call is used to update $\it overflow$ and $\it result$ where the local variable $\it overflow$ accumulates the overflow validation results of all invocations. %The algorithm is terminated immediately when $\it overflow$ becomes  

\begin{figure}[H]
    \centering
%\begin{lrbox}{\mylistingbox}%
\begin{minipage}{.55\linewidth}%
\begin{algorithmic}
\Function{${\it Backward}$}{${\it TS},s,\zeta,{\it result}$}
\State \textbf{bool} ${\it overflow}\,:={\it false}$
\ForAll{$(s', a ,s)\in \rightharpoonup$}
 \If {$(a\neq \tau)$}
 	 \State $\zeta:=\langle a \rangle \oplus
 \zeta$
  \If  {$\neg {\it isWellformed}(\zeta)$}
  \State \Return $({\it true},{\it result})$
  \EndIf 
 \EndIf
 \If {$(s'==(\epsilon,\langle m\rangle))$}
% \Comment{$s'$ is the initial state}
 \State ${\it result}={\it result} \cup \{ \zeta\}$ 
 \EndIf
 \State $({\it temp}_o,{\it temp}_r) := Backward({\it TS},s',\zeta, {\it result})$
 \State ${\it overflow}\,:={\it overflow}\vee {\it temp}_o~$
\If {$\it overflow$} \textbf{break}
\EndIf
\State ${\it result}\,:= {\it result}\cup {\it temp}_r$
\EndFor
\State \Return $({\it overflow},{\it result})$
\EndFunction
\end{algorithmic} %
\caption{Computing finite traces\label{Fig::trace}}
\end{minipage}%
%\end{lrbox}%
%
%\subfloat[Computing finite traces\label{Fig::trace}]{\usebox{\mylistingbox}}%
%
\hfill
%
%\begin{lrbox}{\mylistingbox}%
\begin{minipage}{.45\linewidth}%
 \begin{tikzpicture}[scale=.75, transform shape]
\node[draw](v1) at (4.5,4) {$(\epsilon,\langle {\it reqL}\rangle)$};
\node[draw] (v2) at (4.5,2.5) {$\begin{array}{l}\langle \mathsf{if}(taken==0)~\\~~~~\{ \mathit{left}!\mathit{permitL};taken=1;\}\\~ \mathsf{else}~ \{\mathsf{self}!\mathit{reqL;}\}\rangle,\epsilon)\end{array}$};
\node[draw](v3) at (3.3,0.5) {$\begin{array}{l}(\langle \mathit{left}!\mathit{permitL};,\\ taken=1;\rangle,\epsilon)\end{array}$};
\node[draw] (v4)at (6.5,0.5) {$(\langle \mathsf{self}!\mathit{reqL;}\rangle,\epsilon)$};
\node[draw] (v5) at (3.3,-1) {$(\langle taken=1;\rangle,\epsilon)$};  
\node[draw,fill=black!10] (v6)at (3.3,-2.1) {$(\epsilon,\epsilon)$};
\node (v7)at (4,4.7) {};

\path[ ->,-latex]  (v1) edge []   node[left=-0.05cm] {$\tau$} (v2);
\path[ ->,-latex]  (v2) edge []   node[above=0.06cm] {$\tau$} (v4);
\path[ ->,-latex]  (v2) edge []   node[left] {$\tau$} (v3);
\path[ ->,-latex]  (v3) edge [right]   node[] {$\mathit{permitL}::\mathit{left}$} (v5);
\path[ ->,-latex]  (v5) edge []   node[left=-0.05cm] {$\tau$} (v6);
\path[ ->,-latex]  (v4) edge [bend right=90]  node[above=0.3cm] {$\tau$} (v1);
\path[ ->,-latex]  (v7) edge []  node[] {} (v1);
\end{tikzpicture} %
\caption{The generated LTS for the method $\it reqL${\label{fig:LTSofFG}}}
\end{minipage}%
%\end{lrbox}%
%
%\subfloat[Actor $A_{\mathit{mutex}^c}$\label{fig:actorMc}]{\usebox{\mylistingbox}}%

\end{figure}

\begin{lemma}\label{Lem:Soundcompliance}
For a given actor $A_M$ specified by $ (M,\mathit{vars},\mathit{mtds},n)$ and the set of $\Info$, $A_M \sqsubseteq_c \Info$ holds  if $({\it overflow},{\it result})= {\it Backward}({\it MF}(m),(\epsilon,\epsilon),\epsilon,\emptyset)$ implies that $\neg {\it overflow}$ and
\[
\forall\,\zeta\in {\it result},\,\exists \, (m,\mathcal{I})\in\Info, \,\forall\, A_y\in O\,((\zeta\downharpoonleft y,y)\in \mathcal{I})
\]
\end{lemma}

\begin{proof}
Due to the abstraction of variable values, all control flows of a method are considered in $\it result$. Therefore, the real responses of $M$ in the system $A_M\interleave O$ are trivially a subset of ${\it  result}$. 
\end{proof}
If a given $A_M$ was rejected as there exists $\zeta\in {\it result}$ such that  $\exists\, A_y\in O\,((\zeta\downharpoonleft y,y)\not\in \mathcal{I})$, then the validity of $\zeta$ (to avoid false negative result) can be checked by using the techniques of concolic testing \cite{concolic}. Intuitively, $\zeta$ is used to compute a sequence of symbolic constraints corresponding each conditional statement within a method. The conjunction of these constraints is called a path constraint. With the help of symbolic analyzers, we can find the messages that should have been previously sent by $O$ to validate the path constraint. %Our technique can be improved by considering the variable values and resolving conditional statements deterministically when their condition value is determined. 
Intuitively, our approach works for the class of actors that $O$ 
exerts all the paths of ${\it MF}(m)$. Characterizing such a class and providing an approach to validate a response are among our future work. 
%behaviors depend on neither the history of the actor before executing $m$ nor the behavior of $O$.

We remark that $\Info$ can be provided as an LTS itself to have a more homogeneous approach to check the compliance. However, this makes the assumption generation process more inefficient as we have to first derive the LTS of $O$, second the composition of $O$ and $\Info$, and finally the property LTS. The declarative approach makes it possible to derive $\IntM$ as an actor and consequently benefit from the special composition for computing the property LTS (without computing the intermediate LTSs).

\vspace*{2mm}\noindent\textbf{Step 2}. We check  if $\mathit{Tr}(\rho_{c_3}(T_M)) \subseteq \mathit{Tr}(T_A)$. If that is the case, then $(A_M \interleave O) \models P$ where
the renaming function
$c_3:\Act_r \cup \Act_s^*\rightarrow \Act_r \cup \Act_s^*$
hides the internal send actions of $M$:

\begin{align*}
&\begin{cases}
c_3(\Rcv(m))=\Rcv(m)\\
c_3(\langle \Snd(m)::y|\zeta\rangle)= \langle \Snd(m)::y\rangle \oplus c_3( \zeta)& \Snd(m)::y\in \alpha T_A\\
c_3(\langle \Snd(m)::y|\zeta\rangle)=c_3(\zeta)& \Snd(m)::y\notin \alpha T_A\\
c_3(\epsilon)=\epsilon\\
\end{cases}
\end{align*}Intuitively, the observable actions of $M$ are defined by the actions of $T_A$, constructed based on $\Info$. If $M$ sends a message, not included in $\alpha T_A$, it should be those that $M$ has sent to itself and the renaming function $c_3$ hides them accordingly. %the former actions while the latter ones have been already prohibited in the first step. 
By contrast in \cite{giann}, the traces that contain an internal action of $M$, lead to the sink state and hence, are automatically accepted as valid traces although the property may be violated afterwards. However, it is appealing that the system still satisfies the property despite these internal actions. Therefore, we hide any action that is internal to $M$. For example, the action with the label $\langle \Snd(\mathit{reqR})::\mathit{mutex}\rangle$ (as a consequence of $\mathsf{self}!{\it reqR}$ in Figure \ref{fig:actorM}) is renamed to $\tau$ after performing $\rho_{c_3}$.

\subsection{Correctness of Approach}
Our approach generates the assumption  $T_A$
such that it
% where $\rho_{c_4}(T_A)$
  characterizes
those components whose composition with the open system $O$ will satisfy the safety property $P$. Formally,
%$\rho_{c_4}(T_M) \subseteq \rho_{c_4}(T_A)$
%$\rho_{c_3}(T_M) \subseteq T_A$,
$\mathit{Tr}(\rho_{c_3}(T_M ))\subseteq \mathit{Tr}(T_A)$
  if and only if  $(A_M \interleave O)  \models P$ where  $A_M$ complies with $\Info$ as stated by Theorem \ref{theo:main}. To provide its proof, the following propositions and lemmas are needed.
%(according to Definition \ref{def:respectingInfoO}).

Definition \ref{defAerr} indicates that the error traces of $T_A$ denoted by $\mathit{errTr}(T_A)$ are obtained from the  traces in $(A_{\IntM} \interleave_P O)$ that lead to the error state $\pi$
on which we have applied the operator $\rho_{c_1}$, to  hide the internal actions of $O$ and $\IntM$, and the operator $\psi$.

\begin{definition}
    \label{defAerr}
    $ \mathit{errTr}(T_A) = \{t \in \alpha T_A^*\mid  t \in \mathit{errTr}(\psi(\rho_{c_1}((A_{\IntM} \interleave_P O))) \}$, where $\alpha T_A^*\subseteq (\Act_s\cup \Act_r)^*$.
\end{definition}

Lemma \ref{lem::errTra} expresses that an error trace in the property LTS $\psi(A_M \interleave_P O)$ is a finite trace of $\psi(A_M \interleave O)$ while $P_{\it err}$ reaches  $\pi$ by traversing the messages involved in the trace. 

\begin{lemma}\label{lem::errTra}
For any given property $P$, $\zeta\in\mathit{errTr}(\psi(A_M \interleave_P O))\Leftrightarrow \zeta\in{\it Tr}_{\it fin}(\psi(A_M \interleave O))\,\wedge\,\widehat{\partial}(q_0,\psi^{-1}(\zeta))=\pi$, where $P_{\mathit{err}}=( Q,\Sigma,\partial,q_0,\{\pi\})$, $\psi^{-1}:\Act_s^*\rightarrow \Act_s^*$  converts any action $\Rcv(m)$ in the sequence $t$ to $\Snd(m)$, and $\widehat{\partial}(q,\zeta)$ computes the state reachable from $q$ in $P_{\it err}$ through the communicated messages in $\zeta$ and is defined by the equations $
\widehat{\partial}(q,\epsilon)=q$ and $
\widehat{\partial}(q,\langle \Snd(m)::y|r\rangle)=\widehat{\partial}(\partial(q,m),r)$.
\end{lemma}
\begin{proof}
"$\Rightarrow$" By $\zeta\in\mathit{errTr}(\psi(A_M \interleave_P O))$, we conclude that $\exists\,\varsigma\in\mathit{errTr}(A_M \interleave_P O)$ such that $\varsigma=\langle (t_{x_1},\zeta_1)\rangle \oplus \ldots\oplus \langle (t_{x_i},\zeta_i)\rangle$, $\zeta= \zeta_1\oplus\ldots\oplus\zeta_i$. Therefore, there exists a finite execution path $(s_0,q_0)\xRightarrow{(t_{x_1},\zeta_1)}_p\ldots \xRightarrow{(t_{x_{i-1}},\zeta_{i-1})}_p (s_{i-1},q_{i-1})\xRightarrow{(t_{x_i},\zeta_i)}_p \pi$ such that $s_0$ is the initial state of $A_M \interleave_P O$. % and $\varsigma=\langle (t_{x_1},\zeta_1)\rangle \oplus \ldots\oplus \langle (t_{x_i},\zeta_i)\rangle$. 
By application of the rule $\it Take$ of Table \ref{Tab::SOSP}, we can conclude that $s_{i-1}[x_i\mapsto(v,T,body(x_i,m))],q_{i-1})\xrightarrow{\zeta_i}_p\!\!^* (s',\pi)$ for some message $m$. By recursively applying the rule $\it Exe$, we conclude that $s_{i-1}[x_{i}\mapsto(v,T,body(x,m))]\overto{\zeta_i}^*s'$ and $\widehat{\partial}(q_{i-1},\zeta_i)=\pi$ while the message body of $m$ is not completely executed. Therefore, by the application of the rule $\it Take$ of Table \ref{Tab::SOS}, we conclude that $s_{i-1}\xRightarrow{(t_{x_i},\zeta_i')}s''$, where $\zeta_i$ is a prefix of $\zeta_i'$ and the message body of $m$ is completely executed. With a similar discussion, $\forall\, 0< j< i-1\,(s_{j-1}\xRightarrow{(t_{x_{j}},\zeta_{j})}s_{j}\,\wedge\, \widehat{\partial}(q_{j-1},\zeta_j)=q_j)$. Therefore, $\langle (t_{x_1},\zeta_1)\rangle \oplus \ldots\oplus\langle (t_{x_{i-1}},\zeta_{i-1})\rangle\oplus \langle (t_{x_i},\zeta_i')\rangle\in{\it Tr}_{\it fin}(A_M \interleave O)$, and $\widehat{\partial}(q_{0}, \zeta_1\oplus \ldots \oplus \zeta_i)=\pi$. Consequently $\zeta_1\oplus\ldots\oplus\zeta_{i-1}\oplus\zeta_i'\in {\it Tr}_{\it fin}(\psi(A_M \interleave O))$, %$\zeta'\in {\it Tr}_{\it fin}(\psi(A_M \interleave O))$, 
and  
$\zeta\in {\it Tr}_{\it fin}(\psi(A_M \interleave O))$. 

"$\Leftarrow$" We show that $\zeta\in{\it Tr}_{\it fin}(\psi(A_M \interleave O))\,\wedge\,\widehat{\partial}(q_0,\psi^{-1}(\zeta))=\pi\Rightarrow \zeta\in\mathit{errTr}(\psi(A_M \interleave_P O))$. The assumption $\zeta\in{\it Tr}_{\it fin}(\psi(A_M \interleave O))$ implies that $\zeta'\in{\it Tr}_{\it fin}(\psi(A_M \interleave O))$ where $\zeta$ is a prefix of $\zeta'$ with $\zeta'=\zeta_1\oplus\ldots\oplus\zeta_{i-1}\oplus\zeta_i'$, $s_0\xRightarrow{(t_{x_1},\zeta_1)}\ldots \xRightarrow{(t_{x_{i-1}},\zeta_{i-1})} s_{i-1}\xRightarrow{(t_{x_i},\zeta_i')} s''$, and $s_0$ is the initial state of $\psi(A_M \interleave O)$. By application of the rule $\it Take$ of Table \ref{Tab::SOS}, we conclude that $\forall\, 0< j< i-1\,(s_{j-1}[x_i\mapsto(v,T,body(x_i,m))]\overto{\zeta_{j}}^\ast s_{j}\,\wedge\, \widehat{\partial}(q_{j-1},\zeta_j)=q_j)$. Therefore, by application of the rules $\it Exe$ and $\it Take$ of Table \ref{Tab::SOSP}, it is concluded that $(s_0,q_0)\xRightarrow{(t_{x_1},\zeta_1)}_p\ldots \xRightarrow{(t_{x_{i-1}},\zeta_{i-1})}_p (s_{i-1},q_{i-1})$. Furthermore, $s_{i-1}[x_i\mapsto(v,T,body(x_i,m))]\overto{\zeta_{i}}_p^\ast s'$ while $\widehat{\partial}(q_{j-1},\zeta_i)=\pi$ and $\zeta_i$ is a prefix of $\zeta_i'$. Thus, $(s_{i-1},q_{i-1})\xRightarrow{(t_{x_i},\zeta_i)} \pi$ and consequently $ \zeta\in\mathit{errTr}(\psi(A_M \interleave_P O))$.
\end{proof}

The following proposition is used to generalize the result of Lemma \ref{lem::errTra} when a renaming function is applied on the property LTS. 
\begin{proposition}\label{pro::rename}
For a LTS $T_L$ specified by $(S_L,\alpha T_L,R_L,s_0^L)$, $t\in{\it errTr}(T_L)$ if and only if $c^\ast(t)\in {\it errTr}(\rho_c(T_L))$, where $c:{\it Act}\rightarrow {\it Act}$ is an arbitrary renaming function and $c^\ast$ applies $c$ to the elements of the input action sequence.
\end{proposition}
\begin{proof}
By definition $t\in{\it errTr}(T_L)$ implies that there exists the execution path $s_0^L\overto{a_1}s_1\overto{a_2}\ldots\overto{a_n}\pi$ such that $t=\langle a_1, a_2, \ldots a_n\rangle$. By applying $\rho_c$ on $T_L$, this path is changed into $s_0^L\overto{c(a_1)}s_1\overto{c(a_2)}\ldots\overto{c(a_n)}\pi$. Thus, the trace $c^\ast(t)=\langle c(a_1), c(a_2), \ldots c(a_n)\rangle$ belongs to ${\it errTr}(\rho_c(T_L))$. Its reverse is proved with the same discussion.
\end{proof}

Proposition \ref{pro:err} indicates that $T_A$ has no trace with a prefix included in its error traces. Conversely, each trace not belonging to $T_A$ has a prefix included in its error traces.  
\begin{proposition}\label{pro:err}
$t \in {\it errTr(T_A)}$ if and only if     $\forall t' \,(t\oplus t'   \not \in {\it  Tr}(T_A)) $.
\end{proposition}
\begin{proof}
``$\Rightarrow$'' By Definition \ref{defAerr}, $t \in {\it errTr}(T_A)$ implies that $t \in \mathit{errTr}(\psi(\rho_{c_1}((A_{\IntM} \interleave_P O)))$. Regarding the completion phase of the step three in the construction of an assumption,  $T_A$ is achieved by completing $\psi(\rho_{c_1}((A_{\IntM} \interleave_P O))$ using the subset construction method and then eliminating the $\pi$ state. In other words, $T_A$ is generated by removing traces leading to the state $\pi$ from $\psi(\rho_{c_1}((A_{\IntM} \interleave_P O))$. Hence, definitely $\forall t'\,(t\oplus t' \not \in {\it  Tr(T_A)})$.

``$\Leftarrow$'' We show that if $t''\not\in{\it  Tr}(T_A)$, then $\exists\, t,\,t'(t''=t\oplus t'\,\wedge\, t\in{\it errTr}(T_A))$. Regarding the third step towards the construction of $T_A$,  $t''\not\in{\it  Tr}(T_A)$ implies that $t''$ has a prefix $t$ leading the $\pi$ state in $\psi(\rho_{c_1}((A_{\IntM} \interleave_P O))$. Definition \ref{defAerr} implies that $t \in {\it errTr(T_A)}$. %and the transitions leading to it are removed from the assumption, therefore each trace that does not exist in $T_A$ has  definitely been  a trace leading to the omitted $\pi$ state.
\end{proof}

Lemma \ref{lemma:lemma1} indicates that if $A_M$ complies with $\Info$, then the set of error traces of the system containing $A_M$ and $O$
(from which we have hidden the internal actions of $M$  by employing $\rho_{c_3}$), is the subset of the set of error traces of the system composed of $A_{\IntM}$ and $O$, where the operators $\rho_{c_1}$ and $\psi$ are employed on both sets of traces.

\begin{lemma}
    \label{lemma:lemma1}
%   $A_M \sqsubseteq_c \Info$ $\implies$
$
\mathit{errTr}(\rho_{c_3}(\psi(\rho_{c_1}(A_M \interleave_P O))))
     \subseteq
    \mathit{errTr}(\psi(\rho_{c_1}(A_{\IntM} \interleave_P O)))
    $ if $A_M \sqsubseteq_c \Info$.
\end{lemma}    
\begin{proof}
%Proof is by contradiction. 
Assume that
$
\mathit{errTr}(\rho_{c_3}(\psi(\rho_{c_1}(A_M \interleave_P O))))
    \not \subseteq
\mathit{errTr}(\psi(\rho_{c_1}(A_{\IntM} \interleave_P O)))$ and $P_{\mathit{err}}=( Q,\Sigma,\partial,q_0,\{\pi\})$. Thus, $\exists\, t \in \alpha T_A^*\,( t\in\mathit{errTr}(\rho_{c_3}(\psi(\rho_{c_1}(A_M \interleave_P O))))\, \wedge \, t \not\in         \mathit{errTr}(\psi(\rho_{c_1}(A_{\IntM} \interleave_P O))))$. By Lemma \ref{lem::errTra} and Proposition \ref{pro::rename}, 
$ t \in \mathit{errTr}(\rho_{c_3}(\psi(\rho_{c_1}(A_M \interleave_P O))))\Rightarrow  t \in \mathit{Tr}_{\it fin}(\rho_{c_3}(\psi(A_M \interleave O))) \,\wedge \widehat{\partial}(q_0,t)=\pi$. Hence by application of Lemma \ref{lem::errTra} and Proposition \ref{pro::rename}, $t \not\in         \mathit{errTr}(\psi(\rho_{c_1}(A_{\IntM} \interleave_P O)))$ and $\widehat{\partial}(q_0,t)=\pi$ implies that $t \not\in \mathit{Tr}_{\it fin}(\psi(A_{\IntM} \interleave O))$. As $O$ is common in $\rho_{c_3}(\psi(A_M \interleave O))$ and $\psi(A_{\IntM} \interleave O)$ and $M$ (similarly $\IntM$) act as a service, we can conclude that $M$ upon serving a message $m$ from $O$, sends responses different from $\IntM$. Assume $\zeta$ as a subsequence of $t$, be such a responses. By construction of $\IntM$ which sends all possible shuffling of messages for each  $(m,\mathcal{I})\in\Info$, we conclude that $\nexists (m,\mathcal{I})\in\Info,\,\forall\, A_y\in O\,((\zeta\downharpoonleft y)\in \mathcal{I})$ and consequently by Definition \ref{def:complianceNew}, $A_M \not\sqsubseteq_c \Info$, contradicting the assumption.
    \end{proof}

\begin{lemma}
\label{lemma:lemma2}
For a given $A_M$, $
\mathit{Tr}(\psi(\rho_{c_1}(A_M \interleave O)))\subseteq \mathit{Tr}(\rho_{c_2} (A_M \interleave A_{M^c} \interleave O^d ))$, where $O=\{A_{o_1},\ldots,A_{o_\imath}\}$ and 
$O^d = A_{o_1^d} \interleave \ldots \interleave  A_{o_\imath^d}$. % and $\forall A_{c_i} \in O,(A_{c_i^d}= (c_i^d,\emptyset,\emptyset,0))$.
\end{lemma}

\begin{proof}
%Proof is by contradiction. 
Assume that $\mathit{Tr}(\psi(\rho_{c_1}(A_M \interleave O)))\not\subseteq \mathit{Tr}(\rho_{c_2} (A_M \interleave A_{M^c} \interleave O^d ) )$. Consequently $\exists\, t \in \alpha T_A^*\,(t \in \mathit{Tr}(\psi(\rho_{c_1}(A_M \interleave O)))\,\wedge\, t \not\in\mathit{Tr}(\rho_{c_2} (A_M \interleave A_{M^c} \interleave O^d ) ))$. As $M$ act as a service, the trace $t$ starts with a receive action by $A_M$ and may continue with the responses of $A_M$ or other receive actions. As $A_M$ is common in both $\psi(\rho_{c_1}(A_M \interleave O))$ and $\rho_{c_2} (A_M \interleave A_{M^c} \interleave O^d )$ and $O^d$ is passive by construction, we can conclude that $A_M$ receives a message from $O$ which can not be received from $A_{M^c}$.  %Assume $\zeta^\ast$ as a subsequence of $t$, be such a sequence. 
By the construction of $A_{M^c}$ which non-deterministically sends all messages to $A_M$, %and the concurrent execution of actors, 
this is a contradiction.

%By using the asynchronous semantics of $\interleave$ and by construction of $\rho_{c_2} (A_M \interleave A_{M^c} \interleave O^d )$, we know that all the messages sent to $A_M$ by $O$ in response of a received message are also sent to $A_M$ by $A_{M^c}$. Therefore in $\rho_{c_2} (A_M \interleave A_{M^c} \interleave O^d )$, a sequence of messages that should have be sent out by $A_M$ is not being sent out, meaning $\rho_{c_2} (A_M \interleave A_{M^c} \interleave O^d )$ is not constructed by using $A_M$ which is a contradiction.

\end{proof}

%Theorem \ref{theo:main} claims that if $A_M$ complies with $\Info$ and the set of observable traces of the LTS generated off of $A_M$  is the subset of the set of traces of the generated assumption $T_A$ where
%the internal actions of $M$ are hidden,
%%the renaming operators has been applied on both $T_M$ and $T_A$,
% then the system containing $A_M$ and the actors of $O$ will satisfy the property $P$.

Theorem \ref{theo:main}  claims that $T_A$ is the weakest
assumption about the component $M$ that ensures the safety property $P$.

\begin{theorem}
\label{theo:main}
For any given $A_M$ such that $A_M \sqsubseteq_c \Info$, $\mathit{Tr}(\rho_{c_3}(T_M)) \subseteq \mathit{Tr}(T_A) $ if and only if $ (A_M \interleave O) \models P$, where $T_A$ is the generated assumption, $P_{\mathit{err}}=( Q,\Sigma,\partial,q_0,\{\pi\})$, $T_M=\rho_{c_2} (A_M \interleave A_{M^c} \interleave  O^d)$, $O=\{A_{o_1},\ldots,A_{o_\imath}\}$, and 
$O^d = A_{o_1^d} \interleave \ldots \interleave  A_{o_\imath^d}$.
\end{theorem}

\begin{proof} ``$\Rightarrow$'' We show that for any given $A_M$ such that $A_M \sqsubseteq_c \Info$,  $\mathit{Tr}(\rho_{c_3}(T_M)) \subseteq \mathit{Tr}(T_A)$ implies $ (A_M \interleave O) \models P$.

The proof is managed by contradiction. Assume that $(A_M \interleave O)  \not\models P$ which by Definition \ref{def::satisfaction} and Proposition \ref{pro::rename} implies that $\exists\,t_1\in \mathit{errTr}(A_M \interleave_P O)\,\wedge \, t_2 \in \mathit{errTr}(\rho_{c_3}(\psi(\rho_{c_1}(A_M \interleave_P O))))$ (result $\ast$), where $t_2$ is achieved from $t_1$ by hiding the internal actions of $O$ and $A_M$ (as the effect of $\rho_{c_1}$ and $\rho_{c_3}$, respectively), and removing $t_x$ from the actions $(t_x,\varsigma)$ while send actions to $M$ are converted into receive actions (as the effect of $\psi$).   This together with Lemma \ref{lemma:lemma1} imply that $t_2 \in  \mathit{errTr}(\psi(\rho_{c_1}(A_{\IntM} \interleave_P O)))$ and so by Definition \ref{defAerr}, we have $t_2 \in \mathit{errTr}(T_A)$ (result $\dag$).

From the result $(\ast)$, Lemma \ref{lem::errTra}, and Proposition \ref{pro::rename}, we conclude that $\exists\, t_3\in  \mathit{Tr}(\rho_{c_3}(\psi(\rho_{c_1}(A_M \interleave O))))$ such that $t_2$ is a prefix of $t_3$. By application of  Lemma \ref{lemma:lemma2}, % and the fact that a renaming operator does not effect error traces, 
we  conclude that $t_3\in \mathit{Tr}(\rho_{c_3}(\rho_{c_2} (A_M \interleave A_{M^c} \interleave O^d ) ))$. By the construction of $T_M$, we conclude that $ t_3  \in \mathit{Tr}(\rho_{c_3}(T_M ) )$ (result $\ddag$).

%  together with the fact $\mathit{Tr}(A_M \interleave_P O )\subseteq \mathit{Tr}(A_M \interleave O)$, we have $t_2 \in \mathit{errTr}(\rho_{c_3}(\psi(\rho_{c_1}(A_M \interleave O))))$. 

The result $(\dag)$ together with Proposition \ref{pro:err}, we conclude that $t_3 \notin  \mathit{Tr}(T_A)$ while $t_3\in   \mathit{Tr}(\rho_{c_3}(T_M ))$ meaning that $ \mathit{Tr}(\rho_{c_3}(T_M ))\not\subseteq \mathit{Tr}(T_A)$  which is a contradiction.

%Since in our approach we have omitted the $\pi$ state, from the results $(\dag)$ and $(\ddag)$ we conclude that $t_2 \notin  \mathit{Tr}(T_A)$ 

``$\Leftarrow$'' We show that for any given $A_M$ such that $A_M \sqsubseteq_c \Info$, $(A_M \interleave O) \models P$ implies $\mathit{Tr}(\rho_{c_3}(T_M)) \subseteq \mathit{Tr}(T_A)$. The proof is managed by contradiction. Assume that $\mathit{Tr}(\rho_{c_3}(T_M)) \not\subseteq \mathit{Tr}(T_A)$ which implies that $\exists \, t_1 \in\mathit{Tr}( \rho_{c_3}(T_M))\, \wedge \, t_1 \not \in \mathit{Tr}(T_A)$.

Proposition \ref{pro:err} together with the assumption  $t_1 \not \in \mathit{Tr}(T_A)$ imply that $\exists \, t_2\in {\it errTr}(T_A)$ where $t_2$ is a prefix of $t_1$, and hence by Proposition \ref{pro::rename} and Definition \ref{defAerr}, $t_2' \in \mathit{errTr}(\psi(A_{\IntM} \interleave_P O)$ and $\widehat{\partial}(q_0,\psi^{-1}(t_2'))=\pi$ (result $\ast$), where $t_2$ is achieved from $t_2'$ by hiding the internal actions of $O$ (as the effect of $\rho_{c_1}$). 

Since $t_1 \in \mathit{Tr}(\rho_{c_3}(T_M))$, by the construction of $T_M$ we conclude that $t_1 \in \mathit{Tr}(\rho_{c_3}(\rho_{c_2} (A_M \interleave A_{M^c} \interleave O^d)))$ and hence, $t_2 \in \mathit{Tr}_{\it fin}(\rho_{c_3}(\rho_{c_2} (A_M \interleave A_{M^c} \interleave O^d)))$ (result $\dag$) and $t_2=t_2'$. By the fact  ${\it Tr}(\mathcal{T}_1)\subseteq {\it Tr}(\mathcal{T}_2)\Rightarrow{\it Tr}_{\it fin}(\mathcal{T}_1)\subseteq {\it Tr}_{\it fin}(\mathcal{T}_2)$ \cite{baier2008principles}, where $\mathcal{T}_1$ and $\mathcal{T}_2$ are LTSs, and  application of Lemma \ref{lemma:lemma2}, we conclude that $
\mathit{Tr}_{\it fin}(\rho_{c_3}(\psi(\rho_{c_1}(A_M \interleave O))))\subseteq \mathit{Tr}_{\it fin}(\rho_{c_3}(\rho_{c_2} (A_M \interleave A_{M^c} \interleave O^d )))$. Thus, two cases can be distinguished:\begin{itemize}
\item $t_2 \in\mathit{Tr}_{\it fin}(\rho_{c_3}(\psi(\rho_{c_1}(A_M \interleave O)))$: Lemma \ref{lem::errTra} and Proposition \ref{pro::rename} together with the result ($\ast$) imply that $t_2\in\mathit{errTr}(\rho_{c_3}(\psi(\rho_{c_1}(A_M \interleave_P O))))$ and hence by Proposition \ref{pro::rename}, $t_3\in\mathit{errTr}(A_M \interleave_P O)$, where $t_2$ is achieved from $t_3$ by hiding the internal actions of $O$ and $A_M$ (as the effect of $\rho_{c_1}$ and $\rho_{c_3}$, respectively), and removing $t_x$ from the actions $(t_x,\varsigma)$ while send actions to $M$ are converted into receive actions (as the effect of $\psi$). 
%$\phi(t_3)=\psi^{-1}(t_2)$ (see Lemma \ref{lem::errTra} for the definition of $\phi$).
Concluding by Definition \ref{def::satisfaction} that $(A_M \interleave O)\not\models P$, which is a contradiction;

\item $t_2 \not\in \it{Tr}_{\it fin}(\rho_{c_3}(\psi(\rho_{c_1} (A_M \interleave O))))$: according to the result($\dag$), as $A_M$ is common in both $\rho_{c_3}(\rho_{c_2} (A_M \interleave A_{M^c} \interleave O^d))$ and $\rho_{c_3}(\psi(\rho_{c_1} (A_M \interleave O)))$ and acts an a service, we conclude that $A_{M^c}$ is producing a trace that $O$ does not. Thus, we can safely replace $A_M$ by $A_\IntM$ in $\rho_{c_3}(\rho_{c_2} (A_M \interleave A_{M^c} \interleave O^d))$ and claim that $t_2 \not\in \it{Tr}_{\it fin}(\psi(\rho_{c_1} (A_\IntM \interleave O)))$ while $\rho_{c_3}$ can be safely remove as $A_\IntM$ does not include the internal actions of $A_M$. This results that $t_2 \not\in \it{Tr}_{\it fin}(\psi(\rho_{c_1} (A_\IntM \interleave_P O)))$. By the third step of our algorithm towards the construction of $T_A$ which completes $\psi(\rho_{c_1} (A_\IntM \interleave_P O))$, we can conclude that $t_2\in\it{Tr}_{\it fin}(T_A)$, and consequently, $t_1 \in \mathit{Tr}(T_A)$ which is a contradiction.
\end{itemize}

\end{proof}

\section{Case Studies}\label{sec:CaseStudy}
To demonstrate the applicability of our framework in real world case studies, we apply it to different application domains: we first apply it to a small system in the robotic domain %for better understanding of model properties
and then
study the more complex system of Electronic Funds Transfer. 

\subsection{Quadricopter System}
We apply our approach to
verify the distributed controller of a robotic system.
% mentioned in \cite{zakeriyan2015jacco}.
This system was developed in Robotic Operating System (ROS) which is a framework for development of distributed robotic systems\footnote{\url{http://www.ros.org}} and has been applied to develop many complex and industrial robotic systems, e.g. Husky robot\footnote{\url{wiki.ros.org/Robots/Husky}} and Turtlebot\footnote{\url{wiki.ros.org/Robots/Turtlebot}}.
ROS provides an infrastructure to compose different modules together for building a system. These modules are reusable and may be replaced or altered, in time.

When using ROS, different components of the systems  communicate by asynchronous message passing and therefore
ROS's computation model is similar to the actor model. Each processing unit of a ROS system can be modeled as an actor and  communications in ROS code can be transformed into message passings among actors, as practiced in   %Based on this similarity, we transformed the components of ROS system into actors. For each node an actor is defined and communications in ROS code are transformed into message-passings among actors
\cite{zakeriyan2015jacco}. Since the components of a ROS system are reusable, it is appealing to find an assumption (for a specific component) of the system which specifies a
\textit{well-defined} component of the system. \textit{Well-defined} in the sense that a \textit{correctly specified} component composed with the rest of system will satisfy an arbitrary property.

\begin{figure}
    \begin {tikzpicture}[scale=.75, transform shape]

\node [text width=1.9cm, style=circle,draw] (v1) at (-2,2.5) {$\it~~ controller$};
\node  [text width=1.9cm,style=circle,draw] (v2)at (1,2.5) {$\it ~~~observer$};
\node  [text width=1.9cm,style=circle,draw] (v3)at (-3.5,0) {$\it~ transmitter$};
\node  [text width=1.9cm,style=circle,draw] (v4)at (2.5,0) {$\it~~~feedback$};
\node  [text width=1.9cm,style=circle,draw] (v5)at (-0.5,-1.5) {$\it quadricopter$};

\path[<-,-latex]  (v2) edge  [bend right = -15] node {} (v1);
\path[ ->,-latex]  (v1) edge  [bend right = -15] node {} (v2);
\path[ ->,-latex]  (v4) edge  [bend left = -15] node {} (v2);
\path[ ->,-latex]  (v5) edge  [bend left = -15] node {} (v4);
\path[ ->,-latex]  (v3) edge  [bend left = -15] node {} (v5);
\path[ ->,-latex]  (v1) edge  [bend left = -15] node {} (v3);

\end{tikzpicture}
    \centering
    \caption{Data flow in Quadricopter system }
    \label{QuadriFlow}
\end{figure}
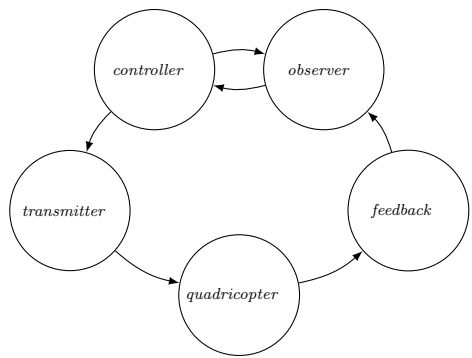

Figure \ref{QuadriFlow} demonstrates the data flow among the different modules of the Quadricopter system. The system contains the components: $\Controller$, $\Feedback$, $\Quadricopter$, $\Observer$ and $\Transmitter$ where  $\Controller$ governs movements of $\Quadricopter$ robot by sending its commands through  $\Transmitter$ to  $\Quadricopter$. To this aim,  $\Controller$ computes a direction based on the received information (from  $\Observer$) and other
parameters and sends an $\mathit{update}$ message to  $\Transmitter$. It also notifies  $\Observer$ of the transmitted command (via a $\mathit{ctrlerUpdate}$ message). The components $\Feedback$ and $\Transmitter$ are the interfaces of  $\Quadricopter$. Component
$\Transmitter$ translates $\mathit{update}$ messages to the commands understandable by
$\Quadricopter$. $\it Quadricopter$ moves to a different position, after processing a received command. It also measures environmental metrics using a camera
recorder and some other sensors. The measured data are sent to $\Feedback$ which translates the data to information comprehensible by  $\Observer$. Component $\Observer$ is supposed to validate this information (based on the previous commands and the received info) and send the necessary parts to $\Controller$ (via an $\mathit{update}$ message) which makes the feedback cycle of the Quadricopter system complete.

\begin{figure}

    \begin{lstlisting}[language=palang, multicols=2]
    actor transmitter(10) {

        initial{
            self!transmit;
        }
        transmit {
            quadricopter!update;
        }
        update{
            self!transmit;
        }
    }

    actor controller(10) {

        initial{
            self!control;
        }
        control{
            transmitter!update;
            observer!ctrlerUpdate;
        }
        update{
            self!control;
        }
    }

    actor feedback(10) {

        feedback{
            observer!update;
        }
        update{
            self!feedback;
        }
    }

    actor quadricopter(10) {

         initial{
             self!move;
         }
        move{
            feedback!update;
        }
        update{
            self!move;
        }
    }

    main {
        controller!initial;
        //transmitter!initial;
        //quadricopter!initial;

    }
    \end{lstlisting}
    \caption{Actors of the Quadricopter system (open system)}
    \label{fig:QuadriOpenSystem}
\end{figure}

Each module of the system might be subject to changes or reversions and therefore can be assumed as the \textit{unspecified component}. Here we run our algorithm with the presumption that $\Observer$ component plays such a role. The specified components of this system are modeled as actors of AML specified in Figure \ref{fig:QuadriOpenSystem}. Since we aim to verify the system against concurrency problems, calculations involving the control values of the system have been abstracted away from the message handlers. The desired property for the system to satisfy is that $\Observer$ only sends an $\mathit{update}$ message to $\Controller$ if $\Transmitter$ has already sent an $\mathit{update}$ message to $\Quadricopter$. The LTL formula describing this property along side the DFA for the bad prefixes of the LTL formula are given in Figure \ref{fig:quadricopterProperty}.

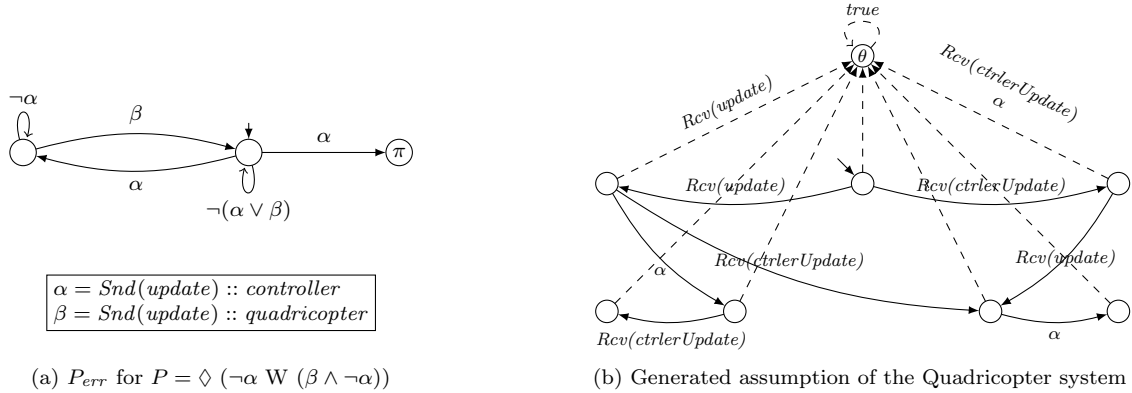
\begin{figure}
    \subfloat[$P_{\mathit{err}}$ for  $P =  \lozenge\hspace{1mm} (\lnot\alpha~ \mathrm{W}~ ( \beta \land \lnot\alpha))$\label{fig:quadricopterProperty}]{%
    \begin {tikzpicture}
    \tikzstyle{bordered} = [draw,outer sep=0,inner sep=1,minimum size=10]

    \node [style=circle,draw,outer sep=0,inner sep=1,minimum size=10] (v1)at (-3,2) {};
    \node [style=circle,draw,outer sep=0,inner sep=1,minimum size=10] (v2)at (-6,2) {};
    \node [style=circle,draw,outer sep=0,inner sep=1,minimum size=10] (v3) at (-1,2) {$\pi$};
    \node (v0) at(-3,2.5) {};

    \path[<-,-latex]  (v2) edge  [bend right = -15] node[above =0.03 cm] {$\beta$ } (v1);
    \path[ ->,-latex]  (v1) edge  [bend right = -15] node[below =0.03 cm] {$\alpha$} (v2);
    \path[->,-latex]  (v1) edge node[above =0.03 cm] {$\alpha $} (v3);
    \path[<-,-latex]  (v2) edge   [loop above] node {$\lnot \alpha$ } (v2);
    \path[<-,-latex]  (v1) edge   [loop below] node {$\lnot(\alpha \lor \beta)$} (v1);
    \path[<-,-latex]  (v0) edge (v1);

    \node at (-3.5,0)[bordered] { $\begin{array}{l}{\alpha =\it Snd(update)::controller}\\
        \beta = {\it Snd(update)::quadricopter} \end{array}$};
    \node at (-3.5,4.1){};
    \node at (-3.5,-.6){};
\end{tikzpicture}
    }
    \hfill
    \subfloat[Generated assumption of the Quadricopter system
    \label{fig:QuadriAssump}]{%
\begin{tikzpicture}[scale=.85, transform shape]
\node [style=circle,draw,outer sep=0,inner sep=1,minimum size=10] (v1)at (-3,4) {};
\node [style=circle,draw,outer sep=0,inner sep=1,minimum size=10] (v2)at (1,4) {};
\node [style=circle,draw,outer sep=0,inner sep=1,minimum size=10] (v3)at (5,4) {};
\node [style=circle,draw,outer sep=0,inner sep=1,minimum size=10] (v4)at (-1,2) {};
\node [style=circle,draw,outer sep=0,inner sep=1,minimum size=10] (v12)at (5,2) {};
\node [style=circle,draw,outer sep=0,inner sep=1,minimum size=10] (v8)at (-3,2) {};
\node [style=circle,draw,outer sep=0,inner sep=1,minimum size=10] (v13)at (3,2) {};
\node (v0)at (.5,4.5) {};
\node [style=circle,draw,outer sep=0,inner sep=1,minimum size=10] (v14) at (1,6) {$\theta$};

\path[<-,-latex]  (v0) edge (v2);
\path[ ->,-latex]  (v2) edge[bend left=15]   node[above =0.03 cm] {\it Rcv(update)} (v1);
\path[ ->,-latex]  (v2) edge[bend right=15,sloped]  node[above =0.03 cm] {\it Rcv(ctrlerUpdate)} (v3);
\path[->,-latex]  (v1) edge [bend right=15]   node[below=0.03 cm] {\it $\alpha$} (v4);
\path[ ->,-latex]  (v3) edge [bend left=15]   node[] {\it Rcv(update)} (v13);
\path[ ->,-latex]  (v1) edge[bend right=15]  node[above] {\it Rcv(ctrlerUpdate)} (v13);
\path[ ->,-latex]  (v4) edge[bend left=15]  node[below =0.03 cm] {\it Rcv(ctrlerUpdate)} (v8);
\path[->,-latex]  (v13) edge [bend right=15,sloped]   node[below=0.03 cm] {\it $\alpha$} (v12);

\path[ ->,dashed,-latex]  (v1) edge[sloped]   node[above =0.03 cm] {\it Rcv(update)} (v14);
\path[ ->,dashed,-latex]  (v2) edge  node[above =0.03 cm] {} (v14);
\path[ ->,dashed,-latex]  (v3) edge [sloped]   node[above =0.03 cm] {$\begin{array}{c}{\it Rcv(ctrlerUpdate)}\\{\it \alpha}\end{array}$} (v14);
\path[ ->,dashed,-latex]  (v4) edge [sloped]   node[above =0.03 cm] {} (v14);
\path[ ->,dashed,-latex]  (v8) edge [sloped]   node[above =0.03 cm] {} (v14);
\path[ ->,dashed,-latex]  (v12) edge [sloped]   node[below =0.03 cm] {} (v14);
\path[ ->,dashed,-latex]  (v13) edge  node [below =0.03 cm] {} (v14);
\path[ ->,dashed,-latex]  (v14) edge[loop,above ]  node [] {${\it true}$} (v14);
\end{tikzpicture}
    }
    \caption{The LTS of $P_{
            \it err}$ compared to the generated assumption}
    \label{fig:dummy}
\end{figure}

Based on the system description,  $\Controller$ of the open system expects to receive an  $\mathit{update}$ message in response to the $\mathit{update}$ message sent to $\Observer$ by $\Controller$ and does not expect a response upon sending a notification message by  $\Controller$. Using this information the $\Info$ is defined as:
$\Info = \{(\mathit{update},
\{(\langle\mathit{update}\rangle,\Controller)\}),$ $ (\mathit{ctrlerUpdate},\{\})\}$.
Using $\Info$ we generate $A_\IntM$ which is the overapproximated replacement of $\Observer$, demonstrated in Figure \ref{fig:ObserverIntM}.

We implemented our algorithm twice on this settings: once only  $\Controller$ is initiator and once the modules of $\Transmitter$, $\Controller$ and $\Quadricopter$ could all be initiators (i.e., the codes of the main body are uncommented in Figure \ref{fig:QuadriOpenSystem}). In the former case, it was discovered that the property satisfies for any $\Observer$ actor which complies with $\Info$, meaning that in the generated property LTS, the $\pi$ state was not reachable and the algorithm was terminated at the first step. In the latter case, the algorithm terminated with an assumption, illustrated in Figure \ref{fig:QuadriAssump}.

% meaning the main block was changed. The new main block is depicted in Figure \ref{fig:newMainBlock}.
We checked a given $\Observer$ module (mentioned in the Figure \ref{fig:ObserverActor}) against the generated assumption which was rejected. Intuitively,  $\Quadricopter$ does not wait for a message from $\Transmitter$ to send its message to $\Feedback$. Consequently, $\Observer$ will receive a message from $\Feedback$ and then will send its message to  $\Controller$ while $\Transmitter$ has not sent anything beforehand.

%\begin{figure}[H]

%   \begin{lstlisting}[language=palang, multicols=2]
%   main {
%       controller!control;
%       transmitter!transmit;
%       quadricopter!move;
%   }
%   \end{lstlisting}
%   \caption{ The Changed Main Block for the Quadricopter System}
%   \label{fig:newMainBlock}
%\end{figure}

\begin{figure}
    \centering
\begin{lrbox}{\mylistingbox}%
    \begin{minipage}{.45\linewidth}%
        \begin{lstlisting}[language=palang]
    actor observer(10) {

        update{
            controller!update;
        }
        ctrlerUpdate{

        }
    }
        \end{lstlisting} %
    \end{minipage}%
\end{lrbox}%
\subfloat[$\IntM$\label{fig:ObserverIntM}]{\usebox{\mylistingbox}}%
\hfill
\begin{lrbox}{\mylistingbox}%
    \begin{minipage}{.45\linewidth}%
    \begin{lstlisting}[language=palang]
    actor observer(3) {
        observe {
            controller!update;
        }
        update{
            self!observe;
        }
        ctrlerUpdate{}
    }

    \end{lstlisting}
    \end{minipage}%
\end{lrbox}%
\subfloat[$\Observer$ actor\label{fig:ObserverActor}]{\usebox{\mylistingbox}}%

\caption{The code of the generated $\IntM$ compared to the given $\Observer$ actor}
\end{figure}

\subsection{Electronic Funds Transfer System}

Electronic Funds Transfer (EFT) systems operate as the infrastructure for online financial transactions. The components of an EFT system such as automatic teller machines, Point-of-Sale terminals (PoS), core banking systems and the EFT Switch  communicate over message passing~\cite{asaadi2011towards}. Figure \ref{EFTSFlow}  demonstrates the data flow between system components.

The EFT switch (or in short the switch) is  a software component designed to route  messages to  their destinations. The transactions between system components  consist of several
messages passing through the switch. For example, a simple purchase transaction is as follows: the  transaction is  originated
by a user inserting a card in a PoS terminal, afterwards PoS sends a purchase request to the switch which in turn forwards the request to the core banking system in order to charge the user's account. The core banking system forwards the response back to the switch and eventually to the PoS  and the user.

Furthermore, each transaction may comprise a complex combination of different possible interaction scenarios among the components of the EFT system. For instance,
a PoS is configured to time-out and send a cancel request to the switch if it does not receive the response to a purchase request from the switch in time.
%a user cancels a purchase request because he changes his mind or
%receiving a response takes too long.
When a user cancels his request voluntarily and its cancel request
is being handled in the switch,
if the authentication of the user has been failed beforehand, then the cancel
request should not be sent to the core banking system (because no purchase has been made). In turn a  message should be sent to PoS and it should notify the
user of the more prior event which is the authentication failure. Therefore, in reality EFT systems are error-prone due to unintentional transaction flows caused by failures in the components and asynchrony in the communication media.
%that when the cancel request is sent, the purchase response has been sent by the Core and is on its way back to the PoS. In
%this case, the PoS receives a purchase response after it has already sent the canceling request which  will be in turn responded by the Switch. A desired behavior of the PoS in such a case may be not to report the purchase response to the User and wait for the response to the cancel request and then report the successful cancellation to the User.
%As demonstrated in the example

%PoS  is a component of the system which may have different designs based on its manufacturer and having  PoSs that conform to the properties of the EFT system is essential. Using model checking techniques to check the state space of the EFT system against a preferred property for each different brand of PoS is not optimal. We propose our approach instead to produce an assumption such that if the PoS satisfies this assumption then the property is satisfied for the EFT system.

 \begin{figure}
    \centering
\begin{tikzpicture}[scale=.7, transform shape]
 \tikzstyle{bordered} = [draw,outer sep=0,inner sep=1,minimum size=10]

\node [style=circle,draw,outer sep=0,inner sep=3,minimum size=10] (v1)at (0,1) {$\it user$};

\node [style=circle,draw,outer sep=0,inner sep=3,minimum size=10] (v2)at (2,1) {$\it PoS$};
\node [style=circle,draw,outer sep=0,inner sep=3,minimum size=10] (v3)at (5,1) {$\it EFTSwitchCore$};
\node [style=circle,draw,outer sep=0,inner sep=3,minimum size=10] (v4)at (7,3) {$\it purchaseT.$};
\node [style=circle,draw,outer sep=0,inner sep=3,minimum size=10] (v5)at (7,-1) {$\it balanceT.$};
\node  [style=circle,draw,outer sep=0,inner sep=3,minimum size=10] (v6)at (11,1) {$\it core$};
\node [style = rectangle,dashed,draw,minimum width=5.3cm,minimum height=6.3cm,draw] at (5.8,1) { };

\path[ <-,-latex]  (v1) edge  [bend right = -20] node[] {} (v2);
\path[ <-,-latex]  (v2) edge  [bend right = -20] node[] {} (v1);
\path[ <-,-latex]  (v2) edge  [bend right = -20] node[] {} (v3);
\path[ <-,-latex]  (v3) edge  [bend right = -20] node[] {} (v2);
\path[ <-,-latex]  (v3) edge  node[] {} (v4);
\path[ <-,-latex]  (v3) edge  node[] {} (v5);
\path[<-,-latex]  (v4) edge  [bend left = -30] node[] {} (v2);
\path[ <-,-latex]  (v5) edge  [bend right = -30] node[] {} (v2);
\path[<-,-latex]  (v4) edge  [bend left = -20] node[] {} (v6);
\path[ <-,-latex]  (v6) edge  [bend left = -20] node[] {} (v4);
\path[<-,-latex]  (v6) edge  [bend right = -20] node[] {} (v5);
\path[ <-,-latex]  (v5) edge  [bend right = -20] node[] {} (v6);

\node at (4,4.5) {$\it EFT~Switch$};
\end{tikzpicture}
    \caption{Data Flow in EFT System }
    \label{EFTSFlow}
    \end{figure}
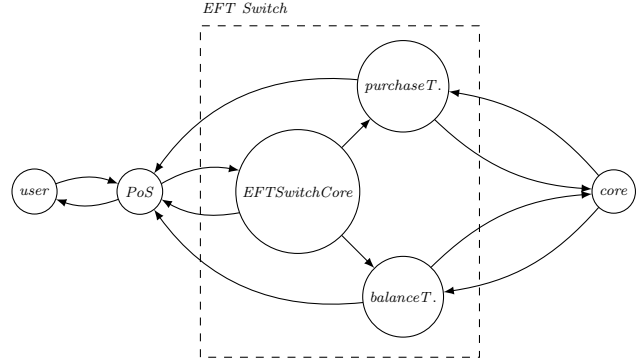

The implemented system contains two transactions of
\textit{balance request}  and \textit{purchase request}. The EFT switch itself contains three components, the $\it{EFTSwitchCore}$ component receives a message from $\it PoS$ and %authenticates the $\it user$  and if the $\it user$ is authenticated
depending on the message type redirects it  to either components of $\it balanceTransaction$  or $\it purchaseTransaction~$ (${\it balanceT}$ and $\it purchaseT$ in short respectively).
Each of these components in turn sends the related message to the core banking system (in short $\it core$)
and eventually the related response is sent back to $\it user$. Full description of the components of the system modeled by AML  is given in Figures \ref{fig:EFTOpenSystem1} and \ref{fig:EFTOpenSystem2} of  Appendix \ref{appendix:actorsOfEFTS} .

Each component of this system could be subject to changes and
the system should still satisfy its requirements despite these changes.
Here we considered the component $\it purchaseTransaction$ as the \textit{unspecified component} $M$.
The set of $\Info$ conceived for this system is:
\begin{align*}
 \Info= \{&(  \mathit{start},\{(\langle  \mathit{purchaseRequest}\rangle, \mathit{core})\}) ,
 \\
& (  \mathit{purchaseSuccessfull},\{(\langle  \mathit{purchaseSuccessfull}\rangle, \mathit{PoS})\}) ,\\
&  ( \mathit{insufficientCredit},\{(\langle  \mathit{insufficientCredit}\rangle, \mathit{PoS})\}),\\
  & ( \mathit{cancelPurchase},\{(\langle  \mathit{cancelPurchase}\rangle, \mathit{core})\}),\\
 &   ( \mathit{cancelPurchase},\{(\langle  \mathit{purchaseCanceled}\rangle,PoS)\}),\\
&   ( \mathit{purchaseCanceled},\{(\langle  \mathit{purchaseCanceled}\rangle, \mathit{PoS})\}) \}
   \end{align*}The specification of the actor $A_{\IntM}$ based on this $\Info$ is given in Appendix \ref{appendix:actorsOfEFTS}, Figure \ref{fig:EFTSIntM}.

The required property of the system is that upon a purchase request and consequent canceling of the request if the authentication of the user has failed, the cancel request should not be sent to the $\it core$ and only $\it PoS$ should be notified of the cancellation. This Property and its corresponding $P_{\mathit{err}}$
are specified in %Appendix \ref{appendix:actorsOfEFTS},
Figure \ref{fig:SwitchDFAProperty}.

\begin{figure}[htbp]
\centering
    \begin {tikzpicture}

      \tikzstyle{bordered} = [draw,outer sep=0,inner sep=1,minimum size=10]

\node [style=circle,draw,outer sep=0,inner sep=1,minimum size=10] (v1)at (-2.5,2) {};
\node [style=circle,draw,outer sep=0,inner sep=1,minimum size=10] (v2)at (-4,2) {};
\node [style=circle,draw,outer sep=0,inner sep=1,minimum size=10] (v3) at (-1,2) {$\pi$};
\node (v0) at(-7.5,2) {};
\node [style=circle,draw,outer sep=0,inner sep=1,minimum size=10] (v4)at (-5.5,2) {};
\node [style=circle,draw,outer sep=0,inner sep=1,minimum size=10] (v5)at (-7,2) {};

\path[<-,-latex]  (v2) edge  [bend right = -15] node[above =0.03 cm] {$\gamma$ } (v1);
\path[<-,-latex]  (v4) edge  [bend right = -15] node[above =0.03 cm] {$\beta$ } (v2);
\path[<-,-latex]  (v5) edge  [bend right = -15] node[above =0.03 cm] {$\alpha$ } (v4);
\path[<-,-latex]  (v1) edge  [bend right = -20] node[below =0.03 cm] {$\eta$ } (v5);
\path[->,-latex]  (v1) edge node[above =0.03 cm] {$\delta $} (v3);
\path[<-,-latex]  (v2) edge   [loop above] node {$\lnot \gamma$ } (v2);
\path[<-,-latex]  (v1) edge   [loop above] node {$\lnot(\delta \lor \eta)$} (v1);
\path[<-,-latex]  (v4) edge   [loop above] node {$\lnot \beta$} (v4);
\path[<-,-latex]  (v5) edge   [loop above] node {$\lnot\alpha$} (v5);
\path[<-,-latex]  (v0) edge (v5);

\node at (3,2)[bordered] { $\begin{array}{l}{\alpha =\it Snd(purchaseRequest)::PoS}\\
\beta ={\it Snd(cancelPurchase)::PoS}\\
\gamma ={\it Snd(authenticationError)::PoS}\\
\delta ={\it Snd(cancelPurchase)::core}\\
\eta ={\it  Snd(purchaseCanceled)::PoS}
\end{array}$};
\end{tikzpicture}
    \caption{$P_{\mathit{err}}$
for the Property
    $P =  \square\hspace{1mm} ((\alpha \wedge \bigcirc \beta \wedge \lozenge \gamma) \rightarrow (\lnot \delta \wedge \eta )) $}
    \label{fig:SwitchDFAProperty}
  \end{figure}
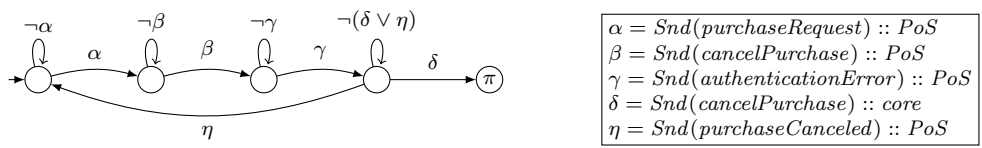

We applied our approach on this setting which terminated by generating an assumption which has $42$ states and $417$ transitions. We checked the actor $\it purchaseTransaction$
of Figure \ref{fig:purchaseTransactionActor}  against the generated assumption which was accepted.
%The actor $A_{\IntM}$ is designed based on this $\Info$ in Appendix \ref{appendix:actorsOfEFTS}, Figure \ref{fig:EFTSIntM}.

   \begin{figure}
    \begin{lstlisting}[language=palang, multicols=2]
    actor purchaseTransaction(2){

        int flag;// 1: Purchase Request is sent to Core
        start{
            flag=1;
            core!purchaseRequest;
        }
        purchaseSuccessful{
            PoS!purchaseSuccessful;
        }
        insufficientCredit{
            PoS!insufficientCredit;
        }
        cancelPurchase{
            if(flag==1){
                core!cancelPurchase;
                flag=0;}
            else{
                PoS!purchaseCanceled;}
        }

        purchaseCanceled{
            PoS!purchaseCanceled;
        }
    }
    \end{lstlisting}

    \caption{Actor
   $\mathit{  purchaseTransaction}$  }
    \label{fig:purchaseTransactionActor}
   \end{figure}

\subsection{Evaluation Results}\label{Subsec::tool}
To evaluate the effectiveness of our approach, we have compared it to the direct approach of \cite{giann} in terms of the state space of their generated property LTSs and assumptions for Quadricopter, EFT, and our running example systems. To derive the  property LTS of a given actor model and $P_{\it err}$, common in both approaches, we have exploited the algebraic mCRL2 toolset \cite{GrooteMousavi14}, based on \emph{ACP} process algebra~\cite{BergstraKlop84}. However, we have manually added the respective $A_\IntM$ to each actor model to derive the property LTS of our approach. We have also implemented the remaining steps of both algorithms in a Java program to generate the assumption from the derived LTSs. We remark that we had to tailor the remaining steps of the approach of \cite{giann} in the same way of our approach to be applicable to asynchronous setting. For asynchronous models,  to be implemented purely using \cite{giann}, each actor of $O$ should get specified by an LTS, generated by the technique explained in Section \ref{sec:generatingM}. This results in huge LTSs that should get composed together in the next phase. In order to skip some long primal steps (which may not be feasible under limited resources of time and space), we decided to generate the LTS of the composition of these actors at once and not to generate each separately. %However since in \cite{giann}'s work there is no interface replacing the missing component, therefore in order to generate the LTS of $O$'s actor 
To this aim, we still need to consider a wild environment for $O$. To skip this also, we decided to make the compositions of $O$ and $P_{\it err}$ so that $P_{\it err}$ would act as the missing component and as the result a closed LTS would be generated. %To sum up we made the property LTS of the first step by computing $c_1 \times  \ldots \times c_n \times P_{\it err}$ at once (same as our approach with the difference that we also have $ A_{\IntM}$). 
This is a considerable favor towards the direct approach %. Also Since $c_1 \times  \ldots \times c_n \times P_{\it err}$ is being generated at once, this implementation of direct approach is also 
benefiting from  one of the main feature of our approach which is not to generate states of the system when the $\pi$ state is reached. The rest of the steps for both approaches is based on the steps in section \ref{sec:gencond}.
Hence, only one Java program is needed to apply the remaining steps. Therefore, our comparison reveals the effectiveness of our interface actor in generating the assumption and making the algorithm terminating.

To derive the property LTS of a given actor model and $P_{\it err}$, AML specifications of the actors are transform into a set of processes\footnote{This encoding is available at \url{http://fghassemi.adhoc.ir/AGcode.zip}.}. Each actor is encoded as a recursive process parameterized with a queue as the mailbox 
and possible state variables, and formed by two sub-processes. One implements the operational semantics of {\it Take} in AML semantic rules: Upon taking a message from the head of the actor queue, the respective method body is selected in terms of the method name to be executed.  %.  The effect of each method is implemented by aligning the possible send actions resulted by execution of its 
Each method body is implemented by using the conditional and send action prefixes of mCRL2 while variable assignments are modeled by assigning updated values to them during the recursive invocation of the actor process. Send action are parameterized by the method names   which is unique for each actor. The other sub-process implements the input-enableness of each actor: Each actor can recursively receive any of its messages (identified by an operation which outputs the list of messages of an actor given its identifier)  by performing receive actions parameterized with its method names (which are unique). The corresponding message is then added to the queue of the actor process. The receive and send actions with the same parameter are synchronized to model receiving a sent message between two actors. In order to implement the atomic execution of each method, a scheduler process is defined which gives turn to the actors non-deterministically. $P_{\it err}$ of the given property is also encoded as a process. The special $\pi$ state is encoded as a $\pi$ action leading to the deadlock process. To derive the property LTS, the parallel composition of the actor processes, the scheduler, and $P_{\it err}$ are computed such that the processes are enforced to synchronize on send and receive actions of actors and actions of the $P_{\it err}$. Consequently whenever the process of $P_{\it err}$ reaches to its $\pi$ state (modeled by the deadlock process followed after the special $\pi$ action) the whole system goes to deadlock.

We have implemented the algorithm using a Java tool. Our java program, takes the name and location of the mCRL2 file (which has the specification of the mentioned processes) as inputs, implements the step one of algorithm by calling  mCRL2 commands (creating the property LTS as the result). The remaining steps of the algorithm are implemented in Java. %programming language. %by calling \textcolor{red}{the defined methods.}

The size of the property LTSs %(derived from the composition of the system with the property, i.e., $A_{\IntM}\interleave_P O$ in our approach) %and the product of $O$ and $P$
%in the direct approach)
and the generated assumptions resulting from both approaches are compared in Table~\ref{table:evaluation}.
In all cases, our approach immediately terminates and the property LTSs and the generated assumptions are smaller thanks to the overapproximated actor.
For instance, in the case of the Quadricopter system with a single initial message, our approach immediately terminates in step $1$, in contrast with the direct approach, as the overapproximated actor prevents from entering into $\pi$ state. The direct approach was unsuccessful in computing the property LTS of EFT after $1$ hour on a system with Corei$7$ CPU and $8G$ RAM configurations. Considering the effect of a given property by promoting the semantic rules of AML, alleviates the state space explosion problem and hence makes our approach efficient and terminating.

To inspect a given specified actor of $M$ against the generated assumption, we have manually check its compliance and then exploited mCRL2 to generate $\rho_{c_3}(T_M)$. To this end, $A_M$ was specified as a process with two subprocesses as before. The subprocess which implements the  input-enableness of the actor, results in $A_M$ receiving any of its messages which implements the functionality of $A_{M^c}$ and %(as apposed to AML well-formedness property where the receiver of a message has to have a method with the same name as the message) 
there is no need to implement $o_i^d$'s as in mCRL2 we can implement a process such that  send actions  manifest in the resulting LTS without the need for the recipient of messages to exist. We used CADP\footnote{\url{http://cadp.inria.fr/}} to compare the traces of $\rho_{c_3}(T_M)$ and $T_A$ due to its support for systems with a huge number of states and transitions.

To validate our results, we verified the system containing open system actors and the given $M$ against the given property. To this aim, the composition of the closed system $A_M\interleave O $ and $P_{\it err}$  was generated using the mCRL2 toolset again (similar to the above). The property is satisfied by the system if the $\pi$ state is not reachable in the resulting LTS. As we have implemented the $\pi$ state with a $\pi$ action terminating in a deadlock state, the $\pi$ state is reachable if the action $\pi$ exists and is reachable from the initial state in the LTS $A_M\interleave O $ and $P_{\it err}$. This is checked by the mCRL2 model checker tool.%a backward analysis from the $\pi$ state to the initial state.

\begin{table}[]
\centering
\caption{Evaluation results}
\label{table:evaluation}
\renewcommand{\arraystretch}{1.2}
\begin{adjustbox}{max width=\textwidth}
\begin{tabular}{|c|l|c|c|c|c|c|}
\hline
\multicolumn{1}{|l|}{\multirow{2}{*}{Implemented Approach}}               & \multirow{2}{*}{System under Analysis} & \multicolumn{3}{c|}{Property LTS} & \multicolumn{2}{c|}{Generated Assumption} \\ \cline{3-7}
\multicolumn{1}{|l|}{}                                                    &                                        & States            & Transitions   &Time(s)       & States            & Transitions           \\ \hline
\multirow{4}{*}{\begin{tabular}[c]{@{}c@{}}Asynchronous\\ AG\end{tabular}} & EFTS                                   & 127               & 186    &43             & 42                & 417                   \\ \cline{2-7}
                                                                          & Quadri. - Single Init                  & 10                & 10        &1           & 0                 & 0                     \\ \cline{2-7}
                                                                          & Quadri. - Multiple Init                & 182               & 343   &17               & 8                 & 22                    \\ \cline{2-7}
                                                                          & Mutex                                  & 15                & 21          &0         & 14                & 68                    \\ \hline
\multirow{4}{*}{\begin{tabular}[c]{@{}c@{}}Tailored Direct\\ AG\end{tabular}}      & EFTS                                   & 224328$>$           & 823293$>$      &3600       & -                 & -                     \\ \cline{2-7}
                                                                          & Quadri. - Single Init                  & 1510              & 4830      &135           & 20                & 59                    \\ \cline{2-7}
                                                                          & Quadri. - Multiple Init                & 688               & 2001     &96            & 9                & 25                    \\ \cline{2-7}
                                                                          & Mutex                                  & 22                & 57       &1           & 21                & 85                    \\ \hline
\end{tabular}
\end{adjustbox}

\end{table}

\subsubsection{Exploiting a More Precise $\Info$}\label{sec:preciseIntM}
If $\Info$ is more accurate, in the sense of being stateful, our approach - with minor changes due to extending $\Info$ and generating $\IntM$-
still works. However the set of possible $M$s which comply to 
$\Info$ will become more limited. We have conducted a set of experiments on our running example to illustrate how using a more accurate $\Info$ affects the performance of our assumption generation algorithm. As a more precise $\Info$ leads to a more precise $\IntM$, we have directly used a more precise $\IntM$ in our experiments to observe its impact on the resulting assumption. 
Such precise specifications can be derived by adding the variable $\it taken$ to $\IntM$, specified by Figure \ref{fig:IntMOfexampleOpenSystem}, and defining its behavior based on this variable while processing $\it reqL$, $\it reqR$, and $\it release$. Conversely, they can be characterized based on the lines removed from $M$, specified by Figure \ref{fig:actorM}. All resulting specifications definitely comply with $\Info$. The experiment results, shown in the Table \ref{table:eval2}, indicate that a more precise $\IntM$ may not affect the performance as extra precision may result in larger property LTSs and assumptions for different values of a variable. However, all experiments result smaller LTSs compared to the direct approach. For the cases that the $\pi$ state is not reachable in their property LTSs, the algorithm converges faster and it terminates immediately.

\begin{table}[htp]
\centering
\caption{Effect of a more precise $\IntM$ for Mutex system}
\label{table:eval2}
\begin{tabular}{|l|c|c|c|c|c|c|}
	\hline
	\multicolumn{1}{|c|}{\multirow{2}{*}{\begin{tabular}[c]{@{}c@{}}Lines Omitted from\\  Figure \ref{fig:actorM}\end{tabular}}} & \multicolumn{3}{c|}{Property LTS} & \multicolumn{3}{c|}{Generated Assumption}              \\ \cline{2-7} 
	\multicolumn{1}{|c|}{}                                                                                                        & States  & Transitions  & Time     & States & Transitions & Time                            \\ \hline
	\begin{tabular}[c]{@{}l@{}}3,5,7-10,14,16,\\ 17-19,23:($\Info$)\end{tabular}                                                  & 15      & 21           & 3518576  & 14     & 68          & \multicolumn{1}{l|}{3087559909} \\ \hline
	7,16,23                                                                                                                       & 15      & 21           & 3546567  & 14     & 68          & \multicolumn{1}{c|}{3081696082} \\ \hline
	16,23                                                                                                                         & 16      & 21           & 3565696  & 15     & 72          & \multicolumn{1}{c|}{3079037789} \\ \hline
	23                                                                                                                            & 14      & 18           & 3527906  & -      & -           & 3527906                         \\ \hline
	- :($M$)                                                                                                                      & 13      & 18           & 3407075  & -      & -           & 3407075                         \\ \hline
	7,23                                                                                                                          & 16      & 21           & 3561963  & 15     & 72          & 3073369906                      \\ \hline
	7                                                                                                                             & 18      & 24           & 3617947  & 11     & 53          & 3079523449                      \\ \hline
	7,16                                                                                                                          & 15      & 21           & 3551699  & 14     & 68          & 3081989530                      \\ \hline
	16                                                                                                                            & 18      & 24           & 3989771  & 11     & 53          & 3082167745                      \\ \hline
\end{tabular}
\end{table}

\section{Related Work}\label{sec:releatedWork}
Compositional verification has been proposed as a way to address the state space explosion problem in model checking~\cite{clarke1999model}. It uses divide-and-conquer approach to verify a system against a property by verifying each of its components individually against a local property~\cite{clarke1989compositional,grumberg1994model}, a technique that is often called quotienting. In the context of process algebra as the language for describing reactive systems and Henessy-Milner logic as the logical specification formalism for describing properties, decompositional reasoning has been originally studies in \cite{larsen1986context,larsen1991compositionality} and has been further developed in, e.g., \cite{bloom2004precongruence,fokkink2006compositionality,raclet2008residual}. A related line of research is the one devoted to  compositional proof systems, e.g., \cite{AndersenSW94} which determines satisfaction between algebraic processes and assertions of the modal $\mu$-calculus. %, can be applied to our problem. 
This technique provides a set of rules guided by the structure of the processes. The idea of compositional verification can be formulated as finding a proof for a term with a set of meta-variables, distinct from recursive variables. These meta-variables denote undefined processes. Thus, proofs carried out with such variables appearing in the terms, are valid for all instantiations of the variable. Such a parameterized proof results in an assertion that should be satisfied by the meta-variable. This assertion expresses possible instantiations for the meta-variable and hence, plays the role of an assumption in our case. A given $M$ guarantees the whole system satisfies $P$ if a proof can be found for the assertion. Although this approach is also applicable to non-safety properties, finding such proofs and assertions cannot be automated and needs an exhaustive user involvement.

To successfully verify an isolated component, it is often needed to have some knowledge about the environment interacting with the component, and to this end  assume-guarantee reasoning approaches have been proposed to address the issue~\cite{jones1983specification,pnueli1985transition}. These approaches guarantee a
system
$S$ composed of components $M_1$ and $M_2$ satisfies a property $P$  if
$M_1$ under an assumption $A$ satisfies $P$ and also $M_2$ satisfies $A$ in the context of all possible environments. As we explained earlier, our problem is reduced to the assume-guarantee reasoning problem %. We have adapted the assume-guarantee reasoning for the verification of asynchronous distributed systems
where
$M_2$ maps to the unspecified component and $M_1$ plays the role of  the open system.
To the best of  our knowledge the assume-guarantee reasoning has not been specifically associated with  asynchronous environments.

In recent work such as \cite{alur2005symbolic,nam2008automatic,puasuareanu2008learning,chaki2007optimized}, the assumption has been generated through learning-based approaches by using the $L^*$ algorithm, first proposed by Angulin \cite {angluin1987learning}. Given a set of positive and negative traces, an initial assumption is generated which will be iteratively refined by posing questions to an oracle.
%The algorithm in \cite{gupta2007automated} uses sampling rather than  $L^*$  to learn the assumption  in a similar way.
These approaches are not time-efficient due to additional overhead to learn  assumptions~\cite{cobleigh2008breaking}.
Also the generated assumption may become very large and also the verification may fail in some cases
\cite{nam2006learning,alur2005symbolic}.
These approaches are useful when the computations for the generation of the weakest assumption run out of memory. The other fact that should be considered is that
in the $L^*$-based approaches in order to produce the assumption,
the component $M_2$ is incorporated in the iterative solution which is in contrast to our original goal of reusing the produced assumption in different settings and specifically with different $M_2$ components. In other words,
despite the solution of \cite{alur2005symbolic,nam2008automatic,puasuareanu2008learning,chaki2007optimized}, our solution is independent from $M_2$ and there is no need to generate a new assumption for each different $M_2$ component.
In the future, we plan to investigate more on how to incorporate learning techniques into our asynchronous setting.

Counter example abstraction refinement technique has been associated with assume-guarantee reasoning in \cite{bobaru2008automated}
by iteratively computing assumptions as conservative abstractions of the interface behavior of $M_2$ which concerns  interactions with $M_1$. In each iteration, the composition of $M_1$ and $A$ is checked against $P$ and based on the achieved counterexample, $A$ is refined accordingly which makes the approach computationally expensive.
%Since in each iteration of the algorithm $M_2$ satisfies the computed assumption $A$ by construction, only the other premise of assume-guarantee reasoning -$M_1$ in context of $A$ should satisfy $P$ -  is checked.
Furthermore the aforementioned issue of reusability due to the
iterative incorporation of   $M_2$ in  computations is also observed in this solution. We also take advantage of an abstract interface of $M_2$ in our approach using interaction information of $M_1$,  %our abstract interface is small enough such that we can incorporate the direct approach mentioned in \cite{giann} to find the weakest assumption
however there is no need for iterative checkings.% of the mentioned premise.

The direct approach toward the generation of the weakest assumption
has been applied in work such as \cite{giann,cheung1999checking}. In
particular in
\cite{giann}, the weakest assumption for the environment of a component is generated such that the component satisfies a required property. But the generated assumption may become too large and the computation may run out of memory \cite{cobleigh2003learning}. % which this  raises the state spaces explosion problem \cite{clarke1999model}.
We have extended the work in \cite{giann} to asynchronous systems and also made some improvements on handling internal actions.  By assuming some information on the interactions of $O$, our approach produces an assumption which is
noticeably smaller than  the generated assumption in \cite{giann}.
%Our approach i more suitable for asynchronous setting. Using \cite{giann} in such a setting, we would have to model each component of the open system as an LTS, which may result in large LTSs and make the computations more difficult, time consuming or in some case impossible.

In \cite{feng2007relationship,gupta2007automated,grinchtein2006inferring,pena1999new}, the problem of  finding the  separating DFA of two disjoint regular languages, so called concurrent separation logic, has been mapped to the problem of assume-guarantee reasoning.
%The work in \cite{gupta2007automated,grinchtein2006inferring} uses sampling rather than $L^*$ and  aims at  finding a separating
To this aim, the separating automaton $A$ is generated which accepts the language of $M_2$ but rejects the language of $M_1'$, where $M_1'$ accepts the language of $M_1$ and does not satisfy $P$. The work in \cite{gupta2007automated,grinchtein2006inferring} exploits an iterative  method using boolean satisfiability solvers.
%However the described method is only suited for  hardware verification.
%It is also proposed to reduce the problem to finding a minimum separating DFA in \cite{grinchtein2006inferring,pena1999new},
Generally all of these approaches are computationally expensive while \cite{gupta2007automated} is only suited for  hardware verification.

The theory of  supervisory control and controller synthesis
also tackles a similar problem but in a very different context.
The goal is to synthesis a controller called the supervisor, such that
the behaviors of the system called plant, are controlled so that the requirements of the plant are satisfied~\cite{ramadge1987supervisory,cassandras2009introduction}.
The theory  mainly focuses on hardware systems, therefore assuming
that the controller reacts sufficiently fast on machine input,  the
supervisor and the plant are originally modeled as synchronizing
processes~\cite{ramadge1987supervisory,cassandras2009introduction}
and the supervisor allows  the occurrence of events of the plan by
synchronizing with them and prevents them by refusing to
synchronize. This also prevails in modern state-of-the art
approaches
\cite{heymann1998discrete,fabian1996non,overkamp1997supervisory}. %%
The notion of asynchrony in this context   is mostly  referred to as events occurring without any reference to a clock \cite{ramadge1984supervisory,cortadella1997petrify,cortadella2012logic} which is very different from the notion of asynchrony in our context.
%The earlier mentioned similarity  is that in both areas of research a specification in being produced.
The problem of controller synthesis is mapped to our context such
that the supervisor is mapped to the unspecified component, system
requirements play the role of the system properties and the plant is
the open system. However, the fundamental difference is that system
requirements in a controller synthesis are limited to the behaviors
of the plant and the generated supervisor describes a
``well-defined'' plant \cite{markovski2016process,van2015maximal}.
However in our case, the behaviors of the unspecified component are
mentioned in the property and the generated specification describes
the unspecified component. In other words  a synthesized controller
is a dependent entity controlling the behaviors of the plant,  but
in our context we synthesis a component which is a part of the
system and has its own  independent functionality.

\section{Conclusion and Future Work}\label{sec:futureWork&conclude}
We developed a framework for asynchronous systems to successfully verify safety properties of such systems in the presence of a modified or an added component. This is done by generating the weakest assumption for the component to be checked against. By applying this framework, we ignore model checking of the whole system with each modification made to a component of the system. We adapted the direct approach of the synchronous settings to make it terminating and more efficient in time and computations.
% while the weakest assumption is generated.
To this aim, we generate an over approximated component instead of the missing component to make the system of asynchronous components closed and used the special composition to build the property LTS to eliminate the error traces simultaneously. We have illustrated the applicability of  our work on two real world case studies and shown that the generated assumptions are significantly smaller. We have also relaxed the condition which restricts $M$ to the actions of $P$ by considering the internal actions of the unspecified component in our computations. Therefore, our approach eliminates false positive results as a consequence of such a condition (In \cite{giann} internal actions are led to the sink state after which observable actions may violate the property).
%Another difference is we use LTL formulas for the specification of properties which is  more applicable and  makes
%defining properties easier.
Such improvements are not limited to the asynchronous settings and therefore can be adopted.

%Using this framework,
%verification of the whole system upon
%modification in a component or adding a new component is ignored and only that specific component can be checked.
%given the reusability of components in CBSE, and the trend toward cyber physical systems this framework is applicable in these contexts.

As future work we focus on improving the third step of our algorithm: instead of completing the generated LTS and exploiting conformance checking to prevent false positive, we find conditions to steer the completion phase. Another improvement is to imply some reductions on the generated LTS of a single actor (regarding its interactions) in order to reduce the time and resources for checking an actor against the generated assumption. Furthermore, providing an approach to validate a rejected $M$ by our compliance checking is another research direction. Extending %e are going to %finish implementing  our framework as a tool. Also we aim to
our framework with data can also be considered. %and make the compliance checking automated. 
%Our work can also take another turn in terms of time consumption and ease of use by introducing equivalence relation between actor components.
We also plan to investigate on how to incorporate grammar inference techniques into our asynchronous setting.

\section*{Acknowledgements}
We would like to
thank Mohammad Reza Mousavi for his helpful discussion on the paper.

%\vspace{5mm} \noindent \textbf{Acknowledgments} We would like to
%thank MohammadReza Mousavi for his advice and comments.

\bibliographystyle{spmpsci}
\bibliography{rosa}

\begin{appendices}

\section{Actors of the EFT System}\label{appendix:actorsOfEFTS}
   \begin{figure}[H]
        \centering
    \begin{lstlisting}[language=palang, multicols=2]
  actor PoS(2){

    int isCanceled;
    int flag1;
    //1: "Insufficient Credit" has
    // happened while "cancelPurchase"
    // has been sent
    int flag2;
    //1: "Authentication error" has
    // happened while "cancelPurchase"
    // has been int sent
    int flag3;
    //1: "cancelPurchase" has been sent

    insertCard{
        user!insertPassword;
    }
    passwordIs{
        user!chooseTransaction;
    }
    balanceRequest{
        eftSwitchCore!balanceRequest;
    }
    purchaseRequest{
        eftSwitchCore!purchaseRequest;
    }
    cancelPurchase{
        isCanceled=1;
        eftSwitchCore!cancelPurchase;
    }

    balanceResponse{
        user!balanceResponse;
    }
    purchaseSuccessful{
        if(isCanceled==1){
            flag3=1;}
        else{
            user!purchaseSuccessful;}
        }
    authenticationError{
        if(isCanceled==1){
            flag2=1;}
        else{
            user!authenticationError;
            isCanceled=0;}
    }

    insufficientCredit{
        if(isCanceled==1){
            flag1=1;}
        else{
            user!insufficientCredit;
            isCanceled=0;}
    }

    purchaseCanceled{
        if(flag1==1){
            user!insufficientCredit;
            flag1=0;}
        else{ if(flag2==1){
            user!authenticationError;
            flag2=0;}
        else {if(flag3==1){
            user!purchaseCanceled;
            flag3=0;}
        else{}}}
        isCanceled=0;
    }
  }
    \end{lstlisting}

    \caption{Actors of the Open System of EFT System}
    \label{fig:EFTOpenSystem1}
   \end{figure}
    \begin{figure}
        \centering
        \small
    \begin{lstlisting}[language=palang, multicols=2]
    actor user(2){

    int scenario;
        start{
            PoS!insertCard;
        }
        insertPassword{
            PoS!passwordIs;
        }
        chooseTransaction{
            scenario=?(0,1,2);
            if(scenario==0){
           	   PoS!balanceRequest;}
            else{ if(scenario==1){
           	   PoS!purchaseRequest;}
            else{ if(scenario==2){
           	   PoS!purchaseRequest;
           	   PoS!cancelPurchase;}
            else{}}}
           cenario=0;
        }
        balanceResponse{
            self!start;
        }
        purchaseSuccessful{
            self!start;
        }
        authenticationError{
            self!start;
        }
        insufficientCredit{
            self!start;
        }
        purchaseCanceled{
            self!start;
        }
    }

    actor EFTSwitchCore(2){

            int wrongPassword;
            balanceRequest{
                wrongPassword=?(0,1);
                if(wrongPassword==0){
                    balanceTransaction!start;}
                else{
                    PoS!authenticationError;}
                wrongPassword=0;
            }
            purchaseRequest{
                wrongPassword=?(0,1);
                if(wrongPassword==1){
                    purchaseTransaction!start;}
                else{
                    PoS!authenticationError;}
                wrongPassword=0;
            }
            cancelPurchase{
                purchaseTransaction!cancelPurchase;
            }
        }

        actor balanceTransaction(2){

            start{
                core!balanceRequest;
            }
            balanceResponse{
                PoS!balanceResponse;
            }
        }

        actor core(2){

            int noCredit;
            balanceRequest{
                balanceTransaction!balanceResponse;
            }
            purchaseRequest{
                noCredit=?(0,1);
                if(noCredit==0){
                    purchaseTransaction!purchaseSuccessful;}
                else{
                    purchaseTransaction!insufficientCredit;}
                noCredit=0;
            }
            cancelPurchase{
                purchaseTransaction!purchaseCanceled;
            }
        }

        main {
            user!start;
        }
    \end{lstlisting}
        \caption{Actors of the Open System of EFT System}
        \label{fig:EFTOpenSystem2}
    \end{figure}

\begin{figure}
\centering
    \begin{lstlisting}[language=palang, multicols=2]
actor purchaseTransaction(2){

    int l;
    start{
        core!purchaseRequest;
    }
    purchaseSuccessful{
         PoS!purchaseSuccessful;
    }
    insufficientCredit{
        PoS!insufficientCredit;
    }
    cancelPurchase{
        int l=?(0,1);
        if(l==1){
            core!cancelPurchase;
        }
        else{
            PoS!purchaseCanceled;
        }
    }
    purchaseCanceled{
        PoS!purchaseCanceled;
    }
}

\end{lstlisting}

\caption{Actor $A_{\IntM}$ of  EFT System}
    \label{fig:EFTSIntM}
\end{figure}

\end{appendices}

\end{document}